\documentclass[journal]{IEEEtran}
\usepackage{amsmath,amssymb,amsfonts}
\usepackage{array}
\usepackage{mathtools}
\usepackage{graphicx}
\usepackage{textcomp}
\usepackage{xcolor}
\usepackage{bm}
\usepackage[hidelinks]{hyperref}
	\hypersetup{pdftitle={Scalable Policy Optimization for Networked Multi-Agent Reinforcement Learning with Continuous State-Action Spaces},pdfauthor={Dongming Wang, Pengcheng Dai, Wenwu Yu, Wei Ren}}
\usepackage{cite}
\usepackage{amsthm}
\usepackage[theorems,algorithms]{textool}
\definecolor{dm}{rgb}{0.80,0.10,0.10}

\theoremstyle{plain}
\newtheorem{theorem}{Theorem}
\newtheorem{lemma}{Lemma}
\newtheorem{proposition}{Proposition}
\newtheorem{corollary}{Corollary}

\providecommand{\hhatphi}{{\ensuremath{\hat{\boldsymbol{\phi}}}}}
\providecommand{\hhatmu}{{\ensuremath{\hat{\boldsymbol{\mu}}}}}
\providecommand{\Cfloor}{{\ensuremath{C_{\rm floor}}}}
\providecommand{\Efloor}{{\ensuremath{\epsilon_{\rm floor}}}}
\providecommand{\bbarr}{{\ensuremath{\bar{r}}}}
\providecommand{\bbarg}{{\ensuremath{\bar{g}}}}
\providecommand{\tdbbphi}{{\ensuremath{\tilde{\boldsymbol{\phi}}}}}
\providecommand{\Bk}{{\ensuremath{B_\kappa(m)}}}

\DeclareMathOperator{\dist}{dist}
\DeclareMathOperator{\sym}{sym}

\begin{document}

	\title{Scalable Policy Optimization for Networked Multi-Agent Reinforcement Learning with Continuous State-Action Spaces}

	\author{Dongming Wang, Pengcheng Dai, Wenwu Yu,~\IEEEmembership{Senior Member,~IEEE}, Wei Ren,~\IEEEmembership{Fellow,~IEEE}
		\thanks{D. Wang and W. Ren are with the Department of Electrical and Computer Engineering, University of California, Riverside, CA 92521, USA.
			{\tt\small \{wdong025, wei.ren\}@ucr.edu}.}
		\thanks{P. Dai is with the Engineering Systems and Design Pillar, Singapore University of Technology and Design, Singapore 487372.
			{\tt\small Jldaipc@163.com}.}
		\thanks{W. Yu is with the Frontiers Science Center for Mobile Information Communication and Security, School of Mathematics, Southeast University, Nanjing 211102, China, and also with Purple Mountain Laboratories, Nanjing, 211102, China.
			{\tt\small wwyu@seu.edu.cn}.}}

	\maketitle

	\begin{abstract}
		Learning local policies for continuous networked systems requires accounting for the effects of decisions beyond each agent's observation neighborhood. Spatial decay limits these effects, but a finite critic must also control representation and estimation errors throughout policy optimization. We analyze the Continuous Distributed Coupled Policy Gradient (CDCPG) algorithm using local random Fourier features and least-squares temporal-difference critics. For features that retain the boundary inputs required by the local dynamics, we derive an action-value representation with separate spatial and finite-feature residuals. A global integrated transition-approximation bound and a projected Bellman argument control population prediction error without an inverse-conditioning multiplier. We then quantify the dependence of critic estimation on feature excitation and dimension, and construct simultaneous lower confidence bounds for temporal-difference conditioning along the executed iterates. Combining critic error with localized reward aggregation bounds the expected squared projected-gradient mapping by an optimization term and an explicit residual separating spatial approximation, finite features, and omitted distant rewards. For fixed neighborhoods and feature dimension, the shared-oracle sample count is inverse-squared in the excess squared-stationarity accuracy, up to logarithmic factors. The guarantee assumes known local dynamics and rewards, independent discounted-occupancy samples, and stated excitation, decay, and smoothness conditions, and is conditional on favorable feature draws. Numerical studies illustrate related implementations on a linear-coupled-quadratic benchmark.
	\end{abstract}


	\section{Introduction}\label{sec:intro}

	Cooperative multi-agent reinforcement learning (MARL) provides a framework for coordinating decision makers whose policies affect a shared objective~\cite{zhang2021multi}. Local information is particularly important in networked control, as illustrated by decentralized voltage regulation~\cite{feng2023stability} and distributed linear-quadratic control~\cite{shin2023near}. Such settings motivate learning methods whose information requirements remain confined to prescribed neighborhoods as the network grows. The central question is how these local decisions can be improved while accounting for their contribution to the performance of the entire network.

	We study this question for networked Markov decision processes (MDPs) with continuous state and action spaces. Each agent observes a prescribed neighborhood, and the team seeks to maximize the sum of its discounted local rewards. Even with local dynamics, the learning problem remains coupled across agents because the consequences of an action propagate through the network over time. Its contribution to the team objective therefore involves rewards beyond the acting agent's immediate neighborhood. Continuous variables create a second difficulty. The joint continuous dimension grows linearly with the number of agents $n$, and a fixed-resolution Cartesian discretization grows exponentially, so a finite critic must approximate the relevant value function, and the resulting statistical error must remain controlled as the policy changes.

	\subsection{Related Work}

	\emph{Locality in networked MARL.} Localized actor--critic analyses exploit decay of distant-variable influence to control local value-approximation bias. Qu et al.~\cite{qu2020scalable,qu2022scalable} developed scalable actor--critic methods based on the \emph{exponential decay property} of networked MDPs, under which the influence of distant agents on each local value decays geometrically in graph distance, with truncated local critics for finite state-action spaces. Lin et al.~\cite{lin2021multi} extended the locality framework to stochastic network dependencies using generalized decay and trajectory-based analysis, and Zhang et al.~\cite{zhang2023global} analyzed localized policy iteration for entropy-regularized networked MARL, with a locality-dependent gap to global performance. Chakraborty et al.~\cite{chakraborty2026locality} studied policy-dependent locality in the average-reward setting by separating environmental and policy sensitivities; in the present analysis the required decay properties are stated explicitly as premises, and a baseline decay rate is derived from finite propagation (Lemma~\ref{lem:baseline_decay}). These finite-space constructions do not directly supply a finite representation and a policy-evaluation error bound for the continuous model considered here.

	\emph{Communication-based decentralized learning.} A related literature distributes learning through a communication network. Fully decentralized actor--critic methods with globally observed states, linear value approximation, and consensus over a communication graph were analyzed asymptotically by Zhang et al.~\cite{zhang2018fully}; Doan et al.~\cite{doan2019finite} gave a finite-time analysis of distributed temporal-difference learning, TD($0$), for policy evaluation with linear function approximation under independent stationary samples; and Chen et al.~\cite{chen2022sample} developed finite-time decentralized actor--critic and natural actor--critic analyses under Markovian sampling. Chu et al.~\cite{chu2020multi} investigated communication and spatial discounting for networked-system control empirically. Communication connectivity and physical locality play different roles: exchanging estimates among neighbors does not, by itself, quantify the error caused by excluding distant physical variables from a critic.

	\emph{Continuous networked control and spectral representations.} For structured continuous network control, Shin et al.~\cite{shin2023near} quantify the performance of spatially truncated distributed linear-quadratic controllers, and localized least-squares temporal-difference evaluation of $Q$-functions (LSTDQ) with gradient policy learning provides a related approach in the linear-quadratic setting~\cite{olsson2024scalable}. Beyond linear-quadratic models, spectral representations derive features from the transition dynamics. Random Fourier features (RFF) approximate positive-definite shift-invariant kernels by finite-dimensional trigonometric features~\cite{rahimi2007random}; further constructions and approximation results are surveyed in~\cite{liu2021random}. Spectral dynamic embeddings extend this idea to stochastic nonlinear control~\cite{ren2023spectral}. The closest predecessor, Ren et al.~\cite{ren2024scalable}, combines local spectral representations, evaluated on the boundary neighborhood required by the local dynamics, with policy optimization for continuous networked MDPs, and expresses its policy-evaluation bound through the inverse norms of the empirical least-squares temporal-difference (LSTD) matrices under stationary sampling. We analyze this type of representation under discounted-occupancy sampling, separate population prediction error from finite-sample coefficient estimation, and relate the latter to population feature excitation, its dimension dependence, and simultaneous sample-based conditioning bounds.

	\emph{Policy evaluation and optimization tools.} Our critic uses least-squares temporal-difference policy evaluation, introduced by Bradtke and Barto~\cite{bradtke1996linear} and developed further in~\cite{boyan2002technical}; Lazaric et al.~\cite{lazaric2012finite} provide a finite-sample analysis of least-squares policy iteration and its evaluation step, and finite-time linear TD analyses distinguish coefficient-error bounds, which depend on feature conditioning, from averaged prediction bounds, which need not~\cite{bhandari2018finite}; the present contribution is the corresponding separation for the local spectral representation under discounted occupancy. Population prediction error is analyzed through a projected Bellman argument following the linear-TD viewpoint of Tsitsiklis and Van Roy~\cite{tsitsiklis1997analysis}; the contraction needed here is established under discounted-occupancy sampling rather than a stationary distribution. Our actor update uses the policy-gradient identity~\cite{sutton1999policy}, and the critic and localization bounds are combined through projected stochastic-ascent analysis following the composite-optimization framework of~\cite{ghadimi2016mini}. Policy-gradient guarantees depend on the parameterization and approximation conditions~\cite{agarwal2021theory}; our target is projected stationarity within the specified local policy class. The inverse-squared dependence on the excess squared-stationarity accuracy has the same exponent as the bounded-variance stochastic first-order benchmark of~\cite{arjevani2023lower}, but the different oracle, projection, estimator bias, and fixed residual preclude a minimax-optimality claim.

	\emph{Empirical multi-agent reinforcement learning.} MADDPG uses centralized training with decentralized execution~\cite{lowe2017multi}, COMA uses a centralized critic with counterfactual baselines for credit assignment~\cite{foerster2018counterfactual}, Iqbal and Sha use attention to aggregate the information relevant to each agent in a centralized critic~\cite{iqbal2019actor}, and Yu et al.\ show that appropriately configured proximal policy optimization (PPO) can perform competitively on cooperative benchmarks~\cite{yu2022surprising}. These works provide empirical context. Our numerical comparisons include adaptations of the decentralized actor--critic method of Zhang et al.~\cite{zhang2018fully} and the attention-critic method of Iqbal and Sha cited above.

	\subsection{Contributions}

	The contributions of this paper connect the local representation used by the critic to the policy-optimization guarantee through three steps.
	\begin{enumerate}
		\item \emph{From local dynamics to population prediction error.}
		The features are evaluated on the neighboring inputs that generate the true local successor state, and the signed random-feature approximant requires an integrated error bound. We establish a global integrated transition-approximation bound for the random-phase construction and derive a local action-value representation with separate spatial and finite-feature residuals (Theorems~\ref{thm:rff_approx} and~\ref{thm:q_linear}). A projected Bellman argument then bounds population prediction error without an inverse-conditioning multiplier (Lemma~\ref{lem:pop_L2}).

		\item \emph{From feature excitation to finite-sample critic estimation.}
		Estimating the critic coefficients additionally requires control of the TD matrix. We show that, under the feature normalization, the TD conditioning constant cannot be dimension-free, derive a sufficient excitation condition for stability, and construct simultaneous lower confidence bounds along executed iterates (Proposition~\ref{prop:mu_M_obstruction}, Lemma~\ref{lem:excitation}, and Proposition~\ref{prop:diagnostic}). These results identify the conditioning dependence of the estimation bound; the diagnostic assesses the sampled run and does not enforce excitation.

		\item \emph{From critic error to localized policy improvement.}
		Combining the critic bounds with localized reward aggregation bounds the expected averaged squared projected-gradient mapping by an optimization term and an explicit residual separating spatial approximation, finite features, and omitted distant rewards (Theorems~\ref{thm:main_convergence} and~\ref{thm:sample_complexity_fixed}). For fixed neighborhoods and feature dimension, the shared-oracle sample count is inverse-squared in the excess squared-stationarity accuracy, up to logarithmic factors. The decomposition also provides an offline neighborhood-selection rule under stronger graph-tail conditions (Proposition~\ref{prop:adaptive}); a predictable-prefix argument gives the corresponding trajectory-stability result (Theorem~\ref{thm:sample_complexity_traj}). A baseline exponential decay of the local action-value functions is derived from finite propagation (Lemma~\ref{lem:baseline_decay}).
	\end{enumerate}

	\emph{Scope.}
	The analysis uses known local drift and reward functions and independent shared samples from the current discounted occupancy. Its excitation, value-decay, and smoothness conditions remain explicit premises; the diagnostic supplies statistical evidence at sampled iterates rather than verifying those premises for an entire model class. The main result bounds expected projected stationarity within the prescribed local policy class, conditional on favorable feature draws. The linear-coupled-quadratic studies illustrate related implementations; their return and reference-discrepancy metrics are distinct from the stationarity quantity of the theorems, and they do not certify the assumptions.

	Table~\ref{tab:comparison} compares the representations, sampling models, and evaluation assumptions of the closest methods.

	\begin{table*}[!t]
		\centering
		\caption{Settings of the closest scalable networked MARL analyses. ``Cont.'' indicates continuous state-action spaces; the last column names the premise under which policy evaluation is analyzed. Rates are not tabulated because targets, sampling models, sample units, and fixed parameters differ across rows.}
		\label{tab:comparison}
		\renewcommand{\arraystretch}{1.2}
		\begin{tabular}{@{}lcp{3.1cm}p{3.4cm}p{4.9cm}@{}}
			\hline
			\textbf{Method} & \textbf{Cont.} & \textbf{Critic representation} & \textbf{Sampling model} & \textbf{Evaluation premise} \\
			\hline
			Qu et al.~\cite{qu2020scalable,qu2022scalable} & $\times$ & truncated local critic & Markovian trajectory & exploration and decay assumptions \\
			Lin et al.~\cite{lin2021multi} & $\times$ & truncated local critic & Markovian trajectory, mixing & generalized decay assumptions \\
			Ren et al.~\cite{ren2024scalable} & \checkmark & local spectral features & shared i.i.d.\ samples; stationary state-action law & empirical LSTD inverse norms enter the bound~\cite[Lemma~6]{ren2024scalable} \\
			\textbf{CDCPG (this work)} & \checkmark & augmented RFF on raw local inputs & shared discounted-occupancy oracle & $\sigma_{\min}(\bbM_i^{\bbtheta})\geq\mu_M(m)$ (Set~C) or along iterates (Set~C$'$), with diagnostic \\
			\hline
		\end{tabular}
	\end{table*}

	The remainder of this paper is organized as follows. Section~\ref{sec:pre} fixes notation, the networked MDP model and objective, and Assumption Sets~A, B, and D. Section~\ref{sec:spectral} develops the spectral framework and the RFF approximation. Section~\ref{sec:main} establishes the decay lemmas and the direct linear representation of the local-reward $Q$-functions. Section~\ref{sec:algr} states the TD-stability Assumption Sets~C and C$'$ once the population TD matrix has been defined, analyzes population and finite-sample LSTD error, and presents the algorithm. Section~\ref{sec:convergence} states the convergence and complexity theorems, the diagnostic, and the locality-selection rule. Section~\ref{sec:simu} reports numerical experiments. Section~\ref{sec:conclu} concludes. Proofs are given inline or in the appendices.

	\section{Preliminaries}\label{sec:pre}

	\subsection{Notation}

	Vectors are bold lowercase, matrices bold uppercase, $\bbI_d$ the identity, $\norm{\cdot}$ the Euclidean vector norm and the induced matrix norm; $\sigma_{\min}(\bbA)$ is the smallest singular value, $\norm{\bbA}_\ast$ the nuclear norm (sum of singular values), $\lambda_{\min}(\bbA),\lambda_{\max}(\bbA)$ the extreme eigenvalues of symmetric $\bbA$, and $\sym(\bbA):=(\bbA+\bbA^\top)/2$. The network is an undirected graph $\ccalG=([n],\ccalE)$ whose maximum degree is at most $\Delta$, with $\Delta\geq 2$ a fixed upper bound (for a disconnected graph we read $\rho^{d(i,j)}:=0$ when $d(i,j)=\infty$, with $\rho$ the decay base of (B1)); agent $i$'s $\kappa$-hop neighborhood is $\ccalN_i^\kappa$ ($\ccalN_i^0=\{i\}$, $\ccalN_i:=\ccalN_i^1$), $D_\kappa:=\max_i\abs{\ccalN_i^\kappa}$, and $d(i,j)$ is graph distance. Local state spaces $\ccalS_i=\mbR^{d_S}$; local action spaces $\ccalA_i\subseteq\mbR^{d_A}$ are compact; the joint spaces are $\ccalS:=\prod_i\ccalS_i$ and $\ccalA:=\prod_i\ccalA_i$. Discount $\gamma\in(0,1)$. The notation $\widetilde\ccalO(\cdot)$ hides polylogarithmic factors in its argument and in the problem primitives. Three locality radii are fixed throughout: $\kappa_\pi\in\mbN_0$ is the \emph{policy observation radius} (each local policy observes $\bbs_{\ccalN_i^{\kappa_\pi}}$; Assumption (A3) below), $\kappa\in\mbN$ with $\kappa\geq\max\{1,\kappa_\pi+1\}$ is the \emph{critic truncation radius} (the standing inequality is used in Lemma~\ref{lem:V_decay} and Theorem~\ref{thm:q_linear}), and $\kappa_c:=\kappa+\kappa_\pi$ is the \emph{gradient aggregation radius}.

	\subsection{Networked model and problem statement}

	The joint state is $\bbs=(\bbs_1,\ldots,\bbs_n)$ and the joint action $\bba=(\bba_1,\ldots,\bba_n)$. The transition law factorizes across agents into local kernels, $\mbP(\bbs'\mid\bbs,\bba)=\prod_{i=1}^n\mbP_i(\bbs_i'\mid\bbs_{\ccalN_i},\bba_{\ccalN_i})$, each agent receives a local reward $r_i(\bbs_{\ccalN_i},\bba_{\ccalN_i})$, and the initial state is drawn as $\bbs^{(0)}\sim\mu_0$, a law that does not depend on the policy parameter; the Gaussian form of $\mbP_i$ and all regularity constants are imposed in Assumption Set~A below. Policies are local and factorized, $\pi(\bba\mid\bbs)=\prod_i\pi_i(\bba_i\mid\bbs_{\ccalN_i^{\kappa_\pi}};\bbtheta_i)$, with joint parameter $\bbtheta=(\bbtheta_1,\ldots,\bbtheta_n)$ ranging over a compact product set $\Theta_0=\prod_i\Theta_{0,i}$ specified in (A4).

	The local value and action-value functions under joint policy $\pi$ are
	\feq{
		V_i^\pi(\bbs)&:=\mbE_\pi\Bigl[\sum_{t=0}^\infty\gamma^t r_i(\bbs_{\ccalN_i}^{(t)},\bba_{\ccalN_i}^{(t)})\,\Big\vert\,\bbs^{(0)}=\bbs\Bigr],\\
		Q_i^\pi(\bbs,\bba)&:=r_i(\bbs_{\ccalN_i},\bba_{\ccalN_i})+\gamma\mbE_{\bbs'\sim\mbP}[V_i^\pi(\bbs')].
	}
	The global objective is $J(\bbtheta):=\mbE_{\bbs^{(0)}\sim\mu_0}[\sum_i V_i^\pi(\bbs^{(0)})]$ with $J^\ast:=\sup_{\bbtheta\in\Theta_0}J(\bbtheta)$. Writing $\mbP_t$ for the law of $\bbs^{(t)}$ under $\pi$ started from $\mu_0$, the \emph{discounted state occupancy} and the state-action occupancy (the composition $d^\pi(d\bbs)\,\pi(d\bba\mid\bbs)$, not a product of independent measures) are
	\feq{
		d^\pi:=d^\pi_{\mu_0}:=(1-\gamma)\sum_{t=0}^\infty\gamma^t\,\mbP_t,\qquad\nu^\pi:=d^\pi\otimes\pi.
	}
	Since $\Theta_0$ is nonempty, compact, and convex by (A4), the ascent projected-gradient mapping
	\feq{
		\ccalG_\eta(\bbtheta):=\eta^{-1}\bigl(\Pi_{\Theta_0}(\bbtheta+\eta\nabla J(\bbtheta))-\bbtheta\bigr)
	}
	is the stationarity object: it coincides with $\nabla J(\bbtheta)$ whenever the unprojected update $\bbtheta+\eta\nabla J(\bbtheta)$ remains in $\Theta_0$, in particular in the unconstrained case, and $\norm{\ccalG_\eta}\leq\norm{\nabla J}$ pointwise by nonexpansiveness of $\Pi_{\Theta_0}$. Under (A4), the product structure $\Theta_0=\prod_i\Theta_{0,i}$ gives $\Pi_{\Theta_0}(\bbtheta+\eta\bbv)=(\Pi_{\Theta_{0,i}}(\bbtheta_i+\eta\bbv_i))_{i=1}^n$, so Algorithm~\ref{alg:cdcpg}'s blockwise update equals the global projection step required by the convergence analysis. We measure progress by the \emph{averaged projected-gradient-mapping stationarity}, $n^{-1}\norm{\ccalG_\eta(\bbtheta)}^2\leq\epsilon$. This per-agent normalization is an exact rescaling of the aggregate mapping: for the averaged objective $J/n$ with step size $n\eta$ one has $\ccalG^{J/n}_{n\eta}=\ccalG_\eta/n$, hence $n^{-1}\norm{\ccalG_\eta}^2=n\,\lVert\ccalG^{J/n}_{n\eta}\rVert^2$; comparisons with results stated for the averaged objective must therefore account for this factor of $n$, and comparisons with generic complexity lower bounds are made at the level of the $\epsilon$-exponent only. The policy gradient theorem~\cite{sutton1999policy} gives, under (A2)--(A3) (differentiation under the integral is justified as in Appendix~\ref{app:LJ_proof}, Step~0, whose first-derivative bounds use only (A3)),
	\begin{equation}\label{eq:pg_theorem}
		\begin{aligned}
			\nabla J(\bbtheta)
			&=\tfrac{1}{1-\gamma}\mbE_{\nu^\pi}\Bigl[\textstyle\sum_{i=1}^n Q_i^\pi(\bbs,\bba)\,\nabla_{\bbtheta}\log\pi(\bba\mid\bbs)\Bigr],
		\end{aligned}
	\end{equation}
	with the occupancy pair $(d^\pi,\nu^\pi)$ as defined above. Since $\pi(\bba\mid\bbs)=\prod_j\pi_j(\bba_j\mid\bbs_{\ccalN_j^{\kappa_\pi}};\bbtheta_j)$ factorizes across agents and only $\pi_i$ depends on $\bbtheta_i$, the per-agent gradient takes the form
	\begin{equation}\label{eq:pg_theorem_agent}
		\begin{aligned}
			\nabla_{\bbtheta_i}J(\bbtheta)
			&=\tfrac{1}{1-\gamma}\mbE_{\nu^\pi}\Bigl[\textstyle\sum_{\ell=1}^n Q_\ell^\pi(\bbs,\bba)\\
			&\qquad\qquad\quad\times\nabla_{\bbtheta_i}\log\pi_i(\bba_i\mid\bbs_{\ccalN_i^{\kappa_\pi}})\Bigr].
		\end{aligned}
	\end{equation}
	This per-agent form is the basis of Algorithm~\ref{alg:cdcpg}'s gradient estimator, constructed in Section~\ref{sec:algr}.

	\subsection{Assumptions}

	We group the assumptions into four sets. The model set~A, the decay set~B, and the sampling set~D are stated here; the TD-stability sets~C and~C$'$ constrain the population temporal-difference matrix built from the critic features, and are therefore stated in Section~\ref{sec:algr}, immediately after that matrix has been defined.

	\begin{assumption}[Set A: Model and regularity]\label{asm:A}
		\hfill
		\begin{enumerate}
			\item[\emph{(A1)}] Local dynamics $\bbs_i'=f_i(\bbs_{\ccalN_i},\bba_{\ccalN_i})+\bbvarepsilon_i$ with $f_i$ deterministic and continuous, $\sup\norm{f_i}\leq B_f$, and $\bbvarepsilon_i\sim\ccalN(\bbzero,\sigma^2\bbI_{d_S})$ independent across $i,t$, for a fixed $\sigma>0$.
			\item[\emph{(A2)}] Local rewards are measurable with $\abs{r_i(\bbs_{\ccalN_i},\bba_{\ccalN_i})}\leq\bbarr$.
			\item[\emph{(A3)}] For each $i$, let $U_i$ be a bounded open convex neighborhood of $\Theta_{0,i}$. The density $\pi_i(\bba_i\mid\bbs_{\ccalN_i^{\kappa_\pi}};\bbtheta_i)$ is jointly measurable in action, state, and parameter, vanishes off a fixed measurable support $\ccalA_i^{\rm supp}\subseteq\ccalA_i$, and is strictly positive on that support outside a common action-null set, for every conditioning state and every $\bbtheta_i\in U_i$. It is differentiable in $\bbtheta_i$ on $U_i$ with $\norm{\nabla_{\bbtheta_i}\log\pi_i}\leq G$ uniformly over $\bbtheta_i\in U_i$, the conditioning state, and $\bba_i\in\ccalA_i^{\rm supp}$ outside that null set. The joint policy factors as $\pi(\bba\mid\bbs)=\prod_i\pi_i(\bba_i\mid\bbs_{\ccalN_i^{\kappa_\pi}})$.
			\item[\emph{(A4)}] The parameter space has a product structure $\Theta_0=\prod_{i=1}^n\Theta_{0,i}\subseteq\mbR^{nd_\theta}$, with each $\Theta_{0,i}\subset\mbR^{d_\theta}$ compact convex containing $\bbtheta_i^{(0)}$.
		\end{enumerate}
	\end{assumption}

	Under (A2), $V_i^\pi$ and $Q_i^\pi$ are uniformly bounded by $Q_{\max}:=\bbarr/(1-\gamma)$; the map $\bbtheta\mapsto J(\bbtheta)$ is continuous on $\Theta_0$ under (A2)--(A3) by dominated convergence, so $J^\ast$ is attained on the compact $\Theta_0$, and $\Delta_J:=J^\ast-J(\bbtheta^{(0)})\leq 2nQ_{\max}=\ccalO(n)$.

	Compactly parameterized exponential families with uniformly bounded sufficient statistics have bounded parameter scores. Conventional nondegenerate Gaussian location policies followed by a $\tanh$ action transformation have unbounded parameter scores near the action boundary, so we use bounded-score policy classes as an explicit premise. Assumption~(A3) is invoked to justify bounded-gradient policy-gradient concentration together with the locality and decay arguments; no uniform upper or lower bound on the policy density is imposed, only strict positivity on a fixed common support, which underwrites the log-density comparison in the continuity argument for the population TD matrix (Lemma~\ref{lem:M_continuity}).

	\begin{assumption}[Set B: Decay structure]\label{asm:B}
		\hfill
		\begin{enumerate}
			\item[\emph{(B1)}] Every policy in the considered class satisfies $(c,\rho)$-exponential decay: for all $i\in[n]$, $\kappa\in\mbN$, and $(\bbs,\bba),(\bbs',\bba')$ agreeing on $\ccalN_i^\kappa$ at the state-action level,
			\feq{
				\abs{Q_i^\pi(\bbs,\bba)-Q_i^\pi(\bbs',\bba')}\leq c\rho^{\kappa+1},\label{eq:exp_decay}
			}
			with constants $c>0$ and $\rho\in(0,1)$ independent of $n$.
		\end{enumerate}
	\end{assumption}

	Set~B contains a single item. Lemma~\ref{lem:baseline_decay} in Section~\ref{sec:main} shows that (B1) always holds with an explicit baseline pair $(c_0,\rho_0)$ obtained from finite propagation and bounded rewards alone. The assumption content of (B1) is therefore the rate: the graph-tail conditions $(\Delta-1)\rho<1$ (used in Lemma~\ref{lem:grad_bias} and (D3)(iii)) and $\Delta\rho<1$ (used in Proposition~\ref{prop:adaptive}) require $\rho$ smaller than the baseline $\rho_0$ unless the discount factor is small relative to the graph growth; on the path graph of Section~\ref{sec:simu}, $\Delta=2$ and $(\Delta-1)\rho_0<1$ holds automatically. Appendix~\ref{app:lipschitz} gives the finite-propagation proof; a faster rate is not derived here. A common decay base for (B1) and (D3)(ii) can always be taken as the larger of the two, provided the required tail inequality survives. The boundedness of the spectral splitting factor used by the feature construction is not assumed but derived from (A1) in Section~\ref{sec:spectral}. Dobrushin-type influence criteria~\cite{dobrushin1968description} and spatial correlation decay in local filtering~\cite{rebeschini2015can} motivate sufficient conditions for faster spatial decay; verifying such a criterion for the present policy-controlled Gaussian model is a separate task, and those filtering assumptions do not automatically verify the controlled-MDP decay premise.

	\begin{assumption}[Set D: Sampling and smoothness]\label{asm:D}
		\hfill
		\begin{enumerate}
			\item[\emph{(D1)}] At each iteration $k$, a shared dataset $\ccalD_s^{(k)}$ of $M_s$ i.i.d.\ transitions $(\bbs,\bba,\bbs',\bba')$ with $\bbs\sim d^{\pi^{(k)}}$, $\bba\sim\pi^{(k)}(\cdot\mid\bbs)$, $\bbs'\sim\mbP(\cdot\mid\bbs,\bba)$, $\bba'\sim\pi^{(k)}(\cdot\mid\bbs')$ is drawn and shared across agents for LSTD; a separate shared dataset $\ccalD_g^{(k)}$ of $M_g$ pairs $(\bbs,\bba)\sim d^{\pi^{(k)}}\otimes\pi^{(k)}$ is drawn for gradient estimation. The random-feature realization is sampled once at preprocessing and held fixed. Conditional on the current iterate and all preprocessing and previous-iteration randomness, the critic and actor batches are independent and have the stated product laws. \emph{(Local model access.)} Each agent $i$ knows the drift functions $f_j$ for $j\in\ccalN_i^{\kappa_c+\kappa}$, the reward functions $r_\ell$ for $\ell\in\ccalN_i^{\kappa_c}$, and the noise scale $\sigma$, and the random-feature draws of all $\ell\in\ccalN_i^{\kappa_c}$ are reproducible from pseudorandom seeds broadcast once at preprocessing; these quantities are exactly what is needed to evaluate the augmented features of Definition~\ref{def:aug_features} for agent $i$ and its $\kappa_c$-neighbors. The feature construction is thus \emph{model-assisted}: transition and reward functions are known locally, while the occupancy measure is accessed only through samples.
			\item[\emph{(D2)}] $J(\bbtheta)$ is $L_J$-smooth on $\Theta_0$, i.e.\ $\nabla J$ is $L_J$-Lipschitz on $\Theta_0$. We treat $L_J$ as a primitive constant; under the additional regularity of (D3) below, $L_J$ is moreover independent of the network size $n$, and Appendix~\ref{app:LJ_proof} derives this $n$-uniformity from (D3).
			\item[\emph{(D3)}] \emph{(Differentiated regularity, used only for the $n$-uniformity of $L_J$.)} (i) Each local policy density is twice continuously differentiable in $\bbtheta_i$ on the neighborhood $U_i$ of (A3), and its score Jacobian has operator norm at most $G'$ uniformly over $\bbtheta_i\in U_i$, the conditioning state, and the action on the support outside the common null set of (A3); under (A2)--(A3) and (i), $J$ is then twice continuously differentiable on $\Theta_0$ (Appendix~\ref{app:LJ_proof}, Step~0). (ii) \emph{(Differentiated decay.)} For the network family under consideration, there exist constants $c_\partial<\infty$ and $\rho\in(0,1)$ (the decay base of (B1)), independent of $n$, such that the mixed parameter curvature of the objective decays geometrically in graph distance: for every pair $i,j$ and every $\bbtheta\in\Theta_0$,
			\feq{
				\norm{[\nabla^2 J(\bbtheta)]_{ij}}\leq c_\partial\,\rho^{\,d(i,j)},
			}
			where $[\nabla^2 J]_{ij}:=\partial^2_{\bbtheta_i\bbtheta_j}J\in\mbR^{d_\theta\times d_\theta}$ is the cross-parameter Hessian block. This is the curvature analogue of the value-decay hypothesis (B1): whereas (B1) postulates that agent $i$'s value couples weakly to distant agents' state-actions, (ii) postulates that the cross-parameter curvature couples agents $i$ and $j$ weakly at graph distance. It does not follow from (B1): differentiating the policy-gradient identity in $\bbtheta_j$ produces value-sensitivity, occupancy-sensitivity, and score-derivative terms, and the score-derivative term contains the sum of all local rewards, so termwise bounds can be of order $n$ and do not establish decay of the full Hessian block; (B1) alone does not control these mixed second-order differences (Appendix~\ref{app:LJ_proof}). (iii) The graphs of the family have the common maximum-degree upper bound $\Delta$, and the graph-tail hypothesis $\tilde\delta:=(\Delta-1)\rho<1$ holds.
		\end{enumerate}
	\end{assumption}

	Sample counts in (D1) measure generative-oracle calls; a standard transition oracle does not by itself provide an independent discounted-occupancy draw at unit cost. One realization restarts from $\mu_0$, draws a geometric time $T$ with $\mbP(T=t)=(1-\gamma)\gamma^t$, and simulates to that time, so the expected rollout length is of order $1/(1-\gamma)$ environment steps per sample, and the critic tuple additionally requires one successor transition. The $n$-uniformity of $L_J$ in (D2)--(D3) is essential to the network-size scalability claim of Theorem~\ref{thm:sample_complexity_fixed}; if $L_J$ instead scales with $n$ (for example, under a sum-of-values objective without the differentiated-decay structure of (D3)), the descent term of the complexity inherits the same scaling. Appendix~\ref{app:LJ_proof} derives the $n$-uniform smoothness constant from (D3); (D3) is a separate premise, not a consequence of (B1), and it is needed only for the network-family claims, not for the fixed-$n$ theorems.

	\section{Spectral Representation Framework}\label{sec:spectral}

	For $I=\ccalN_i^\kappa$, the conditional density of the successor state on $I$ is Gaussian with mean $\bbf_{i,\kappa}(\bbZ)$ defined below, i.e.\ $\mbP_{i,\kappa}=\prod_{j\in I}\mbP_j$ with $\mbP_j(\bbs_j'\mid\cdot)=k(\bbs_j'-f_j(\cdot))$ and $k(\bbDelta)=(2\pi\sigma^2)^{-d_S/2}\exp(-\norm{\bbDelta}^2/(2\sigma^2))$. The Gaussian function $k_0(\bbx-\bby)$ below is shift invariant in its two Euclidean arguments and admits a Fourier representation; we apply this representation after substituting the local drift into the first argument. The transition map itself need not be translation invariant in its state-action input.

	Let $d_{i,\kappa}:=\abs{\ccalN_i^\kappa}\cdot d_S$. For each $i\in[n]$ and $\ell\in[m]$, draw $\bbomega_{i,\ell}\sim\ccalN(\bbzero,\sigma^{-2}\bbI_{d_{i,\kappa}})$ and $b_{i,\ell}\sim\mathrm{Unif}[0,2\pi]$, with all frequency and phase draws mutually independent. Fix $\alpha\in(0,1)$ and define
	\begin{equation}\label{eq:g_alpha_def}
		\begin{aligned}
			g_\alpha^{(i)}(\bbs,\bba)&:=(2\pi\sigma^2)^{-d_{i,\kappa}/2}\exp\!\Bigl(\tfrac{\alpha^2\norm{\bbf_{i,\kappa}}^2}{2\sigma^2(1-\alpha^2)}\Bigr),\\
			p_\alpha^{(i)}(\bbs_{\ccalN_i^\kappa}')&:=\exp\!\Bigl(-\tfrac{\alpha^2\norm{\bbs_{\ccalN_i^\kappa}'}^2}{2\sigma^2}\Bigr),
		\end{aligned}
	\end{equation}
	with $\bbf_{i,\kappa}:=(f_j)_{j\in\ccalN_i^\kappa}$, a function of the $(\kappa{+}1)$-neighborhood input $\bbZ:=(\bbs_{\ccalN_i^{\kappa+1}},\bba_{\ccalN_i^{\kappa+1}})$, giving the splitting identity
	\begin{equation}\label{eq:p_splitting_identity}
		\begin{aligned}
			&\mbP_{i,\kappa}(\bbs_{\ccalN_i^\kappa}'\mid\cdot)\\
			&\quad=g_\alpha^{(i)}\cdot p_\alpha^{(i)}\cdot k_0\!\Bigl(\tfrac{\bbf_{i,\kappa}}{\sqrt{1-\alpha^2}}-\sqrt{1-\alpha^2}\,\bbs_{\ccalN_i^\kappa}'\Bigr),
		\end{aligned}
	\end{equation}
	$k_0(\bbDelta)=\exp(-\norm{\bbDelta}^2/(2\sigma^2))$. Under (A1), $\norm{\bbf_{i,\kappa}}\leq\sqrt{D_\kappa}B_f$, so the splitting factor is explicitly and uniformly bounded:
	\feq{&\sup_{i,\bbZ}g_\alpha^{(i)}(\bbZ)\leq\bbarg_\alpha(\kappa):=\\
		&\qquad\qquad\max\{1,(2\pi\sigma^2)^{-1}\}^{D_\kappa d_S/2}\exp\!\Bigl(\tfrac{\alpha^2 D_\kappa B_f^2}{2\sigma^2(1-\alpha^2)}\Bigr),\label{eq:g_alpha_bound}
	}
	a derived bound (not an assumption) that is finite for fixed $\kappa$ but exponential in the neighborhood dimension $D_\kappa d_S$; this dependence enters $\tilde g_\alpha$ and the resulting approximation constants; locality removes the full network dimension at a fixed radius, but high-dimensional local approximation can still be expensive.

	\begin{definition}[Normalized RFF]\label{def:rff}
		The raw RFF features are
		\feq{
			\phi_{i,\kappa}^{\rm raw}(\bbZ)&:=g_\alpha^{(i)}(\bbZ)\sqrt{\tfrac{2}{m}}\\
			&\quad\times\Bigl\{\cos\!\Bigl(\tfrac{\bbomega_{i,\ell}^\top\bbf_{i,\kappa}(\bbZ)}{\sqrt{1-\alpha^2}}+b_{i,\ell}\Bigr)\Bigr\}_{\ell=1}^m,\\
			\mu_{i,\kappa}^{\rm raw}(\bby)&:=p_\alpha^{(i)}(\bby)\sqrt{\tfrac{2}{m}}\\
			&\quad\times\Bigl\{\cos\!\bigl(\sqrt{1-\alpha^2}\,\bbomega_{i,\ell}^\top\bby+b_{i,\ell}\bigr)\Bigr\}_{\ell=1}^m.
		}
		The normalized features $\hhatphi_{i,\kappa}:=\phi_{i,\kappa}^{\rm raw}/\bbarg_\alpha$ and $\hhatmu_{i,\kappa}:=\mu_{i,\kappa}^{\rm raw}\bbarg_\alpha$ preserve the inner product and satisfy $\norm{\hhatphi_{i,\kappa}}^2\leq 2$. The finite-feature approximant $\inner{\hhatphi_{i,\kappa}}{\hhatmu_{i,\kappa}}$ is a signed function and need not be a (sub)probability kernel; it enters the analysis only through its $\ccalL^1$ distance to $\mbP_{i,\kappa}$ against bounded integrands.
	\end{definition}

	The feature map is the \emph{random-phase} (single-cosine) construction of~\cite{rahimi2007random}. Its kernel estimator $\tfrac2m\sum_\ell\cos(\bbomega_\ell^\top\bbx+b_\ell)\cos(\bbomega_\ell^\top\bby+b_\ell)$ equals $\tfrac1m\sum_\ell\cos(\bbomega_\ell^\top(\bbx-\bby))+\tfrac1m\sum_\ell\cos(\bbomega_\ell^\top(\bbx+\bby)+2b_\ell)$; the second term is mean zero but is a nonzero random process, so the estimator is not a function of $\bbx-\bby$ alone, and the uniform bound stated for the paired sine--cosine map cannot be imported verbatim. The distinction between the paired and random-phase maps, including their different uniform-error bounds, is analyzed in~\cite[Sec.~2.2.1]{sutherland2015error}. We prove a bound in the form needed here and convert it into a global integrated error bound for the local transition density.

	\begin{theorem}[RFF approximation]\label{thm:rff_approx}
		Let $\tilde g_\alpha:=\bbarg_\alpha\max\{1,(2\pi\sigma^2/\alpha^2)\}^{D_\kappa d_S/2}$, $R_2(\epsilon_P):=\sqrt{D_\kappa}B_f+(\sigma/\alpha)(\sqrt{D_\kappa d_S}+\sqrt{2\log(8\tilde g_\alpha/\epsilon_P)})$, $D_\ccalR^{\rm int}(\epsilon_P):=\sqrt{D_\kappa}B_f/\sqrt{1-\alpha^2}+2R_2(\epsilon_P)$, and $K_\delta:=4n\sqrt{D_\kappa d_S}/(\sigma\delta_{\rm rff})+\sigma^{-1}$. For $\delta_{\rm rff}\in(0,1)$, $\epsilon_P\in(0,1]$, if
		\begin{equation}\label{eq:rff_feasibility}
			\begin{aligned}
				m&\geq C_{\rm RFF}(D_\kappa d_S)\,\frac{\tilde g_\alpha^2}{\epsilon_P^2}\\
				&\quad\times\log\!\Bigl(\frac{4n}{\delta_{\rm rff}}\Bigl(1+\frac{16\,D_\ccalR^{\rm int}(\epsilon_P)\,K_\delta\,\tilde g_\alpha}{\epsilon_P}\Bigr)\Bigr),
			\end{aligned}
		\end{equation}
		with $C_{\rm RFF}(d):=128(2d+1)$, then the favorable RFF event
		\feq{
			E_{\rm rff}:=\Bigl\{\sup_{i,\bbZ}\norm{\mbP_{i,\kappa}(\cdot\mid\bbZ)-\inner{\hhatphi_{i,\kappa}(\bbZ)}{\hhatmu_{i,\kappa}(\cdot)}}_{\ccalL^1}\leq\epsilon_P\Bigr\}
		}
		satisfies $\mbP(E_{\rm rff})\geq 1-\delta_{\rm rff}$. We define $\epsilon_P(m,\delta_{\rm rff},n):=\inf\{\epsilon_P\in(0,1]:\text{\eqref{eq:rff_feasibility} holds}\}$; if the feasible set is nonempty its minimum is attained, because the right-hand side of~\eqref{eq:rff_feasibility} is continuous and strictly decreasing on $(0,1]$ and diverges as $\epsilon_P\to0$, so the feasible set is an interval $[\epsilon_0,1]$ with $\epsilon_0>0$ (if the feasible set is empty we set $\epsilon_P(m,\delta_{\rm rff},n):=+\infty$ and every bound involving it is vacuous); the minimum satisfies $\epsilon_P(m)^2=\widetilde\ccalO(C_{\rm RFF}(D_\kappa d_S)\,\tilde g_\alpha^2/m)=\widetilde\ccalO(\tilde g_\alpha^2 D_\kappa d_S/m)$, where the polylogarithmic factor contains $\log(n/\delta_{\rm rff})$. The random-feature dimension $m$ is a problem primitive fixed before the target stationarity tolerance $\epsilon$ is sent to zero in Theorem~\ref{thm:sample_complexity_fixed}.
	\end{theorem}

	The proof (Appendix~\ref{app:rff_proof}) combines (i) a uniform bound for the random-phase estimator on a scaled compact core at tolerance $\epsilon_P/(2\tilde g_\alpha)$, obtained from a Markov bound on the empirical frequency norms, a Lipschitz argument, a deterministic net, and Hoeffding's inequality at the net points, (ii) integration of the weighted core error against the Gaussian factor $p_\alpha^{(i)}$, (iii) the true Gaussian-transition tail, and (iv) the tail of the signed approximant via the pointwise envelope $2g_\alpha^{(i)}p_\alpha^{(i)}$. The constant $C_{\rm RFF}(d)=128(2d+1)$ is conservative; the order $d/m$ in $\epsilon_P^2$ is what the analysis uses. The intrinsic boundedness $\norm{\bbf_{i,\kappa}}\leq\sqrt{D_\kappa}B_f$ from (A1) makes the conclusion global in $\bbZ$; no confinement of the current state is required.

	\section{Exponential Decay and $Q$-Function Representation}\label{sec:main}

	\subsection{Exponential decay}

	\begin{lemma}[Baseline decay from finite propagation]\label{lem:baseline_decay}
		Under (A2)--(A3) and the local transition structure of Section~\ref{sec:pre}, every policy in the class satisfies (B1) with
		\feq{
			\rho_0:=\gamma^{1/(\kappa_\pi+1)},\qquad c_0:=\frac{2\bbarr}{(1-\gamma)\rho_0}.
		}
		The bound uses only bounded rewards, one-hop transitions, and $\kappa_\pi$-hop policies; it does not require contraction of the drift.
	\end{lemma}

	The proof (Appendix~\ref{app:lipschitz}) couples two trajectories through common transition and action-sampling randomization, with the prescribed time-zero actions retained: one transition and the subsequent local action selection reduce the guaranteed agreement radius by at most $\kappa_\pi+1$, so the rewards of agent $i$ agree for the first $\lceil\kappa/(\kappa_\pi+1)\rceil$ steps and the remaining discounted tail is at most $2\bbarr\gamma^{\lceil\kappa/(\kappa_\pi+1)\rceil}/(1-\gamma)$. Faster decay, needed for the graph-tail conditions when $(\Delta-1)\rho_0\geq 1$, remains the assumption content of (B1).

	\begin{lemma}[Value decay from $Q$-decay]\label{lem:V_decay}
		Under (A3) and (B1), for $\kappa\geq\kappa_\pi+1$ and any $\bbs,\bbs'$ with $\bbs_{\ccalN_i^\kappa}=\bbs'_{\ccalN_i^\kappa}$,
		\feq{
			\abs{V_i^\pi(\bbs)-V_i^\pi(\bbs')}\leq\tilde c\rho^{\kappa+1},\quad\tilde c:=2c\rho^{-\kappa_\pi}.
		}
	\end{lemma}

	The short proof is in Appendix~\ref{app:V_decay}.

	\subsection{Direct linear representation in the deployed features}

	The critic of Algorithm~\ref{alg:cdcpg} evaluates the raw local feature at the $(\kappa{+}1)$-neighborhood input $\bbZ$; this input contains all local state-action variables required to evaluate the true $\kappa$-neighborhood transition $\mbP_{i,\kappa}(\cdot\mid\bbZ)=\prod_{j\in\ccalN_i^\kappa}\mbP_j(\cdot\mid\bbs_{\ccalN_j},\bba_{\ccalN_j})$. The representation below therefore uses the actual local transition law, with no truncated or conditional Bellman model, and it is stated pointwise on the full state-action space.

	\begin{definition}[Augmented features and representation weight]\label{def:aug_features}
		For $\bbz=(\bbs_{\ccalN_i^{\kappa+1}},\bba_{\ccalN_i^{\kappa+1}})$,
		\feq{
			\tdbbphi_{i,\kappa}^{\rm raw}(\bbz):=\bigl(r_i(\bbs_{\ccalN_i},\bba_{\ccalN_i}),\hhatphi_{i,\kappa}(\bbz)\bigr)^\top\in\mbR^{m+1},
		}
		which satisfies $\norm{\tdbbphi_{i,\kappa}^{\rm raw}}\leq L:=\sqrt{\bbarr^2+2}$ pointwise. Fix any reference state $\bbs^\circ\in\ccalS$ and write $I:=\ccalN_i^\kappa$. The \emph{local continuation value} is the analysis-only function
		\feq{
			h_i^\pi(\bby):=V_i^\pi(\bby,\bbs^\circ_{I^c}),\qquad\bby\in\ccalS_{\ccalN_i^\kappa},
		}
		which is bounded by $Q_{\max}$, and the representation weight is
		\feq{
			\bbu_{i,\kappa}^\pi:=\Bigl(1,\,\gamma\!\int\!\hhatmu_{i,\kappa}(\bby)\,h_i^\pi(\bby)\,d\bby\Bigr)^\top\in\mbR^{m+1}.\label{eq:weight_vector}
		}
		The reference completion is not substituted into the dynamics; it only selects a representative local continuation value. The integral is well defined because $\int\norm{\hhatmu_{i,\kappa}(\bby)}\,d\bby\leq\sqrt2\,\tilde g_\alpha$ (Appendix~\ref{app:q_linear_proof}), whence $\norm{\bbu_{i,\kappa}^\pi}\leq U_1:=1+\sqrt2\gamma\tilde g_\alpha\bbarr/(1-\gamma)$, a uniform bound on the representation coefficient (which need not solve the population TD equations).
	\end{definition}

	\begin{theorem}[Linear representation of the local-reward $Q$-function]\label{thm:q_linear}
		Under Sets~A--B with $\kappa\geq\max\{1,\kappa_\pi+1\}$, on the favorable RFF event $E_{\rm rff}$ of Theorem~\ref{thm:rff_approx} at accuracy $\epsilon_P$, for every policy $\pi$ in the class, every $i\in[n]$, and every $(\bbs,\bba)\in\ccalS\times\ccalA$,
		\feq{
			Q_i^\pi(\bbs,\bba)=\inner{\tdbbphi_{i,\kappa}^{\rm raw}(\bbZ)}{\bbu_{i,\kappa}^\pi}+e_i(\bbs,\bba),\label{eq:q_linear}
		}
		with $\bbZ=(\bbs_{\ccalN_i^{\kappa+1}},\bba_{\ccalN_i^{\kappa+1}})$ and
		\feq{
			\sup_{\bbs,\bba}\abs{e_i(\bbs,\bba)}\leq\Bk:=\underbrace{\gamma\tilde c\rho^{\kappa+1}}_{=:B_{\rm loc}(\kappa)}+\underbrace{\frac{\gamma\bbarr\,\epsilon_P}{1-\gamma}}_{=:B_{\rm rff}(m)}.\label{eq:xi_bound}
		}
		The event $E_{\rm rff}$ is policy independent; only the coefficient $\bbu_{i,\kappa}^\pi$ varies with the policy.
	\end{theorem}

	Appendix~\ref{app:q_linear_proof} proves the representation using the actual local transition conditioned on its boundary inputs. The two residual terms respectively quantify the omitted distant state in the continuation value and the finite-feature approximation of the true local transition.

	\section{Algorithm and LSTD Analysis}\label{sec:algr}

	\subsection{Population TD matrix and stability assumptions}

	The raw-feature population objects at parameter $\bbtheta$ are
	\begin{equation}\label{eq:M_pop_raw}
		\bbM_i^{\bbtheta}:=\bbH_i^{\bbtheta}-\gamma\bbC_i^{\bbtheta},
	\end{equation}
	with $\bbH_i^{\bbtheta}:=\mbE[\tdbbphi^{\rm raw}(\bbZ)\tdbbphi^{\rm raw}(\bbZ)^\top]$, $\bbC_i^{\bbtheta}:=\mbE[\tdbbphi^{\rm raw}(\bbZ)\tdbbphi^{\rm raw}(\bbZ')^\top]$, and $\bbb_i^{\bbtheta}:=\mbE[\tdbbphi^{\rm raw}(\bbZ)r_i]$, where $(\bbS,\bbA)\sim\nu^{\pi_\bbtheta}$, $\bbZ$ is its restriction to $\ccalN_i^{\kappa+1}$, and $\bbZ'$ is the corresponding restriction of the successor pair $(\bbS',\bbA')$ under the one-step law of (D1).

	\begin{assumption}[Set C: TD stability]\label{asm:C}
		Define the event
		\feq{
			E_M:=\bigl\{\sigma_{\min}(\bbM_i^{\bbtheta})\geq\mu_M(m)>0,\ \forall i\in[n],\,\forall\bbtheta\in\Theta_0\bigr\},\label{eq:td_stability}
		}
		where $\bbM_i^{\bbtheta}$ is the raw-feature TD matrix~\eqref{eq:M_pop_raw}. There is a constant $\delta_M(m)\in(0,1)$ such that $\mbP(E_M)\geq 1-\delta_M(m)$.
	\end{assumption}

	Three remarks delimit the scope of Set~C. First, $E_M$ is measurable: under Set~A, $\bbtheta\mapsto\bbM_i^{\bbtheta}$ is continuous on the compact $\Theta_0$ (Lemma~\ref{lem:M_continuity}, Appendix~\ref{app:M_continuity}), so the infimum over $\Theta_0$ reduces to one over a countable dense subset. Second, the notation suppresses dependencies: $\mu_M$ and $\delta_M$ may depend not only on $m$ but also on $n$, $\kappa$, the graph, and the policy class $\Theta_0$; in particular, with independent per-agent feature draws a per-agent failure probability $\delta_M^{(1)}$ compounds by a union bound to $\delta_M\leq n\,\delta_M^{(1)}$. Third, the confidence $1-\delta_M(m)$ is postulated, not user-controllable: unlike $\delta_{\rm rff}$, it cannot be driven to zero by enlarging a sample size. A practical (not theoretical) recourse is to redraw the feature realization and re-test conditioning via the diagnostic of Proposition~\ref{prop:diagnostic}; the single-draw confidence statement of that proposition does not cover such repeated selection, and we do not analyze the associated confidence spending. Lemma~\ref{lem:excitation} below reduces (C) to a symmetric feature-excitation condition. Sets~A, B, and D do not imply positive excitation: if the drift is constant, all RFF coordinates are constant functions of the input and the augmented feature family has rank at most two, so $\bbH_i^{\bbtheta}$ is singular for $m\geq2$.

	The $m$-dependence of $\mu_M$ is necessary:

	\begin{proposition}[Dimensional obstruction to dimension-free TD-stability]\label{prop:mu_M_obstruction}
		For any TD matrix $\bbM_i^{\bbtheta}=\mbE[\tdbbphi^{\rm raw}(\bbZ)(\tdbbphi^{\rm raw}(\bbZ)-\gamma\tdbbphi^{\rm raw}(\bbZ'))^\top]$ generated by features with values in $\mbR^{m+1}$ and $\norm{\tdbbphi^{\rm raw}}\leq L$ pointwise (so that both the current and the successor feature are bounded),
		\feq{
			\sigma_{\min}(\bbM_i^{\bbtheta})\leq\frac{\norm{\bbM_i^{\bbtheta}}_\ast}{m+1}\leq\frac{(1+\gamma)L^2}{m+1}\xrightarrow{m\to\infty}0.\label{eq:nuclear_obstruction}
		}
		Hence every admissible threshold in Assumption~\ref{asm:C} obeys $\mu_M(m)\leq(1+\gamma)L^2/(m+1)$, no dimension-free lower bound exists, and the analysis below operates in the fixed-$m$ regime with $m$ a problem primitive.
	\end{proposition}

	\begin{proof}
		For a rank-one matrix $\bbu\bbv^\top$ the nuclear norm is $\norm{\bbu}\norm{\bbv}$. The nuclear norm is convex, so by Jensen's inequality for the Bochner integral, $\norm{\bbM_i^{\bbtheta}}_\ast\leq\mbE\norm{\tdbbphi^{\rm raw}(\bbZ)}\,\norm{\tdbbphi^{\rm raw}(\bbZ)-\gamma\tdbbphi^{\rm raw}(\bbZ')}\leq(1+\gamma)L^2$. Every one of the $m+1$ singular values is at least $\sigma_{\min}(\bbM_i^{\bbtheta})$, so $(m+1)\sigma_{\min}(\bbM_i^{\bbtheta})\leq\norm{\bbM_i^{\bbtheta}}_\ast$. Neither stationarity nor positive definiteness is used.
	\end{proof}

	The weaker Frobenius-norm bound $\sigma_{\min}\leq\norm{\bbM_i^{\bbtheta}}_F/\sqrt{m+1}\leq(1+\gamma)L^2/\sqrt{m+1}$ also holds but is not tight in this argument; \eqref{eq:nuclear_obstruction} is the general upper bound used below. It is an upper bound on the true conditioning: a threshold $\mu_M(m)=\Theta(1/m)$ is compatible with it, but attainability for the specific spectral features is not proved, and faster deterioration or singularity remains possible.

	\begin{lemma}[Excitation implies TD stability; accretivity]\label{lem:excitation}
		Under (A1)--(A3) and the occupancy sampling of (D1), for every $i\in[n]$, every $\bbtheta\in\Theta_0$, and every feature realization:
		\begin{enumerate}
			\item[(a)] the symmetric part of the TD matrix is positive semidefinite, $\sym(\bbM_i^{\bbtheta})\succeq(1-\sqrt\gamma)\,\bbH_i^{\bbtheta}\succeq\bbzero$; consequently $\sigma_{\min}(\bbM_i^{\bbtheta}+\lambda\bbI)\geq\sigma_{\min}(\bbM_i^{\bbtheta})$ for every $\lambda\geq0$;
			\item[(b)] whenever $\bbH_i^{\bbtheta}\succ\bbzero$, the whitened cross term obeys $\norm{\bbH_i^{\bbtheta,-1/2}\bbC_i^{\bbtheta}\bbH_i^{\bbtheta,-1/2}}_2\leq\gamma^{-1/2}$, and consequently
			\feq{
				\sigma_{\min}(\bbM_i^{\bbtheta})\;\geq\;(1-\sqrt\gamma)\,\lambda_{\min}(\bbH_i^{\bbtheta}).
			}
		\end{enumerate}
		Hence the symmetric feature-excitation condition $\lambda_{\min}(\bbH_i^{\bbtheta})\geq h_M(m)>0$ for all $i\in[n]$ and $\bbtheta\in\Theta_0$, on an event of probability $\geq 1-\delta_M(m)$ over the feature draw, implies Assumption~(C) with $\mu_M(m)=(1-\sqrt\gamma)\,h_M(m)$.
	\end{lemma}

	\begin{proof}
		Fix $i,\bbtheta$ and the feature realization, and abbreviate $\bbH:=\bbH_i^{\bbtheta}$, $\bbC:=\bbC_i^{\bbtheta}$, $\bbM:=\bbM_i^{\bbtheta}$. Let $\nu:=\nu^{\pi_\bbtheta}$ and let $\nu_+$ denote the law of the successor full state-action pair $(\bbS',\bbA')$ under the one-step law of (D1), whose restriction to $\ccalN_i^{\kappa+1}$ is $\bbZ'$. The occupancy identity $d^\pi=(1-\gamma)\mu_0+\gamma\,\mbP^{\pi\top}d^\pi$ gives $\mbP^{\pi\top}d^\pi\leq d^\pi/\gamma$ as measures, hence $\nu_+=(\mbP^{\pi\top}d^\pi)\otimes\pi_\bbtheta\leq\nu/\gamma$ (this is a domination statement, not stationarity), and therefore the successor feature covariance $\bbH_+:=\mbE_{\nu_+}[\tdbbphi^{\rm raw}\tdbbphi^{\rm raw\top}]$ satisfies $\bbu^\top\bbH_+\bbu=\int(\bbu^\top\tdbbphi^{\rm raw})^2\,d\nu_+\leq\gamma^{-1}\int(\bbu^\top\tdbbphi^{\rm raw})^2\,d\nu=\gamma^{-1}\bbu^\top\bbH\bbu$ for every $\bbu$, i.e.\ $\bbH_+\preceq\bbH/\gamma$.

		(a) For any $\bbx$, Cauchy--Schwarz under the joint one-step law gives $\bbx^\top\bbM\bbx=\bbx^\top\bbH\bbx-\gamma\,\mbE[(\bbx^\top\tdbbphi^{\rm raw}(\bbZ))(\bbx^\top\tdbbphi^{\rm raw}(\bbZ'))]\geq\bbx^\top\bbH\bbx-\gamma(\bbx^\top\bbH\bbx)^{1/2}(\bbx^\top\bbH_+\bbx)^{1/2}\geq(1-\sqrt\gamma)\,\bbx^\top\bbH\bbx$. Since $\bbx^\top\bbM\bbx=\bbx^\top\sym(\bbM)\bbx$, this is the claimed semidefinite inequality. Then, for $\lambda\geq0$, $\norm{(\bbM+\lambda\bbI)\bbx}^2=\norm{\bbM\bbx}^2+2\lambda\bbx^\top\sym(\bbM)\bbx+\lambda^2\norm{\bbx}^2\geq\norm{\bbM\bbx}^2$, and minimizing over unit $\bbx$ gives $\sigma_{\min}(\bbM+\lambda\bbI)\geq\sigma_{\min}(\bbM)$.

		(b) For unit vectors $\bbu,\bbv$, Cauchy--Schwarz under the joint one-step law gives
		\feq{
			\abs{\bbu^\top\bbH^{-1/2}\bbC\bbH^{-1/2}\bbv}
			&\leq\bigl(\mbE(\bbu^\top\bbH^{-1/2}\tdbbphi^{\rm raw}(\bbZ))^2\bigr)^{1/2}\notag\\
			&\qquad\times\bigl(\mbE(\bbv^\top\bbH^{-1/2}\tdbbphi^{\rm raw}(\bbZ'))^2\bigr)^{1/2}\notag\\
			&=\bigl(\bbv^\top\bbH^{-1/2}\bbH_+\bbH^{-1/2}\bbv\bigr)^{1/2}\leq\gamma^{-1/2}.
		}
		Writing $\bbM=\bbH^{1/2}(\bbI-\gamma\bbB)\bbH^{1/2}$ with $\bbB:=\bbH^{-1/2}\bbC\bbH^{-1/2}$ and using $\sigma_{\min}(\bbX\bbY\bbZ)\geq\sigma_{\min}(\bbX)\,\sigma_{\min}(\bbY)\,\sigma_{\min}(\bbZ)$ for square matrices,
		$\sigma_{\min}(\bbM)\geq\lambda_{\min}(\bbH)\,\sigma_{\min}(\bbI-\gamma\bbB)\geq\lambda_{\min}(\bbH)(1-\gamma\norm{\bbB})\geq(1-\sqrt\gamma)\lambda_{\min}(\bbH)$.
		If $\lambda_{\min}(\bbH)=0$ the displayed bound holds trivially.
	\end{proof}

	Lemma~\ref{lem:excitation} is proved from the discounted-occupancy identity: the whitened cross-term bound $\norm{\bbH^{-1/2}\bbC\bbH^{-1/2}}_2\leq\gamma^{-1/2}$ is derived, so TD stability reduces to positive feature covariance, a condition analogous to persistent-excitation conditions in adaptive control~\cite{narendra2012stable}, here defined directly under the discounted occupancy; it is monitorable through $\lambda_{\min}(\hat\bbH_i^{(k)})$, where $\hat\bbH_i^{(k)}:=M_s^{-1}\sum_j\tdbbphi^{\rm raw}(\bbz^{(j)})\tdbbphi^{\rm raw}(\bbz^{(j)})^\top$ is the empirical positive semidefinite feature covariance, by a matrix-Bernstein radius analogous to that of Proposition~\ref{prop:diagnostic}. The reduction does not itself establish positive excitation for the policy class. The route also carries a dimensional cap consistent with Proposition~\ref{prop:mu_M_obstruction}: $\mathrm{tr}(\bbH_i^{\bbtheta})=\mbE\norm{\tdbbphi^{\rm raw}}^2\leq L^2$ forces $\lambda_{\min}(\bbH_i^{\bbtheta})\leq L^2/(m+1)$, so the excitation certificate is at most $(1-\sqrt\gamma)L^2/(m+1)$, while the true $\sigma_{\min}(\bbM_i^{\bbtheta})$ is bounded by $(1+\gamma)L^2/(m+1)$; both bounds are of order $1/m$.

	For some results we use the strictly weaker trajectory-conditional version of (C):

	\begin{assumption}[Set C$'$: Trajectory-conditional stability]\label{asm:Cp}
		For the horizon $K$ in force, there exists $\mu_{\rm traj}>0$ such that
		\feq{
			\inf_{0\leq k<K,\,i\in[n]}\sigma_{\min}(\bbM_i^{\bbtheta^{(k)}})\geq\mu_{\rm traj}
		}
		holds with probability at least $1-\delta_M^{\rm traj}$, where the probability is over the random-feature realization at preprocessing and the realized algorithmic sample path $\{\ccalD_s^{(k)},\ccalD_g^{(k)}\}_{k<K}$ jointly.
	\end{assumption}

	When $K$ is the horizon prescribed by Theorem~\ref{thm:sample_complexity_fixed} for a target accuracy $\epsilon$, an asymptotic reading of Theorem~\ref{thm:sample_complexity_traj} as $\epsilon\to0$ requires $\mu_{\rm traj}$ and $\delta_M^{\rm traj}$ to be uniform over the horizons under consideration; otherwise the constants hidden in the rate may deteriorate with the requested accuracy.

	The empirical proxy $\sigma_{\min}(\bbM_i^{(k)})$ (Definition~\ref{def:lstd}) yields, through matrix-Bernstein concentration, a simultaneous lower confidence bound on $\sigma_{\min}(\bbM_i^{\bbtheta^{(k)}})$ along the realized trajectory in the false-certification sense of Proposition~\ref{prop:diagnostic}; it certifies neither (C$'$) for future runs nor the uniform infimum (C). The paper provides both a (C)-based theorem and a (C$'$)-based companion theorem.

	\subsection{Regularized LSTD}

	\begin{definition}[LSTD matrices]\label{def:lstd}
		With $\bbz^{(j)}=(\bbs_{\ccalN_i^{\kappa+1}}^{(j)},\bba_{\ccalN_i^{\kappa+1}}^{(j)})$ and the one-step successor $\bbz^{\prime(j)}=(\bbs_{\ccalN_i^{\kappa+1}}^{\prime(j)},\bba_{\ccalN_i^{\kappa+1}}^{\prime(j)})$ from $\ccalD_s^{(k)}$,
		\feq{
			\bbM_i^{(k)}&:=\tfrac{1}{M_s}\!\sum_{j=1}^{M_s}\!\tdbbphi_{i,\kappa}^{\rm raw}(\bbz^{(j)})\\
			&\qquad\times\bigl(\tdbbphi_{i,\kappa}^{\rm raw}(\bbz^{(j)})\!-\!\gamma\tdbbphi_{i,\kappa}^{\rm raw}(\bbz^{\prime(j)})\bigr)^{\!\top}\!\!,\\
			\bbb_i^{(k)}&:=\tfrac{1}{M_s}\!\sum_{j=1}^{M_s}r_i^{(j)}\,\tdbbphi_{i,\kappa}^{\rm raw}(\bbz^{(j)}).
		}
		Fix a deterministic threshold $\tau_{\rm sv}\geq0$ and define the regularized LSTD solution by $\hat\bbw_i^{(k),\rm raw}:=(\bbM_i^{(k)}+\lambda\bbI)^{-1}\bbb_i^{(k)}$ if $\sigma_{\min}(\bbM_i^{(k)}+\lambda\bbI)>\tau_{\rm sv}$ and $\hat\bbw_i^{(k),\rm raw}:=\bbzero$ otherwise; $\tau_{\rm sv}=0$ reduces to the singularity-only fallback. The same threshold is used in Algorithm~\ref{alg:cdcpg}, line~5, so the estimator analyzed below coincides on every sample path with the implemented one. The uniform-conditioning theorems assume $0\leq\tau_{\rm sv}<\mu_M/2$ (and $\tau_{\rm sv}<\mu_{\rm traj}/2$ in the trajectory variant): on the favorable events collected in the Randomness subsection below, $\sigma_{\min}(\bbM_i^{(k)}+\lambda\bbI)\geq\mu_M/2>\tau_{\rm sv}$, so the threshold test passes and the inverse is used. On every path, the norm projection of Algorithm~\ref{alg:cdcpg}, line~6, gives $\norm{\hat\bbw_i^{(k)}}\leq W^\ast$ and hence $\abs{\hat Q_i^{(k)}}\leq LW^\ast$, with $W^\ast$ the deterministic radius fixed in Lemma~\ref{lem:bounded_target} below; only these bounds are used off the favorable events.
	\end{definition}

	\subsection{Population LSTD error without conditioning amplification}

	Fix $\bbtheta$ and the feature realization, write $\pi=\pi_\bbtheta$ and $\nu=\nu^\pi$, and define the successor operator $(\ccalP^\pi f)(\bbs,\bba):=\mbE[f(\bbS',\bbA')\mid\bbS=\bbs,\bbA=\bba]$ under the one-step law of (D1). Let $\ccalV_i:=\{\bbz\mapsto\inner{\tdbbphi_{i,\kappa}^{\rm raw}(\bbz)}{\bbw}:\bbw\in\mbR^{m+1}\}$, viewed as a finite-dimensional (hence closed) subspace of $\ccalL^2(\nu)$ through $\bbZ=\bbZ(\bbs,\bba)$, and let $\Pi_i$ be the orthogonal projection onto $\ccalV_i$ in $\ccalL^2(\nu)$.

	\begin{lemma}[Projected Bellman characterization and population error]\label{lem:pop_L2}
		Under Sets~A--B and (D1) with $\kappa\geq\max\{1,\kappa_\pi+1\}$, conditional on $E_M\cap E_{\rm rff}$, for every $i\in[n]$ and $\bbtheta\in\Theta_0$:
		\begin{enumerate}
			\item[(a)] $\gamma\ccalP^\pi$ is well defined on $\ccalL^2(\nu)$ equivalence classes and $\norm{\gamma\ccalP^\pi f}_{\ccalL^2(\nu)}\leq\sqrt\gamma\,\norm{f}_{\ccalL^2(\nu)}$;
			\item[(b)] the population LSTD solution $\bbw_i^\ast:=(\bbM_i^{\bbtheta})^{-1}\bbb_i^{\bbtheta}$ is the unique coefficient for which $q_i^\ast:=\inner{\tdbbphi_{i,\kappa}^{\rm raw}}{\bbw_i^\ast}$ solves the projected Bellman equation $q_i^\ast=\Pi_i(r_i+\gamma\ccalP^\pi q_i^\ast)$, and $\bbH_i^{\bbtheta}\succ\bbzero$;
			\item[(c)] the population prediction error obeys
			\feq{
				\norm{Q_i^\pi-q_i^\ast}_{\ccalL^2(\nu)}^2\leq\frac{\dist_{\ccalL^2(\nu)}(Q_i^\pi,\ccalV_i)^2}{1-\gamma}\leq\frac{\Bk^2}{1-\gamma};\label{eq:pop_L2}
			}
			\item[(d)] the coefficient and pointwise comparisons
			\feq{
				\norm{\bbw_i^\ast-\bbu_{i,\kappa}^\pi}&\leq\frac{(1+\gamma)L\,\Bk}{\mu_M},\\
				\sup\abs{q_i^\ast-Q_i^\pi}&\leq(1+C')\Bk,\label{eq:pop_sup}
			}
			hold with $C':=(1+\gamma)L^2/\mu_M$;
			\item[(e)] $\norm{\bbw_i^\ast}\leq\norm{\bbb_i^{\bbtheta}}/\mu_M\leq L\bbarr/\mu_M$.
		\end{enumerate}
	\end{lemma}

	The proof is in Appendix~\ref{app:pop_L2}. Item (c) is the statement the actor analysis uses: it bounds the prediction error in the sampling measure and contains no $\mu_M^{-1}$ multiplier. Rescaling the features changes coefficient conditioning without changing the span $\ccalV_i$, so a coefficient-norm bound such as (d) can exaggerate the approximation difficulty; the dimensional obstruction of Proposition~\ref{prop:mu_M_obstruction} constrains coefficient conditioning, and thereby the finite-sample budget below, but it does not force a worsening population prediction error. The argument follows the projected Bellman viewpoint of~\cite{tsitsiklis1997analysis}, with the contraction (a) proved for the discounted-occupancy measure of (D1) rather than imported from a stationary-distribution setting.

	\begin{lemma}[Projection radius]\label{lem:bounded_target}
		Fix the projection radius of Algorithm~\ref{alg:cdcpg} as the deterministic constant $W^\ast:=L\bbarr/\mu_M$, which depends only on the primitives $(\bbarr,L,\mu_M)$. Conditional on $E_M$ alone, every population LSTD solution satisfies $\norm{\bbw_i^\ast}=\norm{(\bbM_i^{\bbtheta})^{-1}\bbb_i^{\bbtheta}}\leq\norm{\bbb_i^{\bbtheta}}/\mu_M\leq L\bbarr/\mu_M=W^\ast$ (this is Lemma~\ref{lem:pop_L2}(e), which uses only $E_M$), so the norm projection onto $\{\norm{\bbw}\leq W^\ast\}$ does not increase the distance to $\bbw_i^\ast$.
	\end{lemma}

	Off the favorable events, the analysis uses only the deterministic bounds $\norm{\hat\bbw_i^{(k)}}\leq W^\ast$ (norm projection) and hence $\abs{\hat Q_i^{(k)}}\leq LW^\ast$.

	\begin{theorem}[LSTD concentration, fixed-$(i,k)$]\label{thm:lstd_error}
		Under Sets~A--D, fix $i\in[n]$ and $k<K$, $\delta\in(0,1)$, $\lambda\geq0$, and $0\leq\tau_{\rm sv}<\mu_M/2$. If $M_s\geq C_0 L^4\log(4(m+1)/\delta)/\mu_M^2$, then for almost every realization of the pre-iteration history (the $\sigma$-algebra $\ccalF_k^-$ of Section~\ref{sec:randomness}) on $E_M\cap E_{\rm rff}$, with conditional probability at least $1-\delta$ over the fresh critic dataset $\ccalD_s^{(k)}$,
		\feq{
			&\norm{\hat\bbw_i^{(k),\rm raw}-\bbw_i^\ast(\bbtheta^{(k)})}\\
			&\quad\leq\tfrac{2C_v L(\bbarr+(1+\gamma)LW^\ast)}{\mu_M}\sqrt{\tfrac{\log(4(m+1)/\delta)}{M_s}}+\tfrac{2\lambda W^\ast}{\mu_M}.\label{eq:lstd_error}
		}
		The uniform version over all $i\in[n]$, $k<K$ holds with $\delta$ replaced by $\delta/(nK)$. Here $C_0,C_v>0$ are universal constants; $C_0=512$ and $C_v=\sqrt2$ suffice.
	\end{theorem}

	The proof (Appendix~\ref{app:lstd_concentration}) applies the rectangular matrix-Bernstein inequality~\cite[Thm.~6.1.1]{tropp2015introduction} to the centered i.i.d.\ summands of $\bbM_i^{(k)}-\bbM_i^{\bbtheta^{(k)}}$ through the Hermitian dilation, and the Hilbert-space Hoeffding-type inequality of Pinelis~\cite[Thm.~3.5]{pinelis1994optimum} to the centered vector $\bbb_i^{(k)}-\bbM_i^{(k)}\bbw_i^\ast$, which gives a dimension-free vector deviation; the displayed $\log(m+1)$ factor is retained in the vector term only as a harmless common envelope. Since $\norm{\bbw_i^\ast}\leq W^\ast$ (Lemma~\ref{lem:bounded_target}), nonexpansiveness of the norm projection gives $\norm{\hat\bbw_i^{(k)}-\bbw_i^\ast}\leq\norm{\hat\bbw_i^{(k),\rm raw}-\bbw_i^\ast}$, so the bound transfers verbatim to the projected coefficient used by the algorithm. We write $E_{\rm lstd}^{i,k}$ for the event of Theorem~\ref{thm:lstd_error} at confidence $\delta_\star$, the per-iteration confidence used in Theorem~\ref{thm:joint_optimal}.

	\begin{remark}[Role of the ridge parameter]\label{rmk:lambda_nonsymmetric}
		Under discounted-occupancy sampling the population TD matrix is accretive (Lemma~\ref{lem:excitation}(a)), so the shift $\lambda\bbI$ cannot decrease $\sigma_{\min}(\bbM_i^{\bbtheta}+\lambda\bbI)$ below $\sigma_{\min}(\bbM_i^{\bbtheta})$, and on the concentration event $\norm{\bbM_i^{(k)}-\bbM_i^{\bbtheta}}\leq r$ one has $\sigma_{\min}(\bbM_i^{(k)}+\lambda\bbI)\geq\mu_M-r$ for every $\lambda\geq0$. The empirical matrix need not be accretive on every sample path, so no unconditional monotonicity in $\lambda$ is claimed for $\bbM_i^{(k)}+\lambda\bbI$ itself. In the bound~\eqref{eq:lstd_error}, $\lambda$ contributes only the additive bias $2\lambda W^\ast/\mu_M$, which arises from the exact identity $(\bbM_i^{(k)}+\lambda\bbI)(\hat\bbw_i^{(k),\rm raw}-\bbw_i^\ast)=(\bbb_i^{(k)}-\bbM_i^{(k)}\bbw_i^\ast)-\lambda\bbw_i^\ast$; a fixed positive ridge therefore leaves a bias that increasing the sample count cannot remove, and a large ridge can worsen policy evaluation even while it stabilizes numerical inversion. Accordingly the parameter selection of Theorem~\ref{thm:joint_optimal} takes $\lambda^\ast=0$; a positive $\lambda$ is retained in Algorithm~\ref{alg:cdcpg} as an implementation option whose bias is carried explicitly, and near-singularity is handled by the singular-value threshold of Algorithm~\ref{alg:cdcpg}, line~5, whose tolerance $\tau_{\rm sv}$ lies below the proven margin $\mu_M/2$ (Definition~\ref{def:lstd}) and may in practice be tied to the concentration radius $r_M$ of Proposition~\ref{prop:diagnostic}.
	\end{remark}

	\subsection{Randomness and events}\label{sec:randomness}

	The analysis distinguishes three favorable events: $E_M$ (TD-stability, Assumption~\ref{asm:C}, probability $\geq 1-\delta_M(m)$ over the random-feature draw), $E_{\rm rff}$ (RFF uniform approximation, probability $\geq 1-\delta_{\rm rff}$), and $E_{\rm lstd}^{i,k}$ (per-iteration LSTD concentration, probability $\geq 1-\delta_\star$ over $\ccalD_s^{(k)}$ conditional on past randomness and $E_M\cap E_{\rm rff}$; formalized after Theorem~\ref{thm:lstd_error}). We use two filtrations to make conditional sampling claims precise. Let $\ccalF_{\rm RFF}:=\sigma\bigl(\{\bbomega_{i,\ell},b_{i,\ell}\}_{i,\ell}\bigr)$ collect the preprocessing randomness. Then
	\feq{
		\ccalF_k^-&:=\sigma\bigl(\ccalF_{\rm RFF},\,\bbtheta^{(0)},\,\{\ccalD_s^{(j)},\ccalD_g^{(j)}\}_{j<k},\,\bbtheta^{(k)}\bigr),\\
		\ccalF_k^Q&:=\sigma\bigl(\ccalF_k^-,\,\ccalD_s^{(k)},\,\{\hat\bbw_i^{(k)}\}_{i\in[n]}\bigr).
	}
	Under (D1), $\ccalD_s^{(k)}$ is i.i.d.\ from the joint one-step law $d^{\pi_{\bbtheta^{(k)}}}\otimes\pi_{\bbtheta^{(k)}}\otimes\mbP\otimes\pi_{\bbtheta^{(k)}}$ generating $(\bbs,\bba,\bbs',\bba')$, conditional on $\ccalF_k^-$, and $\ccalD_g^{(k)}$ is i.i.d.\ from $d^{\pi_{\bbtheta^{(k)}}}\otimes\pi_{\bbtheta^{(k)}}$ conditional on $\ccalF_k^Q$. Theorems below are conditional on the preprocessing favorable events $E_M\cap E_{\rm rff}\in\ccalF_{\rm RFF}$, with the failure budget of the per-iteration events absorbed in expectation; the expectation guarantees do not require every per-iteration event to succeed simultaneously, so no union bound over iterations is taken in the main theorems.

	\subsection{Critic error under local parameter selection}

	Here and below, $\mbE_{\rm alg}$ denotes the conditional expectation $\mbE[\,\cdot\mid\ccalF_k^-]$ over the current critic dataset $\ccalD_s^{(k)}$, evaluated at a pre-iteration history in $E_M\cap E_{\rm rff}$; bounds conditional only on the preprocessing events and $\bbtheta^{(k)}$ follow by the tower property. Set $A_Q:=LW^\ast+Q_{\max}$, the deterministic pointwise bound on $\abs{\hat Q_i^{(k)}-Q_i^\pi}$, and $B_Y:=L(\bbarr+(1+\gamma)LW^\ast)$.

	\begin{theorem}[Critic error]\label{thm:joint_optimal}
		Under Sets~A--D with $\kappa\geq\max\{1,\kappa_\pi+1\}$, $m$ satisfying~\eqref{eq:rff_feasibility} at confidence $\delta_{\rm rff}$, and the threshold of Definition~\ref{def:lstd} satisfying $0\leq\tau_{\rm sv}<\mu_M/2$, fix $i\in[n]$ and $k<K$, and choose $\epsilon_Q>0$, $\delta_\star\in(0,1)$. With $\ell_\star:=\log(4(m+1)/\delta_\star)$ and the parameter choices
		\feq{
			\lambda^\ast=0,\qquad
			M_s\geq C_s\max\Bigl\{\frac{L^4\ell_\star}{\mu_M^2},\ \frac{L^2B_Y^2\ell_\star}{\mu_M^2\epsilon_Q^2}\Bigr\}\label{eq:Ms_choice}
		}
		for a universal constant $C_s$ (one may take $C_s=512$), conditional on $E_M\cap E_{\rm rff}$:
		\begin{enumerate}
			\item[(a)] for almost every pre-iteration history $\ccalF_k^-$ in $E_M\cap E_{\rm rff}$, with conditional probability at least $1-\delta_\star$ over $\ccalD_s^{(k)}$, the deployed critic $\hat Q_i^{(k)}=\inner{\tdbbphi_{i,\kappa}^{\rm raw}}{\hat\bbw_i^{(k)}}$ satisfies $\sup_{\bbz}\abs{\hat Q_i^{(k)}(\bbz)-q_i^\ast(\bbz)}\leq\epsilon_Q$, hence
			\begin{equation}\label{eq:joint_optimal_bound}
				\begin{aligned}
					\norm{\hat Q_i^{(k)}-Q_i^\pi}_{\ccalL^2(\nu^\pi)}^2&\leq2\epsilon_Q^2+\frac{2\Bk^2}{1-\gamma},\\
					\sup_{\bbs,\bba}\abs{\hat Q_i^{(k)}-Q_i^\pi}&\leq\epsilon_Q+(1+C')\Bk;
				\end{aligned}
			\end{equation}
			\item[(b)] if moreover $\delta_\star\leq\epsilon_Q^2/A_Q^2$, then
			\feq{
				\mbE_{\rm alg}\norm{\hat Q_i^{(k)}-Q_i^\pi}_{\ccalL^2(\nu^\pi)}^2\leq 3\epsilon_Q^2+\frac{2\Bk^2}{1-\gamma}.\label{eq:critic_L2}
			}
		\end{enumerate}
		The total unconditional failure probability of (a) is at most $\delta_M+\delta_{\rm rff}+\delta_\star$; a uniform version over $i\in[n]$ and $k<K$ follows by replacing $\delta_\star$ with $\delta_\star/(nK)$. If a positive implementation value $\lambda\leq\mu_M\epsilon_Q/(4LW^\ast)$ is used instead of $\lambda^\ast=0$, the additional bias is at most $\epsilon_Q/2$ and the conclusions hold with $\epsilon_Q$ replaced by $\tfrac32\epsilon_Q$.
	\end{theorem}

	The proof is in Appendix~\ref{app:joint_optimal}. The $\ccalL^2$ bound~\eqref{eq:critic_L2} is the quantity the actor analysis uses; it carries the population residual $\Bk^2/(1-\gamma)$ without any conditioning multiplier, while the conditioning threshold $\mu_M$ enters through the sample budget~\eqref{eq:Ms_choice}, the radius $W^\ast$, and the deterministic bound $A_Q$. Three coefficient vectors play distinct roles. The representation coefficient $\bbu_{i,\kappa}^\pi$ establishes that the feature space can approximate $Q_i^\pi$; it need not solve the population TD equations. The population coefficient $\bbw_i^\ast$ solves those equations, and the estimated coefficient $\hat\bbw_i^{(k)}$ is computed from samples and then projected. Lemma~\ref{lem:pop_L2} compares $Q_i^\pi$ with the population prediction $q_i^\ast$, and Theorem~\ref{thm:lstd_error} bounds the additional error from estimating $\bbw_i^\ast$.

	\subsection{Distributed policy gradient}

	Recall from Section~\ref{sec:pre} the aggregation radius $\kappa_c=\kappa+\kappa_\pi$ and the standing inequality $\kappa\geq\kappa_\pi+1$. The gradient estimator at agent $i$, derived from the per-agent policy-gradient identity~\eqref{eq:pg_theorem_agent}, is
	\begin{equation}\label{eq:gradient_estimator}
		\begin{aligned}
			\hat\bbg_i^{(k)}&:=\tfrac{1}{(1-\gamma)M_g}\!\sum_{j=1}^{M_g}\!\hat Q_{i,\kappa_c}^{(k)}(\bbs^{(j)},\bba^{(j)})\\
			&\qquad\quad\cdot\nabla_{\bbtheta_i}\log\pi_i(\bba_i^{(j)}\mid\bbs_{\ccalN_i^{\kappa_\pi}}^{(j)}),
		\end{aligned}
	\end{equation}
	with $\hat Q_{i,\kappa_c}^{(k)}(\bbs,\bba):=\sum_{\ell\in\ccalN_i^{\kappa_c}}\hat Q_\ell^{(k)}(\bbs_{\ccalN_\ell^{\kappa+1}},\bba_{\ccalN_\ell^{\kappa+1}})$ and $\hat Q_\ell^{(k)}(\bbz):=\tdbbphi_{\ell,\kappa}^{\rm raw}(\bbz)^\top\hat\bbw_\ell^{(k)}$, so that $\abs{\hat Q_\ell^{(k)}}\leq LW^\ast$ deterministically by the norm projection. No value clipping is applied; the deployed critic is the linear function analyzed in Theorem~\ref{thm:joint_optimal}.

	\begin{lemma}[Gradient bias]\label{lem:grad_bias}
		Under the conditions of Theorem~\ref{thm:joint_optimal}(b) and the graph-tail assumption $\tilde\delta:=(\Delta-1)\rho<1$, conditional on $E_M\cap E_{\rm rff}$,
		\begin{equation}\label{eq:grad_bias_bound}
			\begin{aligned}
				&\norm{\mbE[\hat\bbg_i^{(k)}\mid\ccalF_k^-]-\nabla_{\bbtheta_i}J(\bbtheta^{(k)})}\\
				&\quad\leq\tfrac{GD_{\kappa_c}}{1-\gamma}\Bigl(\sqrt 3\,\epsilon_Q+\sqrt{\tfrac{2}{1-\gamma}}\,\Bk\Bigr)+\epsilon_{\rm agg}(\kappa),
			\end{aligned}
		\end{equation}
		with the aggregation-bias term
		\feq{
			\epsilon_{\rm agg}(\kappa):=\tfrac{Gc}{(1-\gamma)(1-\tilde\delta)}\cdot\tfrac{\Delta}{\Delta-1}\tilde\delta^{\kappa_c+1}.\label{eq:eps_agg_def}
		}
	\end{lemma}

	\begin{lemma}[Critic-conditional bias-squared]\label{lem:bias_squared}
		Let $\bbd_i^{(k)}:=\mbE[\hat\bbg_i^{(k)}\mid\ccalF_k^Q]-\nabla_{\bbtheta_i}J(\bbtheta^{(k)})$ denote the critic-conditional bias and $A_\kappa:=G^2D_{\kappa_c}^2/(1-\gamma)^2$. Under the conditions of Lemma~\ref{lem:grad_bias}, for every $i$ and $k$, almost surely on $E_M\cap E_{\rm rff}$,
		\begin{equation}\label{eq:eps_b_bound}
			\mbE\bigl[\norm{\bbd_i^{(k)}}^2\,\big|\,\ccalF_k^-\bigr]\leq\bar\epsilon_b^2:=2\epsilon_{\rm agg}(\kappa)^2+2A_\kappa\Bigl(3\epsilon_Q^2+\frac{2\Bk^2}{1-\gamma}\Bigr).
		\end{equation}
	\end{lemma}

	Lemma~\ref{lem:bias_squared} bounds the history-conditional expectation of the squared critic-conditional bias by a deterministic quantity, not the square of the unconditional mean bias; its proof (Appendix~\ref{app:bias_squared}) conditions on the critic, applies Cauchy--Schwarz over the at most $D_{\kappa_c}$ neighboring critics, and then invokes~\eqref{eq:critic_L2}.

	\subsection{Algorithm}

	\begin{algorithm}[!t]
		\caption{CDCPG with projected updates}\label{alg:cdcpg}
		\textbf{Input:} $\{\bbtheta_i^{(0)}\}\subseteq\Theta_0$, $\kappa$, $\kappa_c=\kappa+\kappa_\pi$, $m$, $\lambda$, $\eta$, $M_s$, $M_g$, $K$, projection radius $W^\ast$ (Lemma~\ref{lem:bounded_target}), singular-value threshold $\tau_{\rm sv}\geq0$; $K,M_s,M_g\in\mbN$, $\eta>0$, $\lambda\geq0$.\\
		\textbf{Preprocessing:} for each $i\in[n]$, sample $\{(\bbomega_{i,\ell},b_{i,\ell})\}_{\ell=1}^m$ independently.
		\begin{algorithmic}[1]
			\For{$k=0,\ldots,K-1$}
			\State Draw shared $\ccalD_s^{(k)}, \ccalD_g^{(k)}$ under $\pi^{(k)}$.
			\For{each $i\in[n]$ in parallel}
			\State Compute $\bbM_i^{(k)},\bbb_i^{(k)}$ as in Definition~\ref{def:lstd}.
			\State Set $\hat\bbw_i^{(k),\rm raw}\gets(\bbM_i^{(k)}+\lambda\bbI)^{-1}\bbb_i^{(k)}$ if $\sigma_{\min}(\bbM_i^{(k)}+\lambda\bbI)>\tau_{\rm sv}$, else $\bbzero$ (Definition~\ref{def:lstd}; the analysis assumes $\tau_{\rm sv}<\mu_M/2$).
			\State Set $\hat\bbw_i^{(k)}\gets\Pi_{\{\bbw:\norm{\bbw}\leq W^\ast\}}(\hat\bbw_i^{(k),\rm raw})$.
			\State Set $\hat Q_i^{(k)}(\bbz)\gets\tdbbphi_{i,\kappa}^{\rm raw}(\bbz)^\top\hat\bbw_i^{(k)}$.
			\State Gather $\{\hat\bbw_\ell^{(k)}\}_{\ell\in\ccalN_i^{\kappa_c}}$.
			\State Compute $\hat\bbg_i^{(k)}$ via~\eqref{eq:gradient_estimator}.
			\State $\bbtheta_i^{(k+1)}\gets\Pi_{\Theta_{0,i}}(\bbtheta_i^{(k)}+\eta\hat\bbg_i^{(k)})$.
			\EndFor
			\EndFor
		\end{algorithmic}
		\textbf{Output:} $\bbtheta^{(R)}$ with $R\sim\mathrm{Unif}\{0,\ldots,K-1\}$ drawn independently of the run, so that the trajectory-average guarantees below bound $n^{-1}\mbE\norm{\ccalG_\eta(\bbtheta^{(R)})}^2$ directly; the final iterate $\bbtheta^{(K)}$ may be returned in practice.
	\end{algorithm}

	Per-agent computation is $\ccalO(M_s(mD_\kappa d_S+m^2)+m^3+M_g D_{\kappa_c}mD_\kappa d_S)$, where the $mD_\kappa d_S$ factors account for evaluating the random-feature maps at each sample; inter-agent algorithmic communication, excluding shared-oracle sample delivery, is $\ccalO(D_{\kappa_c}m)$ scalars per iteration for weight exchange; sample delivery additionally contributes $\ccalO((M_s+M_g)D_{\kappa_c+\kappa+1}(d_S+d_A))$ coordinates per agent per iteration, up to constant factors for current and successor tuples, after a one-time preprocessing dissemination of feature seeds (with a public master seed and $s_i=\mathrm{Hash}(i,\text{master})$, no exchange is needed; otherwise each agent broadcasts one seed to its $\kappa_c$-neighborhood, so each agent receives $\ccalO(D_{\kappa_c})$ seed words; explicit exchange of the maps would instead cost $\ccalO(mD_\kappa d_S)$ scalars per neighbor pair). Under the shared-sample model (D1), each agent needs only the coordinates of the shared samples within its $(\kappa_c{+}\kappa{+}1)$-neighborhood, that is, $\ccalO(D_{\kappa_c+\kappa+1}(d_S+d_A))$ scalars per sample, so the oracle admits a local-dissemination realization; the ``shared'' oracle means that agents reuse compatible coordinates of the same global sample, not that each agent's number of observations is physically reduced by a factor $1/n$. With non-shared sampling, evaluating neighboring critics would additionally require coordinate exchange. Preprocessing stores $\ccalO(nm)$ random-feature pairs, corresponding to $\ccalO(nm D_\kappa d_S)$ scalar frequency entries plus $\ccalO(nm)$ phase entries. The stated arithmetic count covers feature construction from evaluated local drift and reward values, LSTD formation and solution, and critic evaluation; local model, score, and parameter-projection evaluation costs are treated separately, and the communication count records coefficient payloads delivered to each aggregation neighborhood, so a physical multi-hop implementation must also account for routing.

	\section{Convergence and Sample Complexity}\label{sec:convergence}

	\subsection{Main convergence bound}

	\begin{theorem}[Projected-gradient-mapping convergence]\label{thm:main_convergence}
		Under Sets~A--D, with the critic hypotheses of Theorem~\ref{thm:joint_optimal}(b) in force at every iteration (in particular $\lambda=0$ and the batch and confidence choices there), with $0\leq\tau_{\rm sv}<\mu_M/2$, and with $0<\eta\leq 1/(4L_J)$ for a positive smoothness bound $L_J$, Algorithm~\ref{alg:cdcpg} satisfies, for almost every feature realization in $E_M\cap E_{\rm rff}$,
		\begin{equation}\label{eq:convergence}
			\begin{aligned}
				\frac{1}{nK}\sum_{k=0}^{K-1}\mbE\norm{\ccalG_\eta(\bbtheta^{(k)})}^2
				&\leq\frac{16\Delta_J}{3n\eta K}+\frac{14}{3}\Bigl(\frac{\sigma_g^2}{M_g}+\bar\epsilon_b^2\Bigr),
			\end{aligned}
		\end{equation}
		where $\sigma_g^2:=G^2 D_{\kappa_c}^2 (LW^\ast)^2/(1-\gamma)^2=A_\kappa(LW^\ast)^2$ (the per-sample gradient magnitude is bounded by $GD_{\kappa_c}LW^\ast/(1-\gamma)$ under the projected linear critic) and $\bar\epsilon_b^2$ is the deterministic bias bound of Lemma~\ref{lem:bias_squared}.
	\end{theorem}

	The proof (Appendix~\ref{app:main_convergence}) argues through the stochastic projected mapping, following~\cite{ghadimi2016mini}; the projection prevents the classical unconstrained noise accounting, and the error correctly enters at order $\eta$ in a pathwise inequality. The result is an expected-stationarity bound for the uniformly selected iterate $\bbtheta^{(R)}$ of Algorithm~\ref{alg:cdcpg} under the favorable preprocessing events; it is not a last-iterate, high-probability, or global-optimality statement.

	\begin{corollary}[Step-size choice]\label{cor:opt_lr}
		With $\eta=1/(4L_J)$ and the bound $\bar\epsilon_b^2$ of Lemma~\ref{lem:bias_squared} (which requires $\lambda=0$ and the critic budgets of Theorem~\ref{thm:joint_optimal}(b)), the right-hand side of~\eqref{eq:convergence} is at most
		\begin{equation}\label{eq:explicit_bound}
			\begin{aligned}
				&\frac{64L_J\Delta_J}{3nK}+\frac{14\sigma_g^2}{3M_g}+28A_\kappa\epsilon_Q^2\\
				&\quad+\frac{28}{3}\epsilon_{\rm agg}(\kappa)^2+\frac{112A_\kappa}{3(1-\gamma)}\bigl(B_{\rm loc}(\kappa)^2+B_{\rm rff}(m)^2\bigr),
			\end{aligned}
		\end{equation}
		using $\Bk^2\leq2(B_{\rm loc}^2+B_{\rm rff}^2)$. With $\Delta_J\leq 2nQ_{\max}=\ccalO(n)$ and $L_J$ $n$-independent per (D2)--(D3) (Appendix~\ref{app:LJ_proof}), the first term is $\ccalO(L_JQ_{\max}/K)$, independent of $n$ in its leading prefactor.
	\end{corollary}

	\subsection{Fixed-locality complexity}

	In this regime $\kappa$ and $m$ are problem primitives. Define the structural floor (a residual term of the upper bound, not a lower bound on the attainable stationarity)
	\begin{equation}\label{eq:efloor_def}
		\begin{aligned}
			\Efloor(\kappa,m):=\Cfloor\bigl[&\Efloor^{\rm trunc}(\kappa)+\Efloor^{\rm rff}(\kappa,m)\\&+\Efloor^{\rm graph}(\kappa)\bigr],
		\end{aligned}
	\end{equation}
	with $\Cfloor:=112/3$ and the three components
	\feq{
		\Efloor^{\rm trunc}(\kappa)&:=\frac{A_\kappa B_{\rm loc}(\kappa)^2}{1-\gamma}=\frac{A_\kappa\gamma^2\tilde c^2\rho^{2(\kappa+1)}}{1-\gamma},\\
		\Efloor^{\rm rff}(\kappa,m)&:=\frac{A_\kappa B_{\rm rff}(m)^2}{1-\gamma}=\frac{A_\kappa\gamma^2\bbarr^2\epsilon_P(m)^2}{(1-\gamma)^3},\\
		\Efloor^{\rm graph}(\kappa)&:=\epsilon_{\rm agg}(\kappa)^2.
	}
	The components are interpretable: $\Efloor^{\rm trunc}$ is the omitted distant state in the local continuation value (controlled by value decay), $\Efloor^{\rm rff}$ is the finite-feature approximation of the local Gaussian kernel (controlled by the global signed-kernel $\ccalL^1$ bound), and $\Efloor^{\rm graph}$ is the omitted distant rewards in the local policy gradient (controlled by score cancellation and graph-shell summation). Conditioning enters the sample budget, the radius $W^\ast$, and $\sigma_g^2$, but multiplies none of these population components.

	\begin{theorem}[Sample complexity, uniform TD-stability]\label{thm:sample_complexity_fixed}
		Under Sets~A--D with $\kappa,m$ fixed problem primitives such that $m$ satisfies the RFF feasibility relation~\eqref{eq:rff_feasibility} of Theorem~\ref{thm:rff_approx} at confidence $\delta_{\rm rff}$ and accuracy $\epsilon_P(m,\delta_{\rm rff},n)$, $\Delta\geq 2$, $\tilde\delta=(\Delta-1)\rho<1$, $G>0$, $\delta_M+\delta_{\rm rff}<1$, and $0\leq\tau_{\rm sv}<\mu_M/2$, let $\bbarr>0$ and $L_J>0$ denote the positive upper bounds used as primitives. For any $\epsilon>0$, Algorithm~\ref{alg:cdcpg} with the schedule
		\feq{
			\epsilon_Q^2&=\frac{\epsilon}{84A_\kappa},\qquad
			\delta_\star=\min\{\tfrac12,\epsilon_Q^2/A_Q^2\},\qquad\lambda^\ast=0,\\
			M_s&=\Bigl\lceil C_s\max\Bigl\{\frac{L^4\ell_\star}{\mu_M^2},\frac{L^2B_Y^2\ell_\star}{\mu_M^2\epsilon_Q^2}\Bigr\}\Bigr\rceil=\widetilde\ccalO\Bigl(\frac{L^6W^{\ast2}A_\kappa}{\mu_M^2\,\epsilon}\Bigr),\\
			M_g&=\max\Bigl\{1,\Bigl\lceil\frac{42\sigma_g^2}{\epsilon}\Bigr\rceil\Bigr\},\qquad\eta=\frac{1}{4L_J},\\
			K&=\max\Bigl\{1,\Bigl\lceil\frac{64L_J\Delta_J}{n\epsilon}\Bigr\rceil\Bigr\},
		}
		satisfies the following guarantee (larger budgets or a smaller critic tolerance preserve the inequality but not necessarily the displayed cost bound; all $\widetilde\ccalO$ estimates concern the limit $\epsilon\to0$ with the positive primitives fixed, whereas the displayed integer budgets are valid for every $\epsilon>0$). For almost every feature realization in $E_M\cap E_{\rm rff}$, an event of probability at least $1-\delta_M-\delta_{\rm rff}$ over the preprocessing random-feature draw, the expectation over the subsequent algorithmic sampling randomness $(\{\ccalD_s^{(k)},\ccalD_g^{(k)}\}_{k<K})$ satisfies
		\feq{
			\frac{1}{nK}\sum_{k<K}\mbE\norm{\ccalG_\eta(\bbtheta^{(k)})}^2\leq\epsilon+\Efloor(\kappa,m),\label{eq:neighborhood_plus_floor}
		}
		where the per-iteration LSTD failure events $(E_{\rm lstd}^{i,k})^c$ are absorbed into the $\epsilon$-budget through the choice of $\delta_\star$. A stronger high-probability guarantee over all algorithmic samples would require a separate failure-budget accounting over $K,n$ and the LSTD and gradient events. With this schedule the shared-oracle complexity is
		\feq{
			N_{\rm global}&=K(M_s+M_g)\notag\\&=\widetilde\ccalO\Bigl(\frac{L_JQ_{\max}}{\epsilon}\Bigl[\frac{L^6W^{\ast2}A_\kappa}{\mu_M^2\epsilon}+\frac{\sigma_g^2}{\epsilon}\Bigr]\Bigr)=\widetilde\ccalO(1/\epsilon^2).
		}
		Each agent uses the relevant local coordinates of every shared sample under the oracle (D1), so $N_{\rm global}$ is also the number of local sample observations per agent. The prefactor depends polynomially on $L_J,D_{\kappa_c},W^\ast,1/\mu_M(m)$, with degree at most $4$ in $1/\mu_M$ through $W^{\ast2}/\mu_M^2$; $D_\kappa$ and $\tilde g_\alpha$ enter only through the floor and the feasibility of $m$.
	\end{theorem}

	\begin{remark}[Lower-bound comparison and excess-to-floor reading]\label{rem:excess}
		The lower-bound comparison invoked for the $\epsilon$-exponent is with the smooth nonconvex benchmark of~\cite{arjevani2023lower}, whose oracle, namely unbiased bounded-variance stochastic gradients without constraints, differs from the present shared generative model with a biased localized estimator and projected updates; the agreement is at the exponent level only. Explicitly, the cited bound is $\Omega(\epsilon_g^{-4})$ for driving $\norm{\nabla F}\leq\epsilon_g$, and becomes $\Omega(\epsilon^{-2})$ under the reparameterization $\epsilon=\epsilon_g^2$ matching the squared-stationarity convention here.

		\emph{Excess-to-floor reading.} Theorem~\ref{thm:sample_complexity_fixed} bounds the averaged stationarity by $\epsilon+\Efloor(\kappa,m)$ with $\epsilon$ the excess above the floor: for a total target $\tau_{\rm tot}>\Efloor(\kappa,m)$, invoking it with $\epsilon=\tau_{\rm tot}-\Efloor(\kappa,m)$ gives $N_{\rm global}=\widetilde\ccalO((\tau_{\rm tot}-\Efloor(\kappa,m))^{-2})$, and the $\widetilde\ccalO(\tau_{\rm tot}^{-2})$ interpretation is justified, in particular, when $\Efloor(\kappa,m)\leq\vartheta\tau_{\rm tot}$ for a fixed $\vartheta<1$, with the other prefactors held fixed or separately controlled. There is no guarantee that an arbitrary $\tau_{\rm tot}>0$ is attainable with fixed features and locality; that would require quantitative control of the conditioning and neighborhood-dependent constants as $m$ and $\kappa$ vary, which is not claimed.
	\end{remark}

	The proof of Theorem~\ref{thm:sample_complexity_fixed} is in Appendix~\ref{app:sample_complexity}. An implementable conservative choice replaces the unknown gap $\Delta_J$ by its bound $2nQ_{\max}$, giving $K=\max\{1,\lceil128L_JQ_{\max}/\epsilon\rceil\}$ with the same order. If $G=0$ the policy gradient vanishes identically and the division by $A_\kappa$ is unnecessary; the theorem is stated for the nondegenerate regime.

	\begin{remark}[Network-size scaling of the prefactor]\label{rem:n_scaling}
		Under fixed $\kappa$, $D_{\kappa_c}\leq C_D\Delta^{\kappa_c}$, with $C_D$ an absolute constant of bounded-degree neighborhood growth, is bounded independently of $n$, and no explicit polynomial factor in $n$ appears in the leading prefactor of $N_{\rm global}$ provided that $L_J$, $\mu_M^{-1}$ (hence $W^\ast$ and $\sigma_g^2$), the feature-scale constant $\tilde g_\alpha$, and the preprocessing failure probabilities $\delta_M,\delta_{\rm rff}$ are all uniform in $n$; network-size-uniform bounds therefore require uniform conditioning and smoothness constants across the family of networks considered, not merely fixed-radius neighborhoods. The simultaneous RFF feasibility relation~\eqref{eq:rff_feasibility} retains a logarithmic dependence on $n$, and each shared-oracle sample carries $\ccalO(n(d_S+d_A))$ coordinates even though only local coordinates need be delivered per agent. Should any of these constants scale with $n$, for example $L_J$, whose descent factor $K=\ccalO(1+L_J\Delta_J/(n\epsilon))$ transmits the scaling to $N_{\rm global}$, or $\mu_M$, which the union-bound discussion of Assumption Set~C allows to deteriorate with $n$, the prefactor inherits that dependence. The cap $D_{\kappa_c}\leq n$ becomes binding only when the realized neighborhood growth reaches the network size, i.e.\ when $\kappa_c$ is of the order of the graph diameter; the crude envelope $\Delta^{\kappa_c}\geq n$ does not by itself imply this (on a path graph $D_\kappa=\min\{n,2\kappa+1\}$ grows linearly while $2^{\kappa}$ exceeds $n$ already at $\kappa=\log_2 n$). In that diameter-scale regime the prefactor carries an explicit factor of $n^2$ through $D_{\kappa_c}^2$.
		Moreover, under the present simultaneous union-bound RFF certificate, maintaining fixed approximation accuracy $\epsilon_P$, fixed confidence $\delta_{\rm rff}$, and fixed local primitives as $n$ grows requires $m=\Omega(\log n)$ through~\eqref{eq:rff_feasibility} (a requirement of this particular union-bound certificate, not an information-theoretic necessity); combined with Proposition~\ref{prop:mu_M_obstruction}, this gives $\mu_M^{-1}=\Omega(m)=\Omega(\log n)$. Hence the present proof cannot maintain a fully $n$-uniform prefactor under these fixed-accuracy conditions, and at least polylogarithmic hidden $n$-dependence should be expected.
	\end{remark}

	\begin{theorem}[Sample complexity, trajectory-conditional]\label{thm:sample_complexity_traj}
		Retain the fixed-locality, RFF-feasibility, positivity, and parameter-range hypotheses of Theorem~\ref{thm:sample_complexity_fixed} (Sets~A, B, D with $\kappa,m$ fixed, the feasibility relation~\eqref{eq:rff_feasibility} at $\delta_{\rm rff}<1$, $\Delta\geq2$, $\tilde\delta<1$, $G>0$, and $\bbarr,L_J>0$), except Set~C and the condition involving $\delta_M$. Replace $\mu_M$ throughout its schedule and derived constants by a threshold $0<\mu_{\rm traj}\leq(1+\gamma)L^2/(m+1)$ fixed before the run (larger thresholds make every $A_k$ below empty by Proposition~\ref{prop:mu_M_obstruction}), in particular $W^\ast=L\bbarr/\mu_{\rm traj}$, $A_Q$, $B_Y$, $\sigma_g^2$, and $M_s$, and take $0\leq\tau_{\rm sv}<\mu_{\rm traj}/2$; the threshold enters statements (a)--(b) only through these constants, and \emph{Assumption~\ref{asm:Cp} (at this $\mu_{\rm traj}$ and horizon $K$) is invoked only in statements (c)--(d)}. Define the predictable per-iteration events
		\feq{
			A_k:=\bigl\{\inf_{i\in[n]}\sigma_{\min}(\bbM_i^{\bbtheta^{(k)}})\geq\mu_{\rm traj}\bigr\},\quad k=0,\ldots,K-1,
		}
		each $\ccalF_k^-$-measurable, and the trajectory event $E_{\rm traj}:=\bigcap_{k=0}^{K-1}A_k$. The following statements jointly characterize the guarantee:
		\begin{enumerate}
			\item[(a)] (Mass-weighted stopped-process inequality; no (C$'$) premise.) The unconditional inequality
			\feq{
				&\mbE\!\Bigl[\,\tfrac{1}{nK}\!\sum_{k=0}^{K-1}\norm{\ccalG_\eta(\bbtheta^{(k)})}^2\,\bbone_{A_0\cap\cdots\cap A_k}\,\bbone_{E_{\rm rff}}\Bigr]\notag\\
				&\qquad\qquad\qquad\qquad\qquad\qquad\leq\epsilon+\Efloor(\kappa,m)
			}
			holds, the expectation being taken over the joint randomness of the random-feature preprocessing and the subsequent algorithmic sampling. Equivalently, writing $\tau:=\inf\{k:A_k^c\}\wedge K$ for the first-failure stopping time and noting that $\bbone_{A_0\cap\cdots\cap A_k}=\bbone_{\{\tau>k\}}$,
			\feq{
				\mbE\!\Bigl[\,\tfrac{1}{nK}\!\sum_{k=0}^{\tau-1}\norm{\ccalG_\eta(\bbtheta^{(k)})}^2\,\bbone_{E_{\rm rff}}\Bigr]\leq\epsilon+\Efloor(\kappa,m).
			}
			In particular, for the randomized output $\bbtheta^{(R)}$ of Algorithm~\ref{alg:cdcpg} with $R\sim\mathrm{Unif}\{0,\ldots,K-1\}$ drawn independently of the run, $\mbE\bigl[n^{-1}\norm{\ccalG_\eta(\bbtheta^{(R)})}^2\,\bbone_{\{R<\tau\}}\,\bbone_{E_{\rm rff}}\bigr]$ is bounded by the same right-hand side. Because the average is normalized by $K$ rather than by the realized prefix length, statement (a) is vacuously satisfiable under early stopping ($\tau=0$ zeroes the left side): it is a mass-weighted stopped-process inequality, not, by itself, a convergence guarantee. The convergence content is supplied by statements (c)--(d), which invoke Assumption~\ref{asm:Cp} prospectively; the selective-output clause inherits the same qualification.
			\item[(b)] (Global favorable-event bound.) Since $E_{\rm traj}\subseteq A_0\cap\cdots\cap A_k$ for every $k\leq K-1$, statement~(a) implies
			\feq{
				\mbE\!\Bigl[\,\tfrac{1}{nK}\!\sum_{k=0}^{K-1}\norm{\ccalG_\eta(\bbtheta^{(k)})}^2\,\bbone_{E_{\rm traj}\cap E_{\rm rff}}\Bigr]\leq\epsilon+\Efloor(\kappa,m).
			}
			\item[(c)] (Event probability.) $\mbP(E_{\rm traj}\cap E_{\rm rff})\geq 1-\delta_M^{\rm traj}-\delta_{\rm rff}$ by Assumption~\ref{asm:Cp} and the favorable RFF event.
			\item[(d)] (Conditional consequence.) If $\delta_M^{\rm traj}+\delta_{\rm rff}\leq 1/2$, dividing the bound of statement~(b) by $\mbP(E_{\rm traj}\cap E_{\rm rff})\geq 1/2$ yields
			\feq{
				&\mbE\!\Bigl[\,\tfrac{1}{nK}\!\sum_{k=0}^{K-1}\norm{\ccalG_\eta(\bbtheta^{(k)})}^2\,\big|\,E_{\rm traj}\cap E_{\rm rff}\Bigr]\notag\\
				&\qquad\qquad\qquad\qquad\qquad\qquad\leq 2\bigl(\epsilon+\Efloor(\kappa,m)\bigr).
			}
		\end{enumerate}
		The shared-oracle complexity is
		\feq{
			N_{\rm global}^{\rm traj}&=K(M_s+M_g)\notag\\
			&=\widetilde\ccalO\!\Bigl(\tfrac{L_JQ_{\max}L^6 W^{\ast 2}A_\kappa}{\mu_{\rm traj}^2\epsilon^2}+\tfrac{L_JQ_{\max}\sigma_g^2}{\epsilon^2}\Bigr),
		}
		with $W^\ast$ and all derived constants evaluated using $\mu_{\rm traj}$ in place of $\mu_M$; the factor $L_J$ enters through $K=\ccalO(1+L_J\Delta_J/(n\epsilon))$ multiplying $M_s$ and $M_g$. Per-iteration LSTD failure events are absorbed into the $\epsilon$-budget by the parameter choices of Theorem~\ref{thm:sample_complexity_fixed} (with $\delta_\star$ included).
	\end{theorem}

	Statement~(a) is the rigorous core of the result: each per-iteration descent inequality is multiplied by the predictable indicator $\bbone_{A_0\cap\cdots\cap A_k}\bbone_{E_{\rm rff}}\in\ccalF_k^-$ before expectations are taken, so that the concentration step never conditions on the global future event $E_{\rm traj}$, and the telescoped inequality coincides with the descent inequality summed over the stopped process $\{\bbtheta^{(k)}\}_{k<\tau}$ on $E_{\rm rff}$. A probability-weighted favorable-event bound does not equal a conditional bound without division by the event probability, which is why (d) is stated separately. The proof is in Appendix~\ref{app:sample_complexity_traj}.

	\begin{proposition}[Diagnostic for (C$'$)]\label{prop:diagnostic}
		For $\delta\in(0,1)$, let $\ell_M:=\log(2nK(m+1)/\delta)$ and $r_M(M_s,\delta):=C_M\bigl(L^2\sqrt{\ell_M/M_s}+L^2\ell_M/M_s\bigr)$ with $C_M=4$. Under Sets~A--D, for every fixed feature realization, with probability at least $1-\delta$ over the algorithmic sampling,
		\feq{
			\sup_{i,k}\norm{\bbM_i^{(k)}-\bbM_i^{\bbtheta^{(k)}}}\leq r_M(M_s,\delta).
		}
		Consequently, the empirical test $T_{\rm cert}:=\{\inf_{i,k}\sigma_{\min}(\bbM_i^{(k)})\geq 2r_M\}$ has false-certification probability at most $\delta$:
		\feq{
			\mbP\bigl(T_{\rm cert}\cap\{\inf_{i,k}\sigma_{\min}(\bbM_i^{\bbtheta^{(k)}})<r_M\}\bigr)\leq\delta.\label{eq:cert_inequality}
		}
		Inequality~\eqref{eq:cert_inequality} is a simultaneous false-certification bound: it controls the joint probability of certifying while the population condition fails. It does not by itself control the conditional probability $\mbP\bigl(\inf_{i,k}\sigma_{\min}(\bbM_i^{\bbtheta^{(k)}})<r_M\,\big|\,T_{\rm cert}\bigr)$, which additionally requires a lower bound on the test-passing probability $\mbP(T_{\rm cert})$: a test that never passes satisfies~\eqref{eq:cert_inequality} vacuously. A less conservative certificate uses $\hat\mu_{\rm cert}:=\inf_{i,k}\sigma_{\min}(\bbM_i^{(k)})-r_M$; whenever $\hat\mu_{\rm cert}>0$, the executed run satisfied $\inf_{i,k}\sigma_{\min}(\bbM_i^{\bbtheta^{(k)}})\geq\hat\mu_{\rm cert}$ except on an event of probability at most $\delta$.
	\end{proposition}

	\begin{proof}
		With $\bbY:=\tdbbphi^{\rm raw}(\tdbbphi^{\rm raw}-\gamma\tdbbphi^{\rm raw}{}')^\top$ one has $\norm{\bbY}\leq(1+\gamma)L^2\leq2L^2$, so the centered summand obeys $\norm{\bbX}\leq4L^2$ with $\max\{\norm{\mbE\bbX\bbX^\top},\norm{\mbE\bbX^\top\bbX}\}\leq4L^4$. Conditional on $\ccalF_k^-$, the transition samples in $\ccalD_s^{(k)}$ are i.i.d.\ from the joint one-step law $d^{\pi_{\bbtheta^{(k)}}}\otimes\pi_{\bbtheta^{(k)}}\otimes\mbP\otimes\pi_{\bbtheta^{(k)}}$, so rectangular matrix Bernstein~\cite[Thm.~6.1.1]{tropp2015introduction}, applied through the Hermitian dilation of dimension $2(m+1)$ at confidence $\delta/(nK)$, gives $\norm{\bbM_i^{(k)}-\bbM_i^{\bbtheta^{(k)}}}\leq2\sqrt2\,L^2\sqrt{\ell_M/M_s}+\tfrac83L^2\ell_M/M_s\leq r_M$ conditionally, and a union bound over $i\in[n]$, $k<K$ yields the simultaneous bound. On that event, Weyl's inequality gives $\inf_{i,k}\sigma_{\min}(\bbM_i^{\bbtheta^{(k)}})\geq\inf_{i,k}\sigma_{\min}(\bbM_i^{(k)})-r_M$, so $T_{\rm cert}\cap\{\inf_{i,k}\sigma_{\min}(\bbM_i^{\bbtheta^{(k)}})<r_M\}$ is contained in the complement of the concentration event, and the second statement follows in the same way.
	\end{proof}

	\emph{Scope of certification.} The diagnostic gives a simultaneous lower confidence bound for the population TD matrices along the executed run, for one feature draw and one fixed horizon; it controls only the population TD-stability event, while the vector-concentration and gradient events are absorbed in expectation through the parameter choices. It does not give a lower bound on the probability that a future run satisfies (C$'$), so it cannot replace that premise in Theorem~\ref{thm:sample_complexity_traj}(c)--(d). The analyzed budget and projection radius use a threshold fixed before the run; a separate pilot may calibrate its scale but does not certify the main trajectory. The confidence statement does not cover repeated feature redraws after failed tests. Prospective certification would require a stability-on-a-tube argument (the continuity of $\bbtheta\mapsto\bbM_i^{\bbtheta}$, Lemma~\ref{lem:M_continuity}, provides a route) or an online stopping rule, which we leave to future work.

	\begin{remark}[Passability of the certificate]\label{rem:passability}
		The test $T_{\rm cert}$ has a deterministic passability barrier beyond the rank requirement $M_s\geq m+1$. The empirical matrix $\bbM_i^{(k)}$ is an average of rank-one matrices of nuclear norm at most $(1+\gamma)L^2$, so the argument of Proposition~\ref{prop:mu_M_obstruction} gives $\sigma_{\min}(\bbM_i^{(k)})\leq(1+\gamma)L^2/(m+1)$ on every sample path. Since $T_{\rm cert}$ requires $\sigma_{\min}(\bbM_i^{(k)})\geq2r_M=8L^2(\sqrt{\ell_M/M_s}+\ell_M/M_s)$, a necessary condition for the test to pass is
		\feq{
			8\Bigl(\sqrt{\frac{\ell_M}{M_s}}+\frac{\ell_M}{M_s}\Bigr)\leq\frac{1+\gamma}{m+1},\label{eq:passability}
		}
		in which $L^2$ has cancelled; in particular $M_s\geq64(m+1)^2\ell_M/(1+\gamma)^2$. As an illustration, for $n=9$, $K=200$, $\gamma=0.95$, and $\delta=0.05$, condition~\eqref{eq:passability} requires $M_s>6.6\times10^5$ at $m=50$ and $M_s>1.2\times10^9$ at $m=2000$. These are necessary thresholds under the conservative uniform radius $r_M$, not sufficient sample sizes. The barrier does not affect the validity of Proposition~\ref{prop:diagnostic}, whose false-certification control holds for a test that rarely or never passes, but it limits its practical use at moderate batch sizes such as the $M_s=2000$ of Section~\ref{sec:simu}: there the certificate cannot pass, and an observed singular value, such as a median over $(i,k)$, is not a certificate. Sharpening the radius is outside the present scope.
	\end{remark}

	\begin{remark}[Conditioning versus approximation as $m$ varies]\label{rem:conditioning}
		The feature dimension $m$ enters the guarantee in two opposite ways. On the approximation side, the floor component $\Efloor^{\rm rff}(\kappa,m)=\widetilde\ccalO\bigl(A_\kappa\gamma^2\bbarr^2\,C_{\rm RFF}(D_\kappa d_S)\tilde g_\alpha^2/((1-\gamma)^3m)\bigr)$ decreases with $m$, and $\Efloor^{\rm trunc}$ and $\Efloor^{\rm graph}$ do not depend on $m$; since these are bounds in the function norm $\ccalL^2(\nu^\pi)$ (Lemma~\ref{lem:pop_L2}(c)), no factor $1/\mu_M$ multiplies them. On the estimation side, Proposition~\ref{prop:mu_M_obstruction} forces $\mu_M(m)\leq(1+\gamma)L^2/(m+1)$, so the radius $W^\ast=L\bbarr/\mu_M$, the gradient-noise scale $\sigma_g^2\propto W^{\ast2}$, the deterministic bound $A_Q$, and the critic batch $M_s=\widetilde\ccalO(L^6W^{\ast2}A_\kappa/(\mu_M^2\epsilon))\propto\mu_M^{-4}$ deteriorate at least polynomially in $m$: under the compatible scenario $\mu_M=\Theta(1/m)$, $M_s=\widetilde\ccalO(m^4)$ at fixed excess $\epsilon$ with the other primitives fixed, and faster deterioration or singularity is not excluded. Increasing $m$ therefore trades a decreasing population approximation floor against an increasing sample budget; the $\epsilon$-exponent of Theorem~\ref{thm:sample_complexity_fixed} is unaffected at fixed $m$, but the prefactor is not. The coefficient-norm comparison~\eqref{eq:pop_sup} does carry the factor $1+C'=1+(1+\gamma)L^2/\mu_M$; it is used only where a pointwise bound is needed and does not enter the floor. Whether a lower bound of order $1/m$ on $\mu_M$ is attainable for the spectral features is not established, and positive excitation itself remains a premise (Set~C).
	\end{remark}

	\subsection{Offline locality selection}

	\begin{proposition}[Offline radius selection from the decay envelopes]\label{prop:adaptive}
		Under the conditions of Theorem~\ref{thm:sample_complexity_fixed} with $m$ fixed and the strictly stronger hypothesis $\Delta\rho<1$, let $C_D$ be the graph-growth constant of Remark~\ref{rem:n_scaling}, so that $D_{\kappa_c}\leq C_D\Delta^{\kappa+\kappa_\pi}$ (capped at $n$), and define the $\kappa$-independent constants
		\feq{
			C_t&:=\frac{G^2C_D^2\Delta^{2\kappa_\pi}\gamma^2\tilde c^2\rho^2}{(1-\gamma)^3},\notag\\
			C_g&:=\Bigl(\frac{Gc\,\Delta}{(1-\gamma)(1-\tilde\delta)(\Delta-1)}\Bigr)^2,
		}
		so that $\Efloor^{\rm trunc}(\kappa)\leq C_t(\Delta\rho)^{2\kappa}$ and $\Efloor^{\rm graph}(\kappa)\leq C_g(\Delta\rho)^{2\kappa}$ for every $\kappa\geq1$. Fix an integer working range $[\![1,\kappa_{\max}]\!]$ with $\kappa_{\max}\geq\max\{1,\kappa_\pi+1\}$ and set
		\feq{
			\kappa^\ast:=\max\Bigl\{&1,\ \kappa_\pi+1,\\
			&\Bigl\lceil\frac{\log\max\{1,\Cfloor\max\{C_t,C_g\}/\epsilon\}}{2\log(1/(\Delta\rho))}\Bigr\rceil\Bigr\}.\label{eq:kappa_star}
		}
		Suppose that $\kappa^\ast\leq\kappa_{\max}$, that $m$ satisfies the RFF feasibility relation~\eqref{eq:rff_feasibility} at the selected radius $\kappa^\ast$ in its full form (the simplified ratio $\tilde g_\alpha(\kappa^\ast)^2D_{\kappa^\ast}d_S/m\leq1$ is a necessary scale, not a restatement of that relation), and that Assumption~\ref{asm:C} holds at $\kappa^\ast$ with threshold $\mu_M(m,\kappa^\ast)$; features are drawn once, at the selected radius, after $\kappa^\ast$ has been determined from the (oracle) decay constants. Then $\Cfloor\Efloor^{\rm trunc}(\kappa^\ast)\leq\epsilon$, $\Cfloor\Efloor^{\rm graph}(\kappa^\ast)\leq\epsilon$, and with the parameters of Theorem~\ref{thm:sample_complexity_fixed} evaluated at $\kappa^\ast$,
		\feq{
			&\mbE\!\Bigl[\,\tfrac{1}{nK}\sum_{k<K}\norm{\ccalG_\eta(\bbtheta^{(k)})}^2\,\Big|\,E_M\cap E_{\rm rff}\Bigr]\notag\\&\qquad\leq 3\epsilon+\Cfloor\,\Efloor^{\rm rff}(\kappa^\ast,m),
		}
		where, with $\kappa_c^\ast:=\kappa^\ast+\kappa_\pi$,
		\feq{
			\Efloor^{\rm rff}(\kappa^\ast,m)=\widetilde\ccalO\!\Bigl(\tfrac{G^2\gamma^2\bbarr^2}{(1-\gamma)^5}\,D_{\kappa_c^\ast}^2\,C_{\rm RFF}(D_{\kappa^\ast}d_S)\,\tfrac{\tilde g_\alpha(\kappa^\ast)^2}{m}\Bigr),
		}
		an upper envelope obtained from $\epsilon_P(m)^2=\widetilde\ccalO(C_{\rm RFF}(D_{\kappa^\ast}d_S)\tilde g_\alpha(\kappa^\ast)^2/m)$ (no matching lower inversion is claimed). Since $\tilde g_\alpha(\kappa)=\exp(\ccalO(D_\kappa d_S))$ and $\kappa^\ast=\ccalO(\log(1/\epsilon))$, substituting the selected radius into this upper envelope gives, according to the neighborhood growth $v(\kappa):=D_\kappa$, the following envelopes in $1/\epsilon$, with no matching lower growth rate asserted: for linear growth ($v(\kappa)=\ccalO(\kappa)$, e.g.\ path graphs with $D_\kappa=\min\{n,2\kappa+1\}$) the envelope is polynomial in $1/\epsilon$; for polynomial growth $v(\kappa)=\ccalO(\kappa^p)$, $p>1$, it is quasi-polynomial, $\exp(\ccalO(\log^p(1/\epsilon)))$; and under exponential neighborhood growth $v(\kappa)=\ccalO(\Delta^\kappa)$ it is exponential in a power of $1/\epsilon$, $\exp(\ccalO((1/\epsilon)^{\chi}))$ with $\chi=\log\Delta/(2\log(1/(\Delta\rho)))$, which need not be smaller than one. In every case the inflation can violate the RFF feasibility requirement if $m$ is held fixed. When the realized growth $D_{\kappa_c^\ast}$ reaches $n$, the cap contributes an explicit factor of $n^2$, consistent with Remark~\ref{rem:n_scaling}.
	\end{proposition}

	The proof is in Appendix~\ref{app:adaptive}. This is an \emph{offline radius-selection rule using known decay and growth constants}, not an online data-adaptive method and not a rate improvement; the selected radius balances the two spatial residuals against the target excess, and the surviving finite-feature residual need not shrink. The fixed-radius theorem guarantees a total stationarity target $\tau_{\rm tot}$ whenever the optimization excess and the complete residual bound satisfy $\epsilon_{\rm opt}+\Efloor(\kappa,m)\leq\tau_{\rm tot}$; for the selected-radius construction, a sufficient allocation is $3\epsilon_{\rm opt}+\Cfloor\Efloor^{\rm rff}(\kappa^\ast,m)\leq\tau_{\rm tot}$, for example $\epsilon_{\rm opt}=\tau_{\rm tot}/6$ with the surviving finite-feature term at most $\tau_{\rm tot}/2$. These are sufficient budget allocations, not necessary conditions, and at fixed $m$ the binding term is the finite-feature residual through $\tilde g_\alpha(\kappa)^2D_\kappa d_S/m$; the rule applies only while the selected radius remains in the feasible range.

	\section{Numerical Experiments}\label{sec:simu}

	We study CDCPG on a Linear-Coupled-Quadratic (LCQ) benchmark, a structured networked-control benchmark with a closed-form linear-quadratic regulator (LQR) reference. Each of $n=9$ agents on a path graph carries scalar state $\bbs_i\in\mbR$ and action $\bba_i\in[-A,A]$, with local dynamics
	\feq{
		\bbs_i^{t+1}=\rho_{\rm dyn}\bbs_i^t+\beta\bba_i^t+\kappa_{\rm dyn}\!\sum_{j\in\ccalN_i\setminus\{i\}}\!(\bbs_j^t-\bbs_i^t)+\bbvarepsilon_i^t,
	}
	$\bbvarepsilon_i^t\sim\ccalN(0,\sigma_d^2)$, quadratic neighbor-coupled rewards $r_i=-\bbs_i^2-w_n\sum_{j\in\ccalN_i\setminus\{i\}}\bbs_j^2-w_a\bba_i^2$, and initial law $\mu_0=\ccalN(\bbzero,\sigma_0^2\bbI_n)$. Parameters are listed in Table~\ref{tab:lcq_params}. States and actions are projected onto $[-S_{\max},S_{\max}]$ and $[-A,A]$ as a numerical safeguard; boundary events were recorded on fewer than $0.1\%$ of transitions. Stacking the agents gives $\bbs'=\bbA\bbs+\bbB\bba+\bbvarepsilon$ with $\bbA=\rho_{\rm dyn}\bbI_n+\kappa_{\rm dyn}(\bbW_\ccalG-\bbD_\ccalG)$, $\bbB=\beta\bbI_n$, and stage cost weights $\bbQ=\bbI_n+w_n\bbD_\ccalG$, $\bbR=w_a\bbI_n$, where $\bbW_\ccalG$ and $\bbD_\ccalG$ are the adjacency and degree matrices. Let $\bbK_{\rm LQR}$ denote the gain solving the unconstrained discounted Riccati equation. The reference $V_\sigma^{\rm ref}=-2.713$ is the closed-form per-agent infinite-horizon discounted return of the stochastic policy $\bba=-\bbK_{\rm LQR}\bbs+\ccalN(\bbzero,\sigma_\pi^2\bbI)$ from $\mu_0$ in the unclipped model (the subscript indicates the $\sigma_\pi$-perturbed policy), confirmed to $10^{-2}$ by $5\times10^3$ Monte Carlo episodes. Because the simulator clips states and actions, $\bbK_{\rm LQR}$ is a reference controller and is not claimed to be optimal for the clipped MDP. Signed gaps below are $100(\bar J-V_\sigma^{\rm ref})/\abs{V_\sigma^{\rm ref}}\%$.

	The LCQ studies examine related implementations rather than an instance of Sets~A--D. The unclipped drift is linear and unbounded and the unclipped reward is quadratic and unbounded; clipping the successor state creates boundary probability masses, so the implemented transition is not the Gaussian density used by the spectral construction; a clipped Gaussian action has atoms at the action bounds, and the actor uses the score of the latent Gaussian draw rather than a bounded score of an action density as in (A3); the learning rate is cosine-decayed, whereas the analysis uses a constant step; and the control comparison uses a structured quadratic critic, whereas the feature study uses the augmented RFF representation. Accordingly, returns and reference discrepancies are reported as empirical observations. They are compared below with the qualitative form of the corresponding bounds, and they do not validate the projected-stationarity bound, its constants, or its excitation and decay assumptions.

	\begin{table}[!t]
		\centering
		\caption{Parameters of the LCQ benchmark.}
		\label{tab:lcq_params}
		\renewcommand{\arraystretch}{1.15}
		\begin{tabular}{@{}llll@{}}
			\hline
			$\rho_{\rm dyn}=0.7$ & $\beta=0.5$ & $\sigma_d=0.10$ & $\kappa_{\rm dyn}=0.20$ \\
			$w_n=0.5$ & $w_a=0.5$ & $A=5.0$ & $S_{\max}=3.0$ \\
			$\sigma_0=0.5$ & $\sigma_\pi=0.30$ & $\gamma=0.95$ & $n=9$ \\
			\hline
		\end{tabular}
	\end{table}

	\subsection{Setup and baselines}\label{sec:simu_setup}

	\emph{Policies.} In every method, agent $i$ samples a latent Gaussian action $\bbu_i\sim\ccalN(\mu_{\bbtheta_i}(\bbs),\sigma_\pi^2)$ with linear mean $\mu_{\bbtheta_i}(\bbs)=\bbtheta_i^\top\bbx_i(\bbs)$, where $\bbx_i(\bbs)$ collects the state coordinates observed by the policy ($\bbx_i=\bbs_{\ccalN_i^{\kappa_\pi}}$ for the local methods), and applies the clipped action $\bba_i=\Pi_{[-A,A]}(\bbu_i)$. The recorded actor update uses the Gaussian score of the latent draw, which for this linear mean is the parameter score $\nabla_{\bbtheta_i}\log p_{\bbtheta_i}(\bbu_i\mid\bbs)=\sigma_\pi^{-2}\bigl(\bbu_i-\bbtheta_i^\top\bbx_i(\bbs)\bigr)\bbx_i(\bbs)$. This is not the score of the distribution of the applied action, which has atoms at $\pm A$, so the policy class differs from that of (A3); the latent likelihood-ratio form is nevertheless a valid gradient identity for the clipped controller, since $\nabla_{\bbtheta}\mbE_{\bbu\sim p_\bbtheta}[F(\Pi_{[-A,A]}(\bbu))]=\mbE_{\bbu\sim p_\bbtheta}[F(\Pi_{[-A,A]}(\bbu))\nabla_\bbtheta\log p_\bbtheta(\bbu)]$ for bounded measurable $F$.

	\emph{CDCPG implementation.} The control comparison uses a one-hop actor ($\kappa_\pi=1$), critic radius $\kappa=1$, $\lambda=10^{-4}$, $\eta=5\times 10^{-2}$ with cosine decay, $M_s=2000$, and $M_g=1024$. This pair of radii lies outside the standing relation $\kappa\geq\kappa_\pi+1$ of the analysis. The critic in this comparison is a structured state-value critic tailored to the linear-quadratic structure: $\hat V_i$ is fitted by bootstrapped least-squares TD on the quadratic basis $\bigl(1,\bbs_{\ccalN_i^\kappa},\mathrm{vech}(\bbs_{\ccalN_i^\kappa}\bbs_{\ccalN_i^\kappa}^\top)\bigr)$ of dimension $1+D_\kappa+D_\kappa(D_\kappa+1)/2$. A state-only critic multiplied by the score has zero conditional mean, so the action dependence required by the estimator~\eqref{eq:gradient_estimator} is supplied, using the known local model, by the mean-transition plug-in surrogate
	\feq{
		\hat Q_i(\bbs,\bba)=r_i(\bbs_{\ccalN_i},\bba_i)+\gamma\,\hat V_i\bigl(\bbf_{i,\kappa}(\bbs,\bba)\bigr),
	}
	which is consistent with the model-assisted access of (D1). The surrogate evaluates $\hat V_i$ at the mean successor and is not an exact one-step Bellman expectation for the clipped simulator. For a quadratic $\hat V_i(\bby)=c+\bbb^\top\bby+\bby^\top\bbH\bby$ and unclipped additive noise, $\mbE[\hat V_i(\bbf+\bbvarepsilon)]=\hat V_i(\bbf)+\sigma_d^2\,\mathrm{tr}(\bbH)$; the omitted term does not depend on the action, so it has zero conditional mean against the centered score and does not change the expected actor update in the unclipped model, whereas under clipping the difference is in general action dependent. The locality, aggregation over $\ccalN_i^{\kappa_c}$, and projected update are those of Algorithm~\ref{alg:cdcpg}; the critic is not the augmented RFF--LSTD critic of Definition~\ref{def:aug_features} analyzed by the theorems, which is used only in Section~\ref{sec:simu_feat}.

	\emph{Baselines.} Independent policy gradient (IPG) applies REINFORCE with each agent's own state ($\kappa_\pi=0$). Networked AC is adapted from the communication-based method of~\cite{zhang2018fully}: a one-hop policy with a one-hop weight-consensus step on a bootstrapped state-value critic. MAAC is adapted from~\cite{iqbal2019actor}: attention aggregation over a fixed random feature map with a Monte Carlo regression critic. Centralized uses a single global linear Gaussian policy with a bootstrapped state-value critic on the global quadratic basis; it is a full-information baseline and not a guaranteed upper bound at finite training. The labels identify the reference methods; the implementations are adaptations to LCQ and do not reproduce the original algorithms or transfer their guarantees. All methods share the exploration scale $\sigma_\pi$ and the cosine schedule, and their per-iteration interaction budgets agree within a factor of two.

	\emph{Evaluation.} All runs use $K=200$ outer iterations and $5$ random seeds; curves report the mean $\pm$ one standard deviation across seeds. Each plotted iterate is evaluated by $32$ rollouts of horizon $100$ from $\mu_0$ under the stochastic policy, so the plotted returns are horizon-$100$ Monte Carlo estimates and not the infinite-horizon objective $J/n$. In the unclipped model the reference controller has infinite-horizon return $-2.71323$ and horizon-$100$ return $-2.69984$. Using its infinite-horizon value as the common benchmark shifts every finite-horizon gap numerator by $0.01339$ relative to using its horizon-$100$ value, that is, by about $0.5$ percentage points with the denominator held at $\abs{V_\sigma^{\rm ref}}$. This reference offset is common to all methods and leaves comparisons among the reported finite-horizon returns unchanged; it is not a common tail correction for the learned policies, whose omitted tails depend on the policy and on clipping. The evaluation protocol of the separately tabulated final returns (checkpoint, horizon, and number of rollouts) was not archived.

	\subsection{Comparison with decentralized baselines}

	Fig.~\ref{fig:convergence} and Table~\ref{tab:lcq_compare} report the return comparison. The structured-critic CDCPG implementation and the centralized baseline attain final mean returns of $-2.728$ and $-2.725$, respectively; no statistical equivalence is claimed at five seeds. Networked AC, MAAC, and IPG attain lower returns under the stated implementations and budgets, in this order. The tabulated final returns and the plotted per-iteration evaluations are two separately recorded summaries: the last plotted means differ from the tabulated values by at most $0.03$ (largest for IPG), which is of the order of one across-seed standard deviation, and the ranking of the methods is the same in both. Whether the tables use the same final iterates, horizon, and rollout count as the curves was not archived, so the two summaries are reported side by side and are not claimed to be reconciled. IPG uses an own-state policy, whereas CDCPG and Networked AC use one-hop policies. For orientation, in the unclipped model a numerically optimized diagonal ($\kappa_\pi=0$) linear feedback attains an infinite-horizon return of $-2.795$, while a numerically optimized one-hop feedback attains $-2.7134$ against the unrestricted Riccati value $-2.7132$, which bounds every linear feedback from above; the one-hop class therefore loses nothing at three decimals, and no global optimality is claimed for the diagonal candidate. The compared implementations also differ in critic construction and aggregation, so the observed returns do not isolate the effects of information, approximation, or optimization. Lemma~\ref{lem:grad_bias} bounds the omitted-reward contribution to the gradient bias; it does not predict return gaps, and the theorems make no prediction about the baselines.

	\begin{figure}[!t]
		\centering
		\includegraphics[width=0.95\columnwidth]{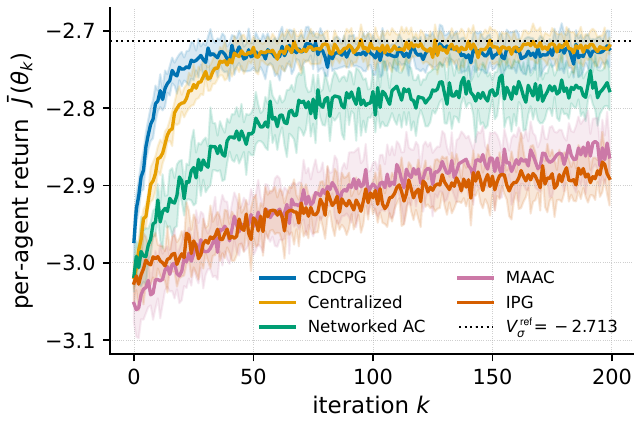}
		\caption{Horizon-$100$ Monte Carlo estimate of the per-agent discounted return (axis label $\bar J(\bbtheta_k)$) versus iteration $k$ on LCQ ($n=9$, $5$ seeds, mean $\pm$ one std). Dotted line: infinite-horizon reference $V_\sigma^{\rm ref}=-2.713$. CDCPG uses the structured quadratic-basis critic of Section~\ref{sec:simu_setup}; the separate RFF feature study is reported in Fig.~\ref{fig:q_vs_m}.}
		\label{fig:convergence}
	\end{figure}

	\begin{table}[!t]
		\centering
		\caption{LCQ benchmark, $n=9$ path graph, $K=200$, $5$ seeds, $\sigma_\pi=0.30$. Final per-agent return $\bar J$ (mean and std across seeds) and signed gap to $V_\sigma^{\rm ref}=-2.713$. The CDCPG row uses the structured quadratic-basis critic.}
		\label{tab:lcq_compare}
		\renewcommand{\arraystretch}{1.15}
		\begin{tabular}{@{}lcc@{}}
			\hline
			Method & $\bar J$ ($\pm$ std) & Gap \\
			\hline
			\textbf{CDCPG (structured critic)} & $-2.728\,(0.018)$ & $-0.55\%$ \\
			Centralized & $-2.725\,(0.015)$ & $-0.44\%$ \\
			Networked AC & $-2.756\,(0.042)$ & $-1.59\%$ \\
			MAAC & $-2.852\,(0.029)$ & $-5.12\%$ \\
			IPG & $-2.863\,(0.025)$ & $-5.53\%$ \\
			\hline
		\end{tabular}
	\end{table}

	\subsection{Feature-dimension trend}\label{sec:simu_feat}

	This study records the reference-discrepancy trend of an augmented RFF implementation (features as in Definition~\ref{def:aug_features}) at $m\in\{50,100,250,500,1000,2000\}$ with the policy, $\kappa$, the rollout protocol, and the remaining settings of Section~\ref{sec:simu_setup} held fixed. Fig.~\ref{fig:q_vs_m} reports the pointwise relative discrepancy $\abs{\hat Q-Q^{\rm ref}_\sigma}/\abs{Q^{\rm ref}_\sigma}$ from the closed-form unconstrained linear-quadratic reference $Q^{\rm ref}_\sigma$, averaged over $10^4$ test samples drawn from the occupancy of the fixed policy. Its across-seed mean decreases from about $43\%$ at $m=50$ to about $9\%$ at $m=2000$, and the log--log least-squares fit to the six means has slope $-0.46$ ($R^2=0.98$). Theorem~\ref{thm:rff_approx} gives $\epsilon_P(m)=\widetilde\ccalO(m^{-1/2})$, so the finite-feature residual $B_{\rm rff}(m)$ of Theorem~\ref{thm:q_linear} decreases at order $m^{-1/2}$ up to logarithms; the fitted exponent is consistent with this order, which would correspond to a regime in which the $m$-independent spatial residual $B_{\rm loc}(\kappa)$ and the finite-sample term of Theorem~\ref{thm:lstd_error}, whose bound grows with $m$ at fixed $M_s$, are not dominant.

	The sweep does not verify that bound, for four reasons. First, the plotted quantity is a discrepancy from the unclipped reference rather than a $Q$-error of the clipped MDP. The identity of the fixed policy was not archived, so a policy mismatch with the reference cannot be excluded, and the record does not state whether $\hat Q$ and $Q^{\rm ref}_\sigma$ are local-reward quantities or network aggregates, nor the values of $\alpha$, $W^\ast$, and $\tau_{\rm sv}$ used; the study is therefore a descriptive observation and not a reproducible same-policy test of the analyzed estimator. Second, finite-sample estimation, regularization, conditioning, and localization also affect it. Third, every tested dimension lies below the smallest $m$ at which the sufficient condition~\eqref{eq:rff_feasibility} can hold: with $D_\kappa d_S\geq3$, $\tilde g_\alpha\geq1$, $\epsilon_P\leq1$, and $\delta_{\rm rff}<1$, that condition requires $m>128\cdot7\cdot\log36\approx3.2\times10^3$. Sufficient constants can be conservative, so this does not show that smaller $m$ approximate poorly; it shows that the trend is descriptive and is not certified by the constants of Theorem~\ref{thm:rff_approx}. Fourth, the empirical TD matrix $\bbM_i^{(k)}$ is a sum of $M_s$ rank-one terms, and at $m=2000$ the augmented dimension $m+1=2001$ exceeds $M_s=2000$. The unshifted empirical matrix is therefore singular at that point and the certificate of Proposition~\ref{prop:diagnostic} cannot pass. This does not determine whether the regularized estimator is well defined: with the recorded ridge $\lambda=10^{-4}$ the estimator of Definition~\ref{def:lstd} inverts the shifted matrix $\bbM_i^{(k)}+\lambda\bbI$, whose invertibility is not precluded by the rank deficiency of $\bbM_i^{(k)}$, and the threshold branch was not logged. Refitting the plotted means without this point gives a slope of about $-0.5$, so the reported trend does not rest on it.

	\begin{figure}[!t]
		\centering
		\includegraphics[width=0.85\columnwidth]{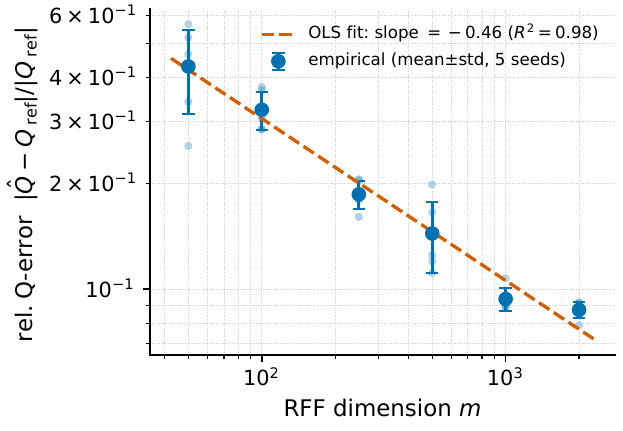}
		\caption{Mean relative discrepancy $\abs{\hat Q-Q^{\rm ref}_\sigma}/\abs{Q^{\rm ref}_\sigma}$ from the unconstrained linear-quadratic reference (the quantity on the vertical axis, abbreviated ``rel.\ $Q$-error'' in the plot; it is a reference discrepancy, not a $Q$-error of the clipped MDP) versus RFF dimension $m$ on LCQ ($M_s=2000$, log--log axes; faint markers: five seeds; error bars: one std). Dashed line: least-squares fit to the means with slope $-0.46$ ($R^2=0.98$). At $m=2000$ the unshifted empirical TD matrix is rank deficient (Section~\ref{sec:simu_feat}).}
		\label{fig:q_vs_m}
	\end{figure}

	\subsection{Decay proxy and ablations}

	The empirical absolute reward covariance decreases with graph distance, with fitted geometric rate $\hat\rho=0.86$ (Fig.~\ref{fig:decay}). This measures spatial dependence of rewards, whereas (B1) bounds worst-case sensitivity of an action-value function to changes outside a neighborhood. The covariance fit is therefore a descriptive statistic; it verifies neither (B1) nor the differentiated-decay premise (D3)(ii) and is not used to select the radius in Proposition~\ref{prop:adaptive}. For orientation, on the path graph ($\Delta=2$), the graph-tail condition $(\Delta-1)\rho<1$ of Lemma~\ref{lem:grad_bias} holds for every $\rho<1$, whereas the stronger condition $\Delta\rho<1$ of Proposition~\ref{prop:adaptive} requires $\rho<1/2$, which the finite-propagation baseline $\rho_0=\gamma^{1/(\kappa_\pi+1)}\geq0.95$ of Lemma~\ref{lem:baseline_decay} does not satisfy.

	\begin{figure}[!t]
		\centering
		\includegraphics[width=0.85\columnwidth]{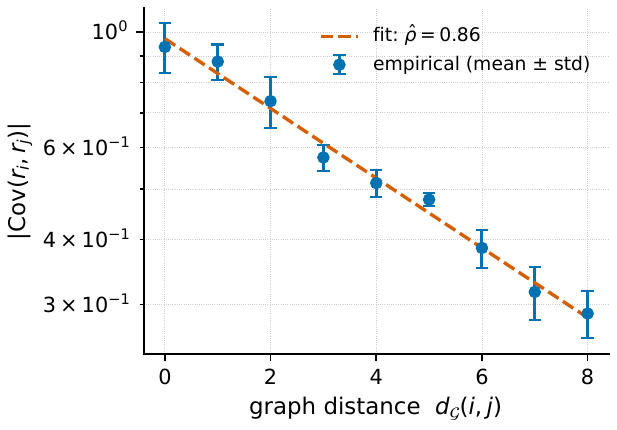}
		\caption{Empirical absolute cross-agent reward covariance $\abs{\mathrm{Cov}(r_i,r_j)}$ versus graph distance on LCQ ($5$ seeds, semi-log axes). Dashed line: geometric fit with rate $\hat\rho=0.86$.}
		\label{fig:decay}
	\end{figure}

	Fig.~\ref{fig:ablation} and Table~\ref{tab:lcq_ablation} sweep $\kappa$ (at $\lambda=10^{-4}$) and $\lambda$ (at $\kappa=3$) for the structured-critic implementation at $K=150$, $3$ seeds. These sweeps are a separate set of runs whose configuration beyond the swept parameters is not archived, and they are not directly comparable with Table~\ref{tab:lcq_compare}: the $\kappa=1$ ablation return ($-2.778$) differs from the main-comparison CDCPG return ($-2.728$) at the same $\kappa$ and $\lambda$, a difference not explained by the recorded settings. As in the main comparison, the tabulated final returns were recorded separately from the plotted evaluations; the two differ by up to $0.04$ (largest at $\kappa=1$ and $\lambda=1$). Both agree on the following. With the state-value basis restricted to the agent's own state ($\kappa=0$, outside the standing radius condition $\kappa\geq\max\{1,\kappa_\pi+1\}$; the model-assisted surrogate still uses neighboring states through $r_i$ and $f_i$) the final gap is about $-5\%$, whereas the positive-radius bases $\kappa\in\{1,2,3\}$ attain gaps between $-2.4\%$ and $-0.7\%$. The differences among the latter three are comparable to the across-seed standard deviations and are not ordered identically in the table and in the plotted endpoints. Because the remaining configurations were not retained, the comparison is descriptive: in these recorded runs the positive-radius bases achieve higher final returns than the radius-zero basis, the effect of the basis radius is not isolated, and no geometric rate in $\kappa$ is inferred; the gap is, moreover, a return difference and not the stationarity quantity bounded through $\Efloor^{\rm trunc}(\kappa)$ and $\Efloor^{\rm graph}(\kappa)$. Returns vary little over $\lambda\in\{10^{-6},10^{-4},10^{-2}\}$ and deteriorate markedly at $\lambda=1$, qualitatively consistent with the additive ridge-bias bound $2\lambda W^\ast/\mu_M$ of Theorem~\ref{thm:lstd_error}, which grows linearly in $\lambda$; the data do not identify whether the deterioration at $\lambda=1$ is caused by this bias, by conditioning, or by other implementation features.

	\begin{figure*}[!t]
		\centering
		\begin{minipage}{0.45\textwidth}
			\centering
			\includegraphics[width=\linewidth]{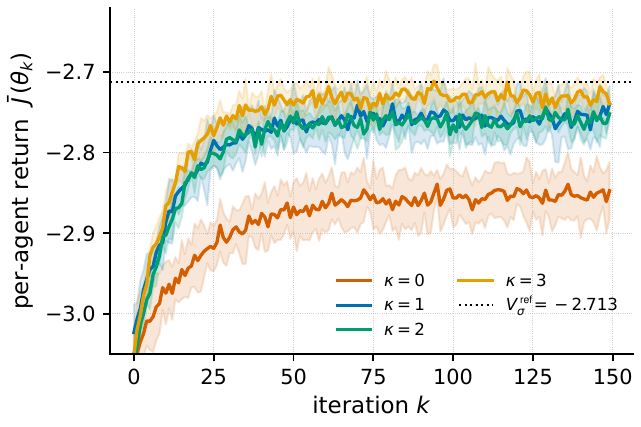}\\
			{\footnotesize (a) Truncation radius $\kappa$.}
		\end{minipage}\hfill
		\begin{minipage}{0.45\textwidth}
			\centering
			\includegraphics[width=\linewidth]{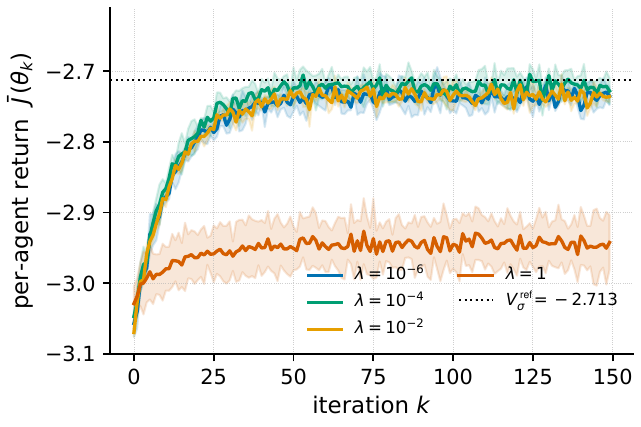}\\
			{\footnotesize (b) Regularization $\lambda$.}
		\end{minipage}
		\caption{Structured-critic CDCPG ablation on LCQ ($K=150$, $3$ seeds; mean $\pm$ one std of the per-iteration return estimates). Dotted line: infinite-horizon reference $V_\sigma^{\rm ref}=-2.713$. (a) Truncation-radius sweep: the radius-zero basis attains lower returns, and $\kappa\in\{1,2,3\}$ are not separated beyond the across-seed variation. (b) Regularization sweep: returns vary little over four decades and deteriorate at $\lambda=1$; the mechanism is not identified by these data.}
		\label{fig:ablation}
	\end{figure*}

	\begin{table}[!t]
		\centering
		\caption{Structured-critic CDCPG ablation on LCQ ($K=150$, $3$ seeds). Final per-agent return $\bar J$ (mean and std) and signed gap to $V_\sigma^{\rm ref}$. The $\kappa$ sweep fixes $\lambda=10^{-4}$, and the $\lambda$ sweep fixes $\kappa=3$. The two sweeps are separate runs; the repeated setting $\kappa=3$, $\lambda=10^{-4}$ appears once in each and was not pooled, and the sweeps are not directly comparable with Table~\ref{tab:lcq_compare} (see the text).}
		\label{tab:lcq_ablation}
		\renewcommand{\arraystretch}{1.15}
		\begin{tabular}{@{}llcc@{}}
			\hline
			Sweep & Setting & $\bar J$ ($\pm$ std) & Gap \\
			\hline
			$\kappa$ & $0$ & $-2.852\,(0.040)$ & $-5.13\%$ \\
			& $1$ & $-2.778\,(0.030)$ & $-2.40\%$ \\
			& $2$ & $-2.748\,(0.025)$ & $-1.29\%$ \\
			& $3$ & $-2.732\,(0.028)$ & $-0.70\%$ \\
			\hline
			$\lambda$ & $10^{-6}$ & $-2.732\,(0.018)$ & $-0.70\%$ \\
			& $10^{-4}$ & $-2.730\,(0.018)$ & $-0.63\%$ \\
			& $10^{-2}$ & $-2.728\,(0.018)$ & $-0.55\%$ \\
			& $1$ & $-2.978\,(0.034)$ & $-9.77\%$ \\
			\hline
		\end{tabular}
	\end{table}

	\section{Conclusion}\label{sec:conclu}

	We established a finite-sample analysis of localized policy optimization that separates population value approximation from critic estimation. The direct spectral representation and the discounted-occupancy Bellman bound yield an explicit spatial and finite-feature residual, while feature excitation controls the sample budget. At fixed neighborhoods and feature dimension, the expected squared projected-gradient mapping is bounded by this residual plus an excess whose shared-oracle sample cost is inverse-squared; the same decomposition yields an offline radius-selection rule under stronger graph-tail conditions and a trajectory-conditional companion result, and the accompanying diagnostic assesses conditioning along an executed run. Numerical studies on a structured benchmark illustrate related implementations. The guarantees remain conditional on the stated excitation, decay, and smoothness assumptions; quantitative verification of these premises for coupled model families and control of conditioning as the representation grows remain open questions for future work.

	\appendices

	\section{Baseline decay by finite propagation}\label{app:lipschitz}

	\subsection*{Proof of Lemma~\ref{lem:baseline_decay}}

	\begin{proof}
	Fix $i$, $\kappa\in\mbN$, and two state-action pairs $(\bbs,\bba)$ and $(\tilde\bbs,\tilde\bba)$ agreeing on $\ccalN_i^\kappa$. Set $h:=\kappa_\pi+1$. Fix the initial state-action pairs $(\bbS^{(0)},\bbA^{(0)})=(\bbs,\bba)$ and $(\widetilde\bbS^{(0)},\widetilde\bbA^{(0)})=(\tilde\bbs,\tilde\bba)$. Because the local state and action spaces are standard Borel, each local transition kernel and each local policy kernel admits a measurable randomization by a uniform variable. For transitions at times $t\geq0$, use the same independent uniform variables in the two trajectories' local transition randomizations; for policy sampling at times $t\geq1$, use the same independent uniform variables in their local policy randomizations. Equal local transition inputs then produce equal successor states, and equal policy observation windows produce equal actions; the time-zero actions remain the prescribed actions in the definition of $Q_i^\pi$. (For the Gaussian model of (A1), common additive noise is a special case.)

	Let $\ccalR_t$ denote a radius such that the two trajectories agree on $\ccalN_i^{\ccalR_t}$ at the state-action level at time $t$; $\ccalR_0=\kappa$. The local transition randomization depends only on the one-hop state-action window, so equal inputs and the shared randomization give equal successor states: the successor states agree for all $j$ with $\ccalN_j\subseteq\ccalN_i^{\ccalR_t}$, in particular on $\ccalN_i^{\ccalR_t-1}$. The actions at time $t+1$ then agree for all $j$ with $\ccalN_j^{\kappa_\pi}\subseteq\ccalN_i^{\ccalR_t-1}$, in particular on $\ccalN_i^{\ccalR_t-1-\kappa_\pi}$. The induction gives agreement on $\ccalN_i^{\kappa-th}$ for every $t$ with $\kappa-th\geq1$ (no negative-radius neighborhood is used). The reward $r_i(\bbs_{\ccalN_i}^{(t)},\bba_{\ccalN_i}^{(t)})$ depends only on $\ccalN_i^1$, so the two reward sequences agree for $t=0,\ldots,\lceil\kappa/h\rceil-1$ (using $\lfloor(\kappa-1)/h\rfloor=\lceil\kappa/h\rceil-1$). Since both trajectories have the correct marginal laws, taking expectations gives
	\feq{
		\abs{Q_i^\pi(\bbs,\bba)-Q_i^\pi(\tilde\bbs,\tilde\bba)}
		&\leq\sum_{t\geq\lceil\kappa/h\rceil}\gamma^t\,2\bbarr
		=\frac{2\bbarr}{1-\gamma}\gamma^{\lceil\kappa/h\rceil}\notag\\
		&\leq\frac{2\bbarr}{1-\gamma}\gamma^{\kappa/h}.
	}
	With $\rho_0:=\gamma^{1/h}$ the right-hand side equals $\frac{2\bbarr}{(1-\gamma)\rho_0}\rho_0^{\kappa+1}=c_0\rho_0^{\kappa+1}$, which is (B1) with $(c_0,\rho_0)$. Only bounded rewards, one-hop transitions, and $\kappa_\pi$-hop policies were used.
	\end{proof}

	The baseline argument establishes value decay with $\rho_0=\gamma^{1/(\kappa_\pi+1)}$. It does not establish a faster spatial rate or the cross-parameter curvature decay in (D3). Such conclusions require additional quantitative bounds on spatial influence and policy sensitivity, including control of the sums over reward locations in the Hessian. We treat the stronger rate in (B1), when required by the graph-tail conditions, and the differentiated-decay condition (D3) as separate assumptions.

	\section{Derivation of $L_J=\ccalO(1)$ under the differentiated regularity (D3)}\label{app:LJ_proof}

	This appendix shows that, under Assumption (D3), the Hessian $\nabla^2 J(\bbtheta)$ has operator norm bounded independently of the network size $n$ on $\Theta_0$, which justifies treating $L_J$ as an $n$-uniform primitive in (D2). The geometric decay of the cross-parameter Hessian block $\partial^2_{\bbtheta_i\bbtheta_j}J$ in graph distance is assumed in (D3)(ii); what this appendix establishes is the implication ``(D3) $\Rightarrow$ $L_J=\ccalO(1)$,'' i.e.\ that the assumed per-block decay, once granted, sums to an $n$-uniform operator-norm bound on the Hessian by a graph row-sum argument.

	\subsection*{Step 0: Twice-differentiability for each fixed $n$}

	Write $R_t:=\sum_i r_i(\bbs^{(t)}_{\ccalN_i},\bba^{(t)}_{\ccalN_i})$, so $J(\bbtheta)=\sum_{t\geq0}\gamma^t\mbE_\bbtheta[R_t]$ with $\abs{R_t}\leq n\bbarr$. The law of the trajectory prefix $(\bbs^{(0)},\bba^{(0)},\ldots,\bbs^{(t)},\bba^{(t)})$ has density, with respect to the $\bbtheta$-free reference measure $\mu_0(d\bbs^{(0)})\,d\bba^{(0)}\,\mbP(d\bbs^{(1)}\mid\bbs^{(0)},\bba^{(0)})\,d\bba^{(1)}\cdots$ (a kernel composition), equal to $\prod_{\tau\leq t}\pi(\bba^{(\tau)}\mid\bbs^{(\tau)};\bbtheta)$. Fix $n$ and the bounded open convex product neighborhood $U=\prod_iU_i$ of (A3) and (D3)(i), on which the common support and the score bounds hold, and a reference parameter $\bar\bbtheta\in U$. Writing $p_{\bbtheta,t}$ for this prefix density, the common support and the score bound give $p_{\bbtheta,t}\leq\exp\bigl(G(t+1)\sum_{j=1}^n\norm{\bbtheta_j-\bar\bbtheta_j}\bigr)p_{\bar\bbtheta,t}\leq C_{n,t}\,p_{\bar\bbtheta,t}$ with $C_{n,t}<\infty$ uniform over $\bbtheta\in U$; for fixed $t$ this is an integrable dominating density (it may grow exponentially in $t$ and is used only for the fixed-horizon integral). Its first parameter-block derivative is bounded in norm by $(t+1)G\,p_{\bbtheta,t}$ and its $(i,j)$ second derivative by $[(t+1)^2G^2+\bbone_{\{i=j\}}(t+1)G']\,p_{\bbtheta,t}$, so differentiation under each finite-prefix integral is justified and gives, with $\bbpsi_j^{(\tau)}:=\nabla_{\bbtheta_j}\log\pi_j(\bba_j^{(\tau)}\mid\bbs^{(\tau)}_{\ccalN_j^{\kappa_\pi}})$,
	\feq{
		&\partial^2_{\bbtheta_i\bbtheta_j}\mbE_\bbtheta[R_t]
		=\mbE_\bbtheta\Bigl[R_t\Bigl(\sum_{\tau\leq t}\bbpsi_i^{(\tau)}\Bigr)\Bigl(\sum_{\tau\leq t}\bbpsi_j^{(\tau)}\Bigr)^{\!\top}\Bigr]\notag\\
		&\qquad+\bbone_{\{i=j\}}\mbE_\bbtheta\Bigl[R_t\sum_{\tau\leq t}\partial_{\bbtheta_i}\bbpsi_i^{(\tau)}\Bigr],
	}
	whose operator norm is at most $n\bbarr\,[(t+1)^2G^2+\bbone_{\{i=j\}}(t+1)G']$. The first-derivative blocks are likewise bounded by $n\bbarr(t+1)G$. Since $\sum_{t\geq0}\gamma^t(t+1)^2=(1+\gamma)/(1-\gamma)^3$ and $\sum_{t\geq0}\gamma^t(t+1)=1/(1-\gamma)^2$, the series of values, first derivatives, and second derivatives converge uniformly on $U$, so termwise differentiation gives $J\in C^2(U)$, in particular on $\Theta_0$, with
	\feq{
		\nabla^2 J(\bbtheta)\in\mbR^{nd_\theta\times nd_\theta},\qquad
		[\nabla^2 J]_{ij}:=\partial^2_{\bbtheta_i\bbtheta_j}J\in\mbR^{d_\theta\times d_\theta}.
	}
	This establishes regularity for each fixed $n$; the bound just obtained is of order $n$ and is not the $n$-uniform bound sought, which is the role of (D3)(ii). We note why termwise bounds cannot do better: differentiating the per-agent policy-gradient identity~\eqref{eq:pg_theorem_agent} in $\bbtheta_j$ produces a value-sensitivity term $\mbE_{\nu^\pi}[\sum_\ell\bbpsi_i(\nabla_{\bbtheta_j}Q_\ell^\pi)^\top]$, an occupancy-sensitivity term, and (for $j=i$) a score-derivative term $\mbE_{\nu^\pi}[\sum_\ell Q_\ell^\pi\,\partial_{\bbtheta_i}\bbpsi_i]$; the last contains the sum of all $n$ local values and can be of order $n$; separate bounds on these terms can obscure cancellations and therefore need not recover decay of the full Hessian block. (For instance, if all local rewards are identically one, $J$ is constant and its Hessian vanishes, while for a nondegenerate parameterized policy the score-derivative term alone can have size proportional to $n$.) This is why (D3)(ii) is imposed on the full block.

	\subsection*{Step 1: Block row sums under (D3)(ii)--(iii)}

	Assumption (D3)(ii) postulates
	\feq{
		\norm{[\nabla^2 J(\bbtheta)]_{ij}}\leq c_\partial\,\rho^{\,d(i,j)},\label{eq:cross_block_decay}
	}
	uniformly in $\bbtheta\in\Theta_0$. The number of agents at exact graph distance $r$ from a fixed agent is at most $\Delta(\Delta-1)^{r-1}$ for $r\geq 1$ (and $1$ for $r=0$) on a graph of maximum degree $\Delta$. Hence, for any fixed agent $i$,
	\feq{
		\sum_{j=1}^n\norm{[\nabla^2 J]_{ij}}
		&\leq c_\partial\sum_{j=1}^n\rho^{d(i,j)}\notag\\
		&\leq c_\partial\Bigl(1+\sum_{r\geq 1}\Delta(\Delta-1)^{r-1}\rho^{r}\Bigr)\notag\\
		&=c_\partial\Bigl(1+\frac{\Delta}{\Delta-1}\cdot\frac{\tilde\delta}{1-\tilde\delta}\Bigr).
	}
	The right-hand side is finite under the standing graph-tail hypothesis $\tilde\delta=(\Delta-1)\rho<1$ of (D3)(iii) and is independent of $n$, because (D3)(ii)--(iii) make $c_\partial$, $\rho$, and the degree bound $\Delta$ uniform over the network family.

	\subsection*{Step 2: Conclusion}

	The Hessian is symmetric, and a symmetric block matrix has operator norm bounded by its maximum block-row sum (block Gershgorin/Schur bound: for symmetric $\bbA$, $\norm{\bbA}\leq\sqrt{\max_i\sum_j\norm{\bbA_{ij}}\cdot\max_j\sum_i\norm{\bbA_{ij}}}=\max_i\sum_j\norm{\bbA_{ij}}$). Therefore
	\feq{
		L_J:=\sup_{\bbtheta\in\Theta_0}\norm{\nabla^2 J(\bbtheta)}_{\rm op}
		\leq c_\partial\Bigl(1+\frac{\Delta}{\Delta-1}\cdot\frac{\tilde\delta}{1-\tilde\delta}\Bigr)=\ccalO(1),
	}
	independent of $n$; equivalently $L_J\leq c_\partial(1+\Delta\rho/(1-(\Delta-1)\rho))$. Because $\Theta_0$ is convex, integrating the Hessian along the segment joining any two parameters in $\Theta_0$ shows that this bound is also a Lipschitz constant for $\nabla J$ on $\Theta_0$, which is the $n$-uniform constant invoked in (D2) and used in Corollary~\ref{cor:opt_lr} and Theorem~\ref{thm:sample_complexity_fixed}. Assumption (B1) controls only $Q_\ell^\pi$ itself and would not, on its own, yield the geometric decay of the cross-parameter Hessian block; this is precisely why (D3)(ii) is stated as a separate primitive.

	\section{Continuity of the population TD matrix and measurability of $E_M$}\label{app:M_continuity}

	\begin{lemma}\label{lem:M_continuity}
		Under Set~A, for each fixed feature realization the map $\bbtheta\mapsto\bbM_i^{\bbtheta}$ is continuous on $\Theta_0$ for every $i$; consequently $\bbtheta\mapsto\sigma_{\min}(\bbM_i^{\bbtheta})$ is continuous, the infimum in~\eqref{eq:td_stability} over the compact $\Theta_0$ is attained, and the event $E_M$ is measurable with respect to the preprocessing $\sigma$-algebra.
	\end{lemma}

	\begin{proof}
		Fix the feature realization and $i$, and write $\mathrm{TV}(P,Q):=\sup_A\abs{P(A)-Q(A)}$. By (A3) the score is bounded by $G$ on the convex parameter segment (A4). Integrating along the segment gives, for $\bbtheta,\bbtheta'\in\Theta_0$, every $\bbs$, and a.e.\ $\bba_i$ on the common support $\ccalA_i^{\rm supp}$ of (A3), $\abs{\log\pi_i(\bba_i\mid\cdot;\bbtheta_i)-\log\pi_i(\bba_i\mid\cdot;\bbtheta_i')}\leq G\norm{\bbtheta_i-\bbtheta_i'}$, whence $\int\abs{\pi_i(\cdot;\bbtheta_i)-\pi_i(\cdot;\bbtheta_i')}\,d\bba_i\leq e^{G\norm{\bbtheta_i-\bbtheta_i'}}-1$, and by subadditivity of total variation over the product policy, $\mathrm{TV}(\pi_{\bbtheta}(\cdot\mid\bbs),\pi_{\bbtheta'}(\cdot\mid\bbs))\leq\sum_j(e^{G\norm{\bbtheta_j-\bbtheta_j'}}-1)\leq n\,(e^{G\norm{\bbtheta-\bbtheta'}}-1)=:\zeta$. The transition kernel does not depend on $\bbtheta$, so a coupling of the two chains started from the same $\mu_0$ gives $\mathrm{TV}(\mbP_t^{\bbtheta},\mbP_t^{\bbtheta'})\leq t\,\zeta$ for the time-$t$ state laws; hence $\mathrm{TV}(\nu^{\pi_\bbtheta},\nu^{\pi_{\bbtheta'}})\leq(1-\gamma)\sum_{t\geq0}\gamma^t\,(t+1)\,\zeta=\zeta/(1-\gamma)$, and the law generating $(\bbZ,\bbZ')$ (one further policy factor for $\bba'$) satisfies the same bound with $\zeta/(1-\gamma)$ replaced by $\zeta/(1-\gamma)+\zeta\leq 2\zeta/(1-\gamma)$. Each entry of $\bbM_i^{\bbtheta}$ is an expectation of a function bounded by $(1+\gamma)L^2$ under this law, so
		\feq{
			\norm{\bbM_i^{\bbtheta}-\bbM_i^{\bbtheta'}}&\leq\norm{\bbM_i^{\bbtheta}-\bbM_i^{\bbtheta'}}_F\notag\\
			&\leq(m+1)\cdot 2(1+\gamma)L^2\cdot\frac{2\zeta}{1-\gamma}\xrightarrow[\bbtheta'\to\bbtheta]{}0,
		}
		proving continuity; the same argument (with the reward factor bounded by $\bbarr$) covers $\bbb_i^{\bbtheta}$. Since $\sigma_{\min}$ is $1$-Lipschitz in the operator norm and $\Theta_0$ is compact, $\inf_{\bbtheta\in\Theta_0}\sigma_{\min}(\bbM_i^{\bbtheta})$ is attained and equals the infimum over any countable dense subset $\Theta_0^{\rm d}\subset\Theta_0$; as each $\omega\mapsto\bbM_i^{\bbtheta}(\omega)$ is measurable in the feature randomness for fixed $\bbtheta$, the event $E_M=\bigcap_{i}\bigcap_{\bbtheta\in\Theta_0^{\rm d}}\{\sigma_{\min}(\bbM_i^{\bbtheta})\geq\mu_M(m)\}$ is measurable. The same continuity in $\bbtheta$ and measurability in the feature realization imply joint measurability on the separable parameter space; consequently $\bbM_i^{\bbtheta^{(k)}}$ and its smallest singular value are $\ccalF_k^-$-measurable whenever $\bbtheta^{(k)}$ is. The thresholded inverse/zero map of Definition~\ref{def:lstd} is measurable (inversion is continuous on the open acceptance set, the zero fallback is measurable on its complement, and Euclidean projection is continuous), which supports the induction defining the random iterates.
	\end{proof}

	\section{Proof of Theorem~\ref{thm:rff_approx} (RFF approximation)}\label{app:rff_proof}

	\begin{proof}
	Fix an agent $i$ and abbreviate $d:=d_{i,\kappa}\leq d_\ast:=D_\kappa d_S$, $g:=g_\alpha^{(i)}$, $p:=p_\alpha^{(i)}$, $\delta:=\delta_{\rm rff}$, and $\bbu(\bbZ):=\bbf_{i,\kappa}(\bbZ)/\sqrt{1-\alpha^2}$. Writing $A_\ell:=\bbomega_{i,\ell}^\top\bbu(\bbZ)+b_{i,\ell}$ and $B_\ell(\bby):=\sqrt{1-\alpha^2}\,\bbomega_{i,\ell}^\top\bby+b_{i,\ell}$, the raw feature inner product is
	\feq{
		\inner{\phi_{i,\kappa}^{\rm raw}(\bbZ)}{\mu_{i,\kappa}^{\rm raw}(\bby)}&=g(\bbZ)\,p(\bby)\,\hat k_m(\bbu(\bbZ),\sqrt{1-\alpha^2}\,\bby),\notag\\
		\hat k_m(\bbx,\bbx'')&:=\frac{2}{m}\sum_{\ell=1}^m\cos(\bbomega_\ell^\top\bbx+b_\ell)\\
		&\qquad\qquad\quad\times\cos(\bbomega_\ell^\top\bbx''+b_\ell),
	}
	with $\abs{\hat k_m}\leq 2$ pointwise. By the product-to-sum identity, $2\cos A\cos B=\cos(A-B)+\cos(A+B)$; the sum-frequency term carries the phase $2b_{i,\ell}$ and has mean zero, while the difference term has mean, by the Gaussian characteristic function of $\bbomega\sim\ccalN(\bbzero,\sigma^{-2}\bbI_d)$, equal to $k_0(\bbx-\bbx'')$. Hence $\mbE\,\hat k_m(\bbx,\bbx'')=k_0(\bbx-\bbx'')$ and, by the splitting identity~\eqref{eq:p_splitting_identity}, the estimator is pointwise unbiased for $\mbP_{i,\kappa}(\bby\mid\bbZ)$. The sum-frequency term is a nonzero random process, so $\hat k_m$ is not a function of $\bbx-\bbx''$ alone; the uniform bound below is therefore proved for the two-argument estimator directly.

	\emph{Step 1 (uniform bound on the compact core).} Set $t:=\epsilon_P/(2\tilde g_\alpha)$ and $R:=D_\ccalR^{\rm int}(\epsilon_P)$. By (A1), $\norm{\bbu(\bbZ)}\leq\sqrt{D_\kappa}B_f/\sqrt{1-\alpha^2}\leq R$ for every $\bbZ$, and $\norm{\sqrt{1-\alpha^2}\,\bby}\leq R_2(\epsilon_P)\leq R$ whenever $\norm{\bby}\leq R_2(\epsilon_P)$; both kernel arguments therefore lie in the closed ball $\ccalB:=\{\bbx\in\mbR^d:\norm{\bbx}\leq R\}$. We show that, with probability at least $1-\delta/n$ over $\{(\bbomega_{i,\ell},b_{i,\ell})\}_\ell$,
	\feq{
		\sup_{\bbx,\bbx''\in\ccalB}\abs{\hat k_m(\bbx,\bbx'')-k_0(\bbx-\bbx'')}\leq t.\label{eq:core_uniform}
	}

	\emph{(1a) Frequency-norm event.} Let $B_\delta:=2n\sqrt{d_\ast}/(\sigma\delta)$ and $F_i:=\{\tfrac1m\sum_\ell\norm{\bbomega_{i,\ell}}\leq B_\delta\}$. Since $\mbE\norm{\bbomega}\leq(\mbE\norm{\bbomega}^2)^{1/2}=\sqrt d/\sigma\leq\sqrt{d_\ast}/\sigma$, Markov's inequality gives $\mbP(F_i^c)\leq(\sqrt{d_\ast}/\sigma)/B_\delta=\delta/(2n)$.

	\emph{(1b) Lipschitz control.} On $F_i$, differentiating in $\bbx$ gives $\norm{\nabla_\bbx\hat k_m}\leq\tfrac2m\sum_\ell\norm{\bbomega_\ell}\leq2B_\delta$, and likewise in $\bbx''$. The Gaussian kernel satisfies $\norm{\nabla k_0(\bbDelta)}=(\norm{\bbDelta}/\sigma^2)e^{-\norm{\bbDelta}^2/(2\sigma^2)}\leq1/(\sigma\sqrt e)\leq1/\sigma$. Hence $f(\bbx,\bbx''):=\hat k_m(\bbx,\bbx'')-k_0(\bbx-\bbx'')$ is $K_\delta$-Lipschitz in each argument on $F_i$, with $K_\delta=2B_\delta+\sigma^{-1}$ as in the theorem statement.

	\emph{(1c) Net.} Let $a:=t/(4K_\delta)$ and let $\ccalN_a\subset\ccalB$ be an $a$-net; a volumetric bound gives $\abs{\ccalN_a}\leq(1+2R/a)^d=(1+8RK_\delta/t)^d$. The net is deterministic, since $R$ and $K_\delta$ are. For any $\bbx,\bbx''\in\ccalB$ choose net points $\tilde\bbx,\tilde\bbx''$ within distance $a$; on $F_i$, $\abs{f(\bbx,\bbx'')-f(\tilde\bbx,\tilde\bbx'')}\leq2K_\delta a=t/2$.

	\emph{(1d) Concentration at net pairs.} For fixed $(\tilde\bbx,\tilde\bbx'')$, $\hat k_m(\tilde\bbx,\tilde\bbx'')$ is an average of $m$ i.i.d.\ summands $2\cos(\bbomega_\ell^\top\tilde\bbx+b_\ell)\cos(\bbomega_\ell^\top\tilde\bbx''+b_\ell)\in[-2,2]$ with mean $k_0(\tilde\bbx-\tilde\bbx'')$, so Hoeffding's inequality gives $\mbP(\abs{f(\tilde\bbx,\tilde\bbx'')}\geq t/2)\leq2\exp(-2m(t/2)^2/4^2)=2\exp(-mt^2/32)$. A union bound over the at most $\abs{\ccalN_a}^2$ net pairs bounds the probability that some net pair deviates by $t/2$ by $2(1+8RK_\delta/t)^{2d}\exp(-mt^2/32)$; this is at most $\delta/(2n)$ provided
	\feq{
		m\geq\frac{32}{t^2}\Bigl[\log\frac{4n}{\delta}+2d\log\Bigl(1+\frac{8RK_\delta}{t}\Bigr)\Bigr].\label{eq:m_condition_raw}
	}
	With $t=\epsilon_P/(2\tilde g_\alpha)$ we have $32/t^2=128\tilde g_\alpha^2/\epsilon_P^2$ and $8RK_\delta/t=16RK_\delta\tilde g_\alpha/\epsilon_P$; since $4n/\delta\geq1$ and $1+X\geq1$, $\log(4n/\delta)+2d\log(1+X)\leq(2d+1)\log(4n(1+X)/\delta)$, and the right-hand side of~\eqref{eq:m_condition_raw} is nondecreasing in $d$. Thus~\eqref{eq:rff_feasibility}, evaluated at the worst case $d=d_\ast$ with $C_{\rm RFF}(d_\ast)=128(2d_\ast+1)$, implies~\eqref{eq:m_condition_raw} for every agent. On the intersection of $F_i$ and the complement of the net-deviation event, which has probability at least $1-\delta/n$ by a union bound (no independence between the two events is used), every pair $(\bbx,\bbx'')\in\ccalB^2$ satisfies $\abs{f(\bbx,\bbx'')}\leq t/2+t/2=t$, which is~\eqref{eq:core_uniform}. Call this event $E_{\rm rff}^{(i)}$.

	\emph{Step 2 (weighted core error).} On $E_{\rm rff}^{(i)}$, for every $\bbZ$ (since $\bbu(\bbZ)\in\ccalB$ always, and $\sqrt{1-\alpha^2}\,\bby\in\ccalB$ for $\norm{\bby}\leq R_2$),
	\feq{
		\int_{\norm{\bby}\leq R_2}g\,p\,\abs{k_0-\hat k_m}\,d\bby&\leq t\,\bbarg_\alpha(\kappa)\Bigl(\frac{2\pi\sigma^2}{\alpha^2}\Bigr)^{d/2}\notag\\
		&\leq t\,\tilde g_\alpha\;\leq\;\frac{\epsilon_P}{2},
	}
	using $\int p_\alpha^{(i)}=(2\pi\sigma^2/\alpha^2)^{d/2}$, the derived bound~\eqref{eq:g_alpha_bound}, and $\bbarg_\alpha(2\pi\sigma^2/\alpha^2)^{d/2}\leq\tilde g_\alpha$ (both cases of the maximum in the definition of $\tilde g_\alpha$, since $d\leq D_\kappa d_S$).

	\emph{Step 3 (true-kernel tail).} Under $\mbP_{i,\kappa}(\cdot\mid\bbZ)$ the successor is $\bby=\bbf_{i,\kappa}(\bbZ)+\bbvarepsilon$ with $\bbvarepsilon\sim\ccalN(\bbzero,\sigma^2\bbI_{d})$. For $\bbg\sim\ccalN(\bbzero,\bbI_d)$ and $u>0$, the $\chi^2$ tail bound $\mbP(\norm{\bbg}^2>d+2\sqrt{du}+2u)\leq e^{-u}$ of Laurent and Massart~\cite{laurent2000adaptive} (it follows from $\log\mbE e^{t(\norm{\bbg}^2-d)}=-td-\tfrac d2\log(1-2t)\leq dt^2/(1-2t)$ for $0<t<1/2$ and the Chernoff bound at $t=\sqrt u/(\sqrt d+2\sqrt u)$), together with $(\sqrt d+\sqrt{2u})^2\geq d+2\sqrt{du}+2u$, gives $\mbP(\norm{\bbg}>\sqrt d+\sqrt{2u})\leq e^{-u}$. Since $\alpha\in(0,1)$ and $d\leq D_\kappa d_S$, $R_2(\epsilon_P)\geq\sqrt{D_\kappa}B_f+\sigma(\sqrt d+\sqrt{2\log(8\tilde g_\alpha/\epsilon_P)})$, so $\int_{\norm{\bby}>R_2}\mbP_{i,\kappa}(\bby\mid\bbZ)\,d\bby\leq\epsilon_P/(8\tilde g_\alpha)\leq\epsilon_P/8$, using $\tilde g_\alpha\geq 1$.

	\emph{Step 4 (signed-approximant tail).} Pointwise, $\abs{\inner{\phi^{\rm raw}}{\mu^{\rm raw}}}\leq 2\,g(\bbZ)\,p(\bby)\leq 2\,\bbarg_\alpha\,p_\alpha^{(i)}(\bby)$: the approximant is signed and need not integrate to one, but it is dominated by twice the Gaussian envelope. As $p_\alpha^{(i)}$ is the unnormalized density of $\ccalN(\bbzero,(\sigma^2/\alpha^2)\bbI_{d})$ with total mass $(2\pi\sigma^2/\alpha^2)^{d/2}$, and $R_2(\epsilon_P)\geq(\sigma/\alpha)(\sqrt d+\sqrt{2\log(8\tilde g_\alpha/\epsilon_P)})$, writing $\bbxi\sim\ccalN(\bbzero,(\sigma^2/\alpha^2)\bbI_d)$ and using the same Gaussian tail,
	\feq{
		\int_{\norm{\bby}>R_2}\abs{\inner{\phi^{\rm raw}}{\mu^{\rm raw}}}\,d\bby&\leq 2\,\bbarg_\alpha\Bigl(\frac{2\pi\sigma^2}{\alpha^2}\Bigr)^{d/2}\mbP_{\bbxi}\bigl(\norm{\bbxi}>R_2\bigr)\notag\\
		&\leq 2\,\tilde g_\alpha\cdot\frac{\epsilon_P}{8\,\tilde g_\alpha}=\frac{\epsilon_P}{4}.
	}
	Combining Steps 2--4 by the triangle inequality on $\{\norm{\bby}\leq R_2\}$ and $\{\norm{\bby}>R_2\}$, for every $\bbZ$ on $E_{\rm rff}^{(i)}$,
	$\norm{\mbP_{i,\kappa}(\cdot\mid\bbZ)-\inner{\phi_{i,\kappa}^{\rm raw}(\bbZ)}{\mu_{i,\kappa}^{\rm raw}(\cdot)}}_{\ccalL^1}\leq\tfrac{\epsilon_P}{2}+\tfrac{\epsilon_P}{8}+\tfrac{\epsilon_P}{4}<\epsilon_P$.

	\emph{Step 5 (agents, heterogeneous dimensions, normalization).} A union bound over $i\in[n]$ at confidence $\delta/n$ each gives $\mbP(\bigcap_i E_{\rm rff}^{(i)})\geq 1-\delta$. The normalization of Definition~\ref{def:rff} preserves the inner product ($\hhatphi=\phi^{\rm raw}/\bbarg_\alpha$, $\hhatmu=\mu^{\rm raw}\bbarg_\alpha$), so the same bound holds for $\inner{\hhatphi_{i,\kappa}}{\hhatmu_{i,\kappa}}$; hence $\bigcap_iE_{\rm rff}^{(i)}\subseteq E_{\rm rff}$ and $\mbP(E_{\rm rff})\geq1-\delta$. The event $E_{\rm rff}$ is measurable: for each fixed feature realization the displayed $\ccalL^1$ error is continuous in $\bbZ$ (the drift and the splitting factor are continuous, and the integrand is dominated by the integrable envelope $3\bbarg_\alpha p_\alpha^{(i)}$), the input space $\ccalS_{\ccalN_i^{\kappa+1}}\times\ccalA_{\ccalN_i^{\kappa+1}}$ is separable, so the supremum equals a supremum over a fixed countable dense subset, and each such error is measurable in the feature draw. Finally, the right-hand side of~\eqref{eq:rff_feasibility} is strictly decreasing in $\epsilon_P$ because $\tilde g_\alpha^2/\epsilon_P^2$, $D_\ccalR^{\rm int}(\epsilon_P)$ (through $R_2$), and $\tilde g_\alpha/\epsilon_P$ are all decreasing in $\epsilon_P$; and solving~\eqref{eq:rff_feasibility} for $\epsilon_P$ at given $m$ yields $\epsilon_P(m)^2=\widetilde\ccalO(C_{\rm RFF}(d_\ast)\tilde g_\alpha^2/m)$, with a polylogarithmic factor in $m$, $n/\delta$, and the primitives.
	\end{proof}

	\section{Proof of Lemma~\ref{lem:V_decay} (value decay)}\label{app:V_decay}

	\begin{proof}
	Let $\bbs,\bbs'$ agree on $\ccalN_i^\kappa$. Couple two action draws $\bba\sim\pi(\cdot\mid\bbs)$, $\bba'\sim\pi(\cdot\mid\bbs')$ so that $\bba_j=\bba_j'$ for every $j$ with $\ccalN_j^{\kappa_\pi}\subseteq\ccalN_i^\kappa$; this is possible because $\pi$ factorizes and each $\pi_j$ depends only on $\bbs_{\ccalN_j^{\kappa_\pi}}$, which coincide for such $j$. The set of agents $j$ with $\ccalN_j^{\kappa_\pi}\subseteq\ccalN_i^\kappa$ contains $\ccalN_i^{\kappa-\kappa_\pi}$. Hence $(\bbs,\bba)$ and $(\bbs',\bba')$ agree on $\ccalN_i^{\kappa-\kappa_\pi}$ at the state-action level. By the $V$-to-$Q$ relation $V_i^\pi(\bbs)=\mbE_{\bba\sim\pi(\cdot\mid\bbs)}[Q_i^\pi(\bbs,\bba)]$ and (B1) applied at radius $\kappa-\kappa_\pi\geq1$,
	\feq{
		\abs{V_i^\pi(\bbs)-V_i^\pi(\bbs')}\leq c\rho^{(\kappa-\kappa_\pi)+1}=c\rho^{-\kappa_\pi}\rho^{\kappa+1}.
	}
	Thus any $\tilde c\geq c\rho^{-\kappa_\pi}$ is admissible; we fix $\tilde c:=2c\rho^{-\kappa_\pi}$ throughout.
	\end{proof}

	\section{Proof of Theorem~\ref{thm:q_linear} (linear representation)}\label{app:q_linear_proof}

	\begin{proof}
	\emph{Step 1: $\ccalL^1$-mass of $\hhatmu_{i,\kappa}$ and the coefficient bound $U_1$.} By Definition~\ref{def:rff}, $\hhatmu_{i,\kappa}(\bby)=\bbarg_\alpha\,p_\alpha^{(i)}(\bby)\sqrt{2/m}\,\{\cos(\sqrt{1-\alpha^2}\,\bbomega_{i,\ell}^\top\bby+b_{i,\ell})\}_{\ell=1}^m$; since each of the $m$ cosines lies in $[-1,1]$, the cosine vector has Euclidean norm at most $\sqrt m$, so $\norm{\hhatmu_{i,\kappa}(\bby)}\leq\sqrt2\,\bbarg_\alpha\,p_\alpha^{(i)}(\bby)$ pointwise. The damping factor $p_\alpha^{(i)}(\bby)=\exp(-\alpha^2\norm{\bby}^2/(2\sigma^2))$ is Lebesgue-integrable, with $\int_{\mbR^{d_{i,\kappa}}}p_\alpha^{(i)}(\bby)\,d\bby=(2\pi\sigma^2/\alpha^2)^{d_{i,\kappa}/2}$. Hence, using $d_{i,\kappa}\leq D_\kappa d_S$ and the definition of $\tilde g_\alpha$,
	\feq{
		\int\norm{\hhatmu_{i,\kappa}(\bby)}\,d\bby\leq\sqrt2\,\bbarg_\alpha\,(2\pi\sigma^2/\alpha^2)^{d_{i,\kappa}/2}\leq\sqrt2\,\tilde g_\alpha.\label{eq:mu_L1_mass}
	}
	Since $\abs{h_i^\pi}\leq Q_{\max}=\bbarr/(1-\gamma)$, the RFF block of $\bbu_{i,\kappa}^\pi$ is a well-defined Bochner integral of norm at most $\sqrt2\gamma\tilde g_\alpha\bbarr/(1-\gamma)$, and $\norm{\bbu_{i,\kappa}^\pi}\leq1+\sqrt2\gamma\tilde g_\alpha\bbarr/(1-\gamma)=U_1$ by $\norm{(1,\bbv)}\leq1+\norm{\bbv}$.

	\emph{Step 2: replacing the full continuation by the local continuation.} Write $I=\ccalN_i^\kappa$. The genuine Bellman equation reads $Q_i^\pi(\bbs,\bba)=r_i(\bbs_{\ccalN_i},\bba_{\ccalN_i})+\gamma\int V_i^\pi(\bbs')\,\mbP(d\bbs'\mid\bbs,\bba)$. For every full successor state $\bbs'$, the states $\bbs'$ and $(\bbs'_I,\bbs^\circ_{I^c})$ agree on $\ccalN_i^\kappa$, so Lemma~\ref{lem:V_decay} (which requires $\kappa\geq\kappa_\pi+1$) gives $\abs{V_i^\pi(\bbs')-h_i^\pi(\bbs'_I)}\leq\tilde c\rho^{\kappa+1}$. Moreover, the marginal law of $\bbs'_I$ under $\mbP(\cdot\mid\bbs,\bba)$ is $\prod_{j\in I}\mbP_j(\cdot\mid\bbs_{\ccalN_j},\bba_{\ccalN_j})=\mbP_{i,\kappa}(\cdot\mid\bbZ)$, because $\ccalN_j\subseteq\ccalN_i^{\kappa+1}$ for $j\in I$. Hence
	\feq{
		Q_i^\pi(\bbs,\bba)&=r_i(\bbs_{\ccalN_i},\bba_{\ccalN_i})+\gamma\!\int\!h_i^\pi(\bby)\,\mbP_{i,\kappa}(d\bby\mid\bbZ)\\
		&\quad+e_{{\rm loc},i}(\bbs,\bba),\label{eq:local_bellman}
	}
	with $\abs{e_{{\rm loc},i}}\leq\gamma\tilde c\rho^{\kappa+1}=B_{\rm loc}(\kappa)$ for every $(\bbs,\bba)$. The reference completion is used only inside $h_i^\pi$; the transition law is the true one.

	\emph{Step 3: finite-feature approximation of the local continuation.} On $E_{\rm rff}$, for every $\bbZ$,
	\feq{
		&\Bigl|\gamma\!\int\!h_i^\pi\,d\mbP_{i,\kappa}(\cdot\mid\bbZ)-\gamma\!\int\!h_i^\pi\inner{\hhatphi_{i,\kappa}(\bbZ)}{\hhatmu_{i,\kappa}}\,d\bby\Bigr|\notag\\
		&\quad\leq\gamma\norm{h_i^\pi}_\infty\,\norm{\mbP_{i,\kappa}(\cdot\mid\bbZ)-\inner{\hhatphi_{i,\kappa}(\bbZ)}{\hhatmu_{i,\kappa}(\cdot)}}_{\ccalL^1}\notag\\
		&\quad\leq\frac{\gamma\bbarr\epsilon_P}{1-\gamma}=B_{\rm rff}(m).
	}
	By~\eqref{eq:mu_L1_mass} and Fubini, $\gamma\int h_i^\pi(\bby)\inner{\hhatphi_{i,\kappa}(\bbZ)}{\hhatmu_{i,\kappa}(\bby)}\,d\bby=\bigl\langle\hhatphi_{i,\kappa}(\bbZ),\gamma\!\int\!\hhatmu_{i,\kappa}(\bby)h_i^\pi(\bby)\,d\bby\bigr\rangle$, which is exactly the inner product of the RFF block of $\tdbbphi_{i,\kappa}^{\rm raw}(\bbZ)$ with the RFF block of $\bbu_{i,\kappa}^\pi$; the discount factor $\gamma$ is carried inside $\bbu_{i,\kappa}^\pi$ and appears once. The first coordinates contribute $r_i(\bbs_{\ccalN_i},\bba_{\ccalN_i})\cdot1$. Hence $\inner{\tdbbphi_{i,\kappa}^{\rm raw}(\bbZ)}{\bbu_{i,\kappa}^\pi}=r_i+\gamma\int h_i^\pi\inner{\hhatphi_{i,\kappa}(\bbZ)}{\hhatmu_{i,\kappa}}\,d\bby$, and combining with~\eqref{eq:local_bellman},
	\feq{
		Q_i^\pi(\bbs,\bba)&=\inner{\tdbbphi_{i,\kappa}^{\rm raw}(\bbZ)}{\bbu_{i,\kappa}^\pi}+e_i(\bbs,\bba),\notag\\
		\abs{e_i}&\leq B_{\rm loc}(\kappa)+B_{\rm rff}(m)=\Bk,
	}
	for every $(\bbs,\bba)$, which is~\eqref{eq:q_linear}--\eqref{eq:xi_bound}. The event $E_{\rm rff}$ does not depend on $\pi$; only $h_i^\pi$ and hence $\bbu_{i,\kappa}^\pi$ do.
	\end{proof}

	\section{Proof of Lemma~\ref{lem:pop_L2} (projected Bellman characterization)}\label{app:pop_L2}

	\begin{proof}
	Work on $E_M\cap E_{\rm rff}$, fix $i,\bbtheta$ and the feature realization, and write $\nu=\nu^\pi$, $\tdbbphi=\tdbbphi_{i,\kappa}^{\rm raw}$, $\bbM=\bbM_i^{\bbtheta}$, $\bbH=\bbH_i^{\bbtheta}$, $\bbb=\bbb_i^{\bbtheta}$. Let $\nu_+$ be the law of the successor pair $(\bbS',\bbA')$ under the one-step law of (D1); as in the proof of Lemma~\ref{lem:excitation}, the occupancy identity gives $\nu_+\leq\nu/\gamma$.

	\emph{(a) Contraction.} If $f=g$ $\nu$-a.e.\ then $f=g$ $\nu_+$-a.e.\ by domination, so $\ccalP^\pi f=\ccalP^\pi g$ $\nu$-a.e.; thus $\ccalP^\pi$ is well defined on $\ccalL^2(\nu)$ equivalence classes. By conditional Jensen, $(\ccalP^\pi f)^2\leq\ccalP^\pi(f^2)$ $\nu$-a.e., hence
	\feq{
		\norm{\gamma\ccalP^\pi f}_{\ccalL^2(\nu)}^2\leq\gamma^2\!\int\!\ccalP^\pi(f^2)\,d\nu=\gamma^2\!\int\!f^2\,d\nu_+\leq\gamma\!\int\!f^2\,d\nu,
	}
	where the middle equality is the tower property under the one-step law and the last step is domination. So $\norm{\gamma\ccalP^\pi}_{\ccalL^2(\nu)\to\ccalL^2(\nu)}\leq\sqrt\gamma<1$.

	\emph{(b) Normal equations.} For $q=\inner{\tdbbphi}{\bbw}\in\ccalV_i$ and $T^\pi q:=r_i+\gamma\ccalP^\pi q$, the residual $T^\pi q-q$ is orthogonal in $\ccalL^2(\nu)$ to every coordinate $\tdbbphi_j$ if and only if
	\feq{
		&\mbE_\nu\bigl[\tdbbphi\,(r_i+\gamma\,\mbE[\tdbbphi(\bbZ')^\top\bbw\mid\bbS,\bbA]-\tdbbphi^\top\bbw)\bigr]\notag\\
		&\qquad=\bbb+\gamma\bbC_i^{\bbtheta}\bbw-\bbH\bbw=\bbb-\bbM\bbw=\bbzero,
	}
	where the tower step $\mbE[\tdbbphi\,\mbE[\tdbbphi'^\top\mid\bbS,\bbA]]=\mbE[\tdbbphi\tdbbphi'^\top]$ is valid because $\tdbbphi(\bbZ)$ is measurable in $(\bbS,\bbA)$. Orthogonality of the residual to $\ccalV_i$ is precisely the statement $q=\Pi_iT^\pi q$. On $E_M$, $\bbM$ is invertible, so $\bbw_i^\ast=\bbM^{-1}\bbb$ is the unique coefficient solving the normal equations, and $q_i^\ast=\inner{\tdbbphi}{\bbw_i^\ast}$ is a fixed point of $\Pi_iT^\pi$. Since $\Pi_i$ is nonexpansive and $T^\pi$ is a $\sqrt\gamma$-contraction on $\ccalL^2(\nu)$ by (a), $\Pi_iT^\pi$ has a unique fixed point in $\ccalL^2(\nu)$. Finally, $\bbH\succ\bbzero$: if $\bbv^\top\bbH\bbv=0$ then $\bbv^\top\tdbbphi=0$ $\nu$-a.s., whence $\bbv^\top\bbM=\mbE[(\bbv^\top\tdbbphi)(\tdbbphi-\gamma\tdbbphi')^\top]=\bbzero$, contradicting invertibility; so $\bbw\mapsto\inner{\tdbbphi}{\bbw}$ is injective from $\mbR^{m+1}$ into $\ccalL^2(\nu)$ and the coefficient and function formulations agree.

	\emph{(c) Population prediction error.} Set $e:=Q_i^\pi-q_i^\ast$. Since $Q_i^\pi=T^\pi Q_i^\pi$ (the genuine Bellman equation, with $\ccalP^\pi Q_i^\pi(\bbs,\bba)=\mbE[V_i^\pi(\bbS')\mid\bbs,\bba]$) and $q_i^\ast=\Pi_iT^\pi q_i^\ast$,
	\feq{
		\Pi_ie=\Pi_iT^\pi Q_i^\pi-\Pi_iT^\pi q_i^\ast=\gamma\,\Pi_i\ccalP^\pi e.
	}
	As $q_i^\ast\in\ccalV_i$, $(I-\Pi_i)e=(I-\Pi_i)Q_i^\pi$, and orthogonality gives
	\feq{
		\norm{e}_{\ccalL^2(\nu)}^2&=\norm{(I-\Pi_i)Q_i^\pi}^2_{\ccalL^2(\nu)}+\norm{\Pi_ie}_{\ccalL^2(\nu)}^2\notag\\
		&\leq\dist_{\ccalL^2(\nu)}(Q_i^\pi,\ccalV_i)^2+\gamma\norm{e}_{\ccalL^2(\nu)}^2,
	}
	using $\norm{\Pi_i\gamma\ccalP^\pi e}\leq\norm{\gamma\ccalP^\pi e}\leq\sqrt\gamma\norm{e}$ from (a). Rearranging yields the first inequality of~\eqref{eq:pop_L2}. For the second, Theorem~\ref{thm:q_linear} gives $\abs{Q_i^\pi-\inner{\tdbbphi}{\bbu_{i,\kappa}^\pi}}\leq\Bk$ pointwise, so $\dist_{\ccalL^2(\nu)}(Q_i^\pi,\ccalV_i)\leq\norm{Q_i^\pi-\inner{\tdbbphi}{\bbu_{i,\kappa}^\pi}}_{\ccalL^2(\nu)}\leq\Bk$.

	\emph{(d) Coefficient and pointwise comparison.} Substitute the representation $Q_i^\pi=\inner{\tdbbphi}{\bbu_{i,\kappa}^\pi}+e_i$ of Theorem~\ref{thm:q_linear} into the Bellman equation $Q_i^\pi(\bbS,\bbA)=r_i+\gamma\,\mbE[Q_i^\pi(\bbS',\bbA')\mid\bbS,\bbA]$, multiply by $\tdbbphi(\bbZ)$, which is $(\bbS,\bbA)$-measurable, and take $\mbE_\nu$; the tower property gives
	\feq{
		\bbM\bbu_{i,\kappa}^\pi=\bbb+\mbE_\nu\bigl[\tdbbphi(\bbZ)\bigl(\gamma\,e_i(\bbS',\bbA')-e_i(\bbS,\bbA)\bigr)\bigr].
	}
	The last term has norm at most $(1+\gamma)L\Bk$ by Jensen for the Bochner integral and the pointwise bounds. Hence $\bbw_i^\ast-\bbu_{i,\kappa}^\pi=-\bbM^{-1}\mbE_\nu[\tdbbphi(\gamma e_i'-e_i)]$ has norm at most $(1+\gamma)L\Bk/\mu_M$, and $\abs{q_i^\ast-Q_i^\pi}\leq\abs{\inner{\tdbbphi}{\bbw_i^\ast-\bbu_{i,\kappa}^\pi}}+\abs{e_i}\leq L\cdot(1+\gamma)L\Bk/\mu_M+\Bk=(1+C')\Bk$ pointwise.

	\emph{(e) Radius.} $\norm{\bbb}=\norm{\mbE_\nu[\tdbbphi\,r_i]}\leq L\bbarr$ and $\norm{\bbM^{-1}}\leq1/\mu_M$ on $E_M$, so $\norm{\bbw_i^\ast}\leq L\bbarr/\mu_M$.
	\end{proof}

	\section{Proof of Theorem~\ref{thm:lstd_error} (LSTD concentration)}\label{app:lstd_concentration}

	\begin{proof}
	Work on $E_M\cap E_{\rm rff}$ and fix $(i,k)$; write $\bbtheta=\bbtheta^{(k)}$, $\bbM=\bbM_i^{\bbtheta}$, $\bbM^{(k)}=\bbM_i^{(k)}$, $\bbb^{(k)}=\bbb_i^{(k)}$, $\bbw^\ast=\bbw_i^\ast$, and $u:=\log(4(m+1)/\delta)$. Conditional on $\ccalF_k^-$, the $M_s$ transitions in $\ccalD_s^{(k)}$ are i.i.d.\ from the one-step law of (D1).

	\emph{Step 1: matrix-Bernstein for $\bbM^{(k)}-\bbM$.} Each summand $\bbY_j:=\tdbbphi^{\rm raw}(\bbz^{(j)})(\tdbbphi^{\rm raw}(\bbz^{(j)})-\gamma\tdbbphi^{\rm raw}(\bbz^{\prime(j)}))^\top$ is a (generally non-symmetric) $(m{+}1)\times(m{+}1)$ matrix bounded in operator norm by $B:=(1+\gamma)L^2\leq2L^2$, with conditional mean $\bbM$. The centered summands $\bbS_j:=\bbY_j-\bbM$ are independent conditional on $\ccalF_k^-$, satisfy $\mbE\bbS_j=\bbzero$, $\norm{\bbS_j}\leq2B$, and $\max\{\norm{\mbE\bbS_j\bbS_j^\top},\norm{\mbE\bbS_j^\top\bbS_j}\}\leq B^2$ (subtracting the positive semidefinite mean outer product from each second moment). The rectangular matrix-Bernstein inequality~\cite[Thm.~6.1.1]{tropp2015introduction}, applied through the Hermitian dilation $\mathcal{H}(\bbX):=\begin{psmallmatrix}\bbzero&\bbX\\\bbX^\top&\bbzero\end{psmallmatrix}$ of dimension $2(m+1)$, gives for every $s>0$
	\feq{
		\mbP\Bigl(\norm{\bbM^{(k)}-\bbM}\geq s\,\Big|\,\ccalF_k^-\Bigr)&\leq2(m+1)\\
		&\quad\times\exp\Bigl(-\frac{M_s s^2/2}{B^2+2Bs/3}\Bigr).
	}
	Hence, with probability $\geq1-\delta/2$,
	\feq{
		\norm{\bbM^{(k)}-\bbM}\leq r_\delta:=B\Bigl(\sqrt{\frac{2u}{M_s}}+\frac{4u}{3M_s}\Bigr).\label{eq:rM_explicit}
	}
	For $M_s\geq C_0L^4u/\mu_M^2$ with $C_0=512$ one has $2L^2\sqrt{2u/M_s}\leq\mu_M/8$ and, since $\mu_M\leq(1+\gamma)L^2/(m+1)\leq2L^2$ by Proposition~\ref{prop:mu_M_obstruction}, also $2L^2\cdot4u/(3M_s)\leq\mu_M/8$; thus $r_\delta\leq\mu_M/4$.

	\emph{Step 2: dimension-free vector concentration for the centered target.} Define the centered vector $\bbzeta_i^{(k)}:=\bbb^{(k)}-\bbM^{(k)}\bbw^\ast$, whose conditional mean is $\bbb_i^{\bbtheta}-\bbM\bbw^\ast=\bbzero$ by the definition of $\bbw^\ast$. Each i.i.d.\ summand $\bbY_j^{\rm vec}:=\tdbbphi^{\rm raw}(\bbz^{(j)})\bigl(r_i^{(j)}-(\tdbbphi^{\rm raw}(\bbz^{(j)})-\gamma\tdbbphi^{\rm raw}(\bbz^{\prime(j)}))^\top\bbw^\ast\bigr)$ of $M_s\bbzeta_i^{(k)}$ is a mean-zero random vector in the finite-dimensional Hilbert space $\mbR^{m+1}$, bounded in norm by $B_Y=L(\bbarr+(1+\gamma)LW^\ast)$ (using $\norm{\tdbbphi^{\rm raw}}\leq L$, $\abs{r_i}\leq\bbarr$, $\norm{\bbw^\ast}\leq W^\ast$). Conditionally on $\ccalF_k^-$ the centered summands are independent, hence martingale differences with the deterministic increment bound $B_Y/M_s$ for the normalized sum. The Hilbert-space Hoeffding-type inequality of Pinelis~\cite[Thm.~3.5]{pinelis1994optimum} (a Hilbert space is $(2,1)$-smooth) gives $\mbP(\norm{\bbzeta_i^{(k)}}\geq s\mid\ccalF_k^-)\leq2\exp(-s^2M_s/(2B_Y^2))$, so with probability $\geq 1-\delta/2$,
	\feq{
		\norm{\bbzeta_i^{(k)}}\leq B_Y\sqrt{\frac{2\log(4/\delta)}{M_s}}\leq C_vB_Y\sqrt{\frac{u}{M_s}},\qquad C_v=\sqrt2.
	}
	This bound carries no $\sqrt{m}$ or $\log(m+1)$ dimensional factor; the theorem writes $\log(4(m+1)/\delta)$ in place of $\log(4/\delta)$ only so that both terms share one logarithmic envelope.

	\emph{Step 3: assembling the LSTD error.} On the intersection of the two favorable events of Steps 1--2 (probability $\geq 1-\delta$), Lemma~\ref{lem:excitation}(a) and the reverse triangle inequality give, for every $\lambda\geq0$, $\sigma_{\min}(\bbM^{(k)}+\lambda\bbI)\geq\sigma_{\min}(\bbM+\lambda\bbI)-\norm{\bbM^{(k)}-\bbM}\geq\mu_M-\mu_M/4\geq\mu_M/2$, so the lower bound exceeds $\tau_{\rm sv}$ and the thresholded fallback of Definition~\ref{def:lstd} is not triggered. Writing $(\bbM^{(k)}+\lambda\bbI)(\hat\bbw_i^{(k),\rm raw}-\bbw^\ast)=\bbzeta_i^{(k)}-\lambda\bbw^\ast$ and bounding,
	\feq{
		\norm{\hat\bbw_i^{(k),\rm raw}-\bbw^\ast}
		&\leq\frac{2}{\mu_M}\bigl(\norm{\bbzeta_i^{(k)}}+\lambda\norm{\bbw^\ast}\bigr)\notag\\
		&\leq\frac{2C_v B_Y}{\mu_M}\sqrt{\frac{u}{M_s}}+\frac{2\lambda W^\ast}{\mu_M},
	}
	which is~\eqref{eq:lstd_error}. The uniform version follows by replacing $\delta$ with $\delta/(nK)$ and union-bounding over $i\in[n]$, $k<K$.
	\end{proof}

	\section{Proof of Theorem~\ref{thm:joint_optimal} (critic error)}\label{app:joint_optimal}

	\begin{proof}
	Fix $(i,k)$ and work on $E_M\cap E_{\rm rff}$; write $\pi=\pi_{\bbtheta^{(k)}}$, $\nu=\nu^\pi$.

	\emph{(a) Favorable event.} With $\lambda^\ast=0$, Theorem~\ref{thm:lstd_error} at confidence $\delta_\star$ requires $M_s\geq C_0L^4\ell_\star/\mu_M^2$ and gives, with probability $\geq1-\delta_\star$ over $\ccalD_s^{(k)}$,
	\feq{
		\norm{\hat\bbw_i^{(k)}-\bbw_i^\ast}\leq\norm{\hat\bbw_i^{(k),\rm raw}-\bbw_i^\ast}\leq\frac{2C_vB_Y}{\mu_M}\sqrt{\frac{\ell_\star}{M_s}},
	}
	where the first inequality is nonexpansiveness of the norm projection onto the ball of radius $W^\ast\geq\norm{\bbw_i^\ast}$ (Lemma~\ref{lem:bounded_target}). Since $\abs{\hat Q_i^{(k)}-q_i^\ast}=\abs{\inner{\tdbbphi_{i,\kappa}^{\rm raw}}{\hat\bbw_i^{(k)}-\bbw_i^\ast}}\leq L\norm{\hat\bbw_i^{(k)}-\bbw_i^\ast}$ pointwise, the second requirement in~\eqref{eq:Ms_choice}, $M_s\geq C_sL^2B_Y^2\ell_\star/(\mu_M^2\epsilon_Q^2)$ with $C_s\geq4C_v^2=8$, ensures $\sup\abs{\hat Q_i^{(k)}-q_i^\ast}\leq\epsilon_Q$; the choice $C_s=512\geq\max\{C_0,8\}$ covers both requirements. Then, by $(a+b)^2\leq2a^2+2b^2$ and Lemma~\ref{lem:pop_L2}(c),
	\feq{
		\norm{\hat Q_i^{(k)}-Q_i^\pi}_{\ccalL^2(\nu)}^2&\leq2\norm{\hat Q_i^{(k)}-q_i^\ast}_{\ccalL^2(\nu)}^2+2\norm{q_i^\ast-Q_i^\pi}_{\ccalL^2(\nu)}^2\notag\\
		&\leq2\epsilon_Q^2+\frac{2\Bk^2}{1-\gamma},
	}
	and, by Lemma~\ref{lem:pop_L2}(d), $\sup\abs{\hat Q_i^{(k)}-Q_i^\pi}\leq\epsilon_Q+(1+C')\Bk$. This is~\eqref{eq:joint_optimal_bound}.

	\emph{(b) Expectation over the critic sample.} Let $\ccalE$ denote the favorable event of (a), with $\mbP(\ccalE^c\mid\ccalF_k^-)\leq\delta_\star$. On $\ccalE^c$ the deterministic projection bound $\abs{\hat Q_i^{(k)}}\leq LW^\ast$ together with $\abs{Q_i^\pi}\leq Q_{\max}$ gives $\abs{\hat Q_i^{(k)}-Q_i^\pi}\leq A_Q$ pointwise, hence $\norm{\hat Q_i^{(k)}-Q_i^\pi}_{\ccalL^2(\nu)}^2\leq A_Q^2$. Therefore
	\feq{
		\mbE_{\rm alg}\norm{\hat Q_i^{(k)}-Q_i^\pi}_{\ccalL^2(\nu)}^2&\leq2\epsilon_Q^2+\frac{2\Bk^2}{1-\gamma}+\delta_\star A_Q^2\notag\\
		&\leq3\epsilon_Q^2+\frac{2\Bk^2}{1-\gamma},
	}
	using $\delta_\star A_Q^2\leq\epsilon_Q^2$. This is~\eqref{eq:critic_L2}. No confinement of the occupancy is used: the off-event contribution is controlled by the deterministic projection bound alone.

	\emph{Positive ridge.} If $\lambda\in(0,\mu_M\epsilon_Q/(4LW^\ast)]$ is used, the additional term $2\lambda W^\ast/\mu_M$ in~\eqref{eq:lstd_error}, multiplied by $L$, contributes at most $\epsilon_Q/2$ to $\sup\abs{\hat Q_i^{(k)}-q_i^\ast}$; the conclusions hold with $\epsilon_Q$ replaced by $\tfrac32\epsilon_Q$. The total unconditional failure probability of (a) is at most $\delta_M+\delta_{\rm rff}+\delta_\star$, and the uniform version replaces $\delta_\star$ by $\delta_\star/(nK)$.
	\end{proof}

	\section{Proof of Lemma~\ref{lem:grad_bias} (gradient bias)}\label{app:grad_bias}

	\begin{proof}
	Work conditional on $\ccalF_k^Q$, so the critics $\{\hat\bbw_\ell^{(k)}\}$ are fixed and $\ccalD_g^{(k)}$ is i.i.d.\ from $d^{\pi^{(k)}}\otimes\pi^{(k)}$; then take the expectation over the critic randomness. The bias of $\hat\bbg_i^{(k)}$ relative to $\nabla_{\bbtheta_i}J(\bbtheta^{(k)})$ decomposes into two sources: (i) the critic-approximation error within the $\kappa_c$-neighborhood, and (ii) the aggregation error from truncating the global value sum $\sum_\ell Q_\ell^\pi$ to the neighborhood sum $\sum_{\ell\in\ccalN_i^{\kappa_c}}\hat Q_\ell^{(k)}$.

	\emph{Critic-approximation route.} The per-agent policy-gradient identity~\eqref{eq:pg_theorem_agent} and the estimator~\eqref{eq:gradient_estimator} differ, on the $\kappa_c$-neighborhood, by the score-weighted critic error $\tfrac{1}{1-\gamma}\mbE_{\nu^\pi}[\sum_{\ell\in\ccalN_i^{\kappa_c}}(\hat Q_\ell^{(k)}-Q_\ell^\pi)\nabla_{\bbtheta_i}\log\pi_i]$. Using $\norm{\nabla_{\bbtheta_i}\log\pi_i}\leq G$, $\abs{\ccalN_i^{\kappa_c}}\leq D_{\kappa_c}$, Jensen ($\mbE\abs{\cdot}\leq(\mbE\abs{\cdot}^2)^{1/2}$ applied to $\mbE_{\rm alg}\mbE_{\nu^\pi}$), and the $\ccalL^2$-critic bound~\eqref{eq:critic_L2},
	\feq{
		\text{(i)}&\leq\frac{G}{1-\gamma}\sum_{\ell\in\ccalN_i^{\kappa_c}}\Bigl(\mbE_{\rm alg}\norm{\hat Q_\ell^{(k)}-Q_\ell^\pi}_{\ccalL^2(\nu^\pi)}^2\Bigr)^{1/2}\notag\\
		&\leq\frac{GD_{\kappa_c}}{1-\gamma}\Bigl(\sqrt 3\,\epsilon_Q+\sqrt{\tfrac{2}{1-\gamma}}\,\Bk\Bigr),
	}
	using $\sqrt{a+b}\leq\sqrt a+\sqrt b$.

	\emph{Aggregation route.} The truncated value sum omits the agents $\ell\notin\ccalN_i^{\kappa_c}$, so route (ii) equals $\tfrac{1}{1-\gamma}\norm{\mbE_{\nu^\pi}[\sum_{\ell\notin\ccalN_i^{\kappa_c}}Q_\ell^\pi\,\nabla_{\bbtheta_i}\log\pi_i]}$. We first establish a \emph{single-coordinate score-decoupling bound}. Fix $\ell$ with $d(i,\ell)>\kappa_c$. Under $\nu^\pi=d^\pi\otimes\pi$ the joint policy factorizes, so conditional on $\bbS$ the actions are independent across agents and the score satisfies $\mbE[\nabla_{\bbtheta_i}\log\pi_i(\bbA_i\mid\bbS_{\ccalN_i^{\kappa_\pi}})\mid\bbS,\bbA_{-i}]=\mbE_{\bbA_i\sim\pi_i}[\nabla_{\bbtheta_i}\log\pi_i]=\bbzero$, the score identity. Hence, for any function $\tilde Q_\ell$ of $(\bbS,\bbA)$ that does not depend on the single coordinate $\bbA_i$,
	\feq{
		\norm{\mbE_{\nu^\pi}[Q_\ell^\pi\,\nabla_{\bbtheta_i}\log\pi_i]}&=\norm{\mbE_{\nu^\pi}[(Q_\ell^\pi-\tilde Q_\ell)\,\nabla_{\bbtheta_i}\log\pi_i]}\notag\\
		&\leq G\,\sup\abs{Q_\ell^\pi-\tilde Q_\ell}.
	}
	Take $\tilde Q_\ell$ to be $Q_\ell^\pi$ with the coordinate $\bba_i$ frozen at an arbitrary fixed value. The tuple $(\bbS,\bbA)$ and its $\bba_i$-frozen counterpart agree on $\ccalN_\ell^{\,d(i,\ell)-1}$, since agent $i$ lies outside that neighborhood; the exponential decay (B1) at radius $d(i,\ell)-1\geq\kappa_c\geq1$ therefore gives $\sup\abs{Q_\ell^\pi-\tilde Q_\ell}\leq c\rho^{\,d(i,\ell)}$. Thus $\norm{\mbE_{\nu^\pi}[Q_\ell^\pi\,\nabla_{\bbtheta_i}\log\pi_i]}\leq Gc\,\rho^{\,d(i,\ell)}$, with no $\kappa_\pi$-dependent shift in the exponent: the zero-mean cancellation is taken over the single coordinate $\bbA_i$, not over the score-support neighborhood $\ccalN_i^{\kappa_\pi}$. (An omitted agent in another connected component contributes zero: apply (B1) at every finite radius, which is the convention $\rho^{\infty}:=0$ of Section~\ref{sec:pre}.) Summing over the omitted agents, with at most $\Delta(\Delta-1)^{r-1}$ agents at graph distance $r$ and $\tilde\delta:=(\Delta-1)\rho<1$,
	\feq{
		\text{(ii)}&\leq\frac{G}{1-\gamma}\sum_{\ell:\,d(i,\ell)>\kappa_c}c\,\rho^{\,d(i,\ell)}\leq\frac{Gc}{1-\gamma}\sum_{r>\kappa_c}\Delta(\Delta-1)^{r-1}\rho^{r}\notag\\
		&=\frac{Gc}{(1-\gamma)(1-\tilde\delta)}\cdot\frac{\Delta}{\Delta-1}\,\tilde\delta^{\kappa_c+1}=:\epsilon_{\rm agg}(\kappa),
	}
	which is~\eqref{eq:eps_agg_def}. Summing (i) and (ii) gives~\eqref{eq:grad_bias_bound}.
	\end{proof}

	\section{Proof of Lemma~\ref{lem:bias_squared}}\label{app:bias_squared}

	\begin{proof}
	We bound $\mbE[\norm{\bbd_i^{(k)}}^2\mid\ccalF_k^-]$, the history-conditional expectation over the critic dataset of the squared $\ccalF_k^Q$-conditional bias $\bbd_i^{(k)}=\mbE[\hat\bbg_i^{(k)}\mid\ccalF_k^Q]-\nabla_{\bbtheta_i}J(\bbtheta^{(k)})$. We do not square the fully-averaged bound of Lemma~\ref{lem:grad_bias} directly; instead we use the $\hat Q$-conditional decomposition that underlies it. Conditional on $\ccalF_k^Q$ the critics $\{\hat\bbw_\ell^{(k)}\}$ are fixed, and the two routes of Appendix~\ref{app:grad_bias} give
	\feq{
		&\norm{\mbE[\hat\bbg_i^{(k)}\mid\hat Q,\bbtheta^{(k)}]-\nabla_{\bbtheta_i}J}\notag\\
		&\quad\leq\frac{G}{1-\gamma}\sum_{\ell\in\ccalN_i^{\kappa_c}}\mbE_{\nu^\pi}\bigl[\abs{\hat Q_\ell^{(k)}-Q_\ell^\pi}\,\big|\,\hat Q\bigr]+\epsilon_{\rm agg}(\kappa),
	}
	the first term being the critic-approximation route (score bounded by $G$, neighborhood size $\abs{\ccalN_i^{\kappa_c}}\leq D_{\kappa_c}$, triangle inequality) and the second the aggregation route, which is deterministic. The bound holds $\hat Q$-conditionally because, given $\ccalF_k^Q$, the gradient sample $\ccalD_g^{(k)}$ is i.i.d.\ from $\nu^\pi$ and each $\hat Q_\ell^{(k)}$ is a fixed function. Applying $(a+b)^2\leq 2a^2+2b^2$ and then the Cauchy--Schwarz inequality over the at most $D_{\kappa_c}$ summands,
	\feq{
		&\mbE_{\hat Q}\norm{\mbE[\hat\bbg_i^{(k)}\mid\hat Q,\bbtheta^{(k)}]-\nabla_{\bbtheta_i}J}^2\notag\\
		&\quad\leq 2\epsilon_{\rm agg}(\kappa)^2+\frac{2G^2 D_{\kappa_c}}{(1-\gamma)^2}\sum_{\ell\in\ccalN_i^{\kappa_c}}\mbE_{\hat Q}\,\mbE_{\nu^\pi}\bigl[\abs{\hat Q_\ell^{(k)}-Q_\ell^\pi}\,\big|\,\hat Q\bigr]^2.
	}
	By conditional Jensen, $\mbE_{\nu^\pi}[\abs{\hat Q_\ell^{(k)}-Q_\ell^\pi}\mid\hat Q]^2\leq\mbE_{\nu^\pi}[\abs{\hat Q_\ell^{(k)}-Q_\ell^\pi}^2\mid\hat Q]$; taking $\mbE_{\hat Q}$ and invoking~\eqref{eq:critic_L2} for each $\ell$, the sum over the at most $D_{\kappa_c}$ neighbors is at most $D_{\kappa_c}(3\epsilon_Q^2+2\Bk^2/(1-\gamma))$. Substituting,
	\feq{
		\mbE\bigl[\norm{\bbd_i^{(k)}}^2\,\big|\,\ccalF_k^-\bigr]&\leq 2\epsilon_{\rm agg}(\kappa)^2\\
		&\quad+\frac{2G^2 D_{\kappa_c}^2}{(1-\gamma)^2}\Bigl(3\epsilon_Q^2+\frac{2\Bk^2}{1-\gamma}\Bigr),
	}
	which is~\eqref{eq:eps_b_bound} with $A_\kappa=G^2D_{\kappa_c}^2/(1-\gamma)^2$.
	\end{proof}

	\section{Proof of Theorem~\ref{thm:main_convergence} (convergence)}\label{app:main_convergence}

	\begin{proof}
	By (D2), $J$ is $L_J$-smooth on $\Theta_0$. Write $\hat\bbg^{(k)}=\nabla J(\bbtheta^{(k)})+\bbe^{(k)}$ with total error $\bbe^{(k)}:=\bbd^{(k)}+\bbn^{(k)}$, where $\bbd^{(k)}:=\mbE[\hat\bbg^{(k)}\mid\ccalF_k^Q]-\nabla J(\bbtheta^{(k)})$ is the conditional bias and $\bbn^{(k)}$ the conditionally mean-zero noise, and define the stochastic projected-gradient mapping $\tilde\ccalG_\eta^{(k)}:=\eta^{-1}(\Pi_{\Theta_0}(\bbtheta^{(k)}+\eta\hat\bbg^{(k)})-\bbtheta^{(k)})$, so that $\bbtheta^{(k+1)}=\bbtheta^{(k)}+\eta\tilde\ccalG_\eta^{(k)}$.

	\emph{Why the classical noise accounting fails under projection.} An inequality of the form $\mbE J(\bbtheta^{(k+1)})\geq\mbE J(\bbtheta^{(k)})+c\,\eta\,\mbE\norm{\ccalG_\eta(\bbtheta^{(k)})}^2-C\eta^2 L_J\,\mbE\norm{\bbn^{(k)}}^2$ does not hold in general for projected stochastic updates: mean-zero noise can produce an expected objective decrease of order $\Theta(\eta)$ at a point where $\ccalG_\eta=\bbzero$. A one-dimensional example shows the obstruction. Let $\Theta_0=[0,1]$, $J(x)=-x-x^2/2$ (so $L_J=1$), $x=0$, $0<\eta\leq1$, and noise $\xi=+2$ with probability $1/3$ and $\xi=-1$ with probability $2/3$; then $\ccalG_\eta(0)=0$ yet $\mbE[J(x^+)]-J(0)=-\eta/3-\eta^2/6$. We therefore argue through the stochastic mapping, following the composite-optimization route of Ghadimi, Lan, and Zhang~\cite{ghadimi2016mini} (see~\cite{ghadimi2013stochastic} for the unconstrained smooth counterpart).

	\emph{Step 1 (pathwise one-step inequality).} The projection optimality condition for $\bbtheta^{(k+1)}=\Pi_{\Theta_0}(\bbtheta^{(k)}+\eta\hat\bbg^{(k)})$, tested at the feasible point $\bbtheta^{(k)}$, (using that $\Theta_0$ is closed and convex, as each $\Theta_{0,i}$ is compact convex by (A4)) reads $\inner{\bbtheta^{(k)}+\eta\hat\bbg^{(k)}-\bbtheta^{(k+1)}}{\bbtheta^{(k)}-\bbtheta^{(k+1)}}\leq 0$, i.e.\ $\inner{\hat\bbg^{(k)}}{\tilde\ccalG_\eta^{(k)}}\geq\norm{\tilde\ccalG_\eta^{(k)}}^2$. Combining with the $L_J$-smoothness lower bound and Young's inequality $\inner{\bbe^{(k)}}{\tilde\ccalG_\eta^{(k)}}\leq\tfrac12\norm{\bbe^{(k)}}^2+\tfrac12\norm{\tilde\ccalG_\eta^{(k)}}^2$,
	\feq{
		J(\bbtheta^{(k+1)})&\geq J(\bbtheta^{(k)})+\eta\inner{\nabla J(\bbtheta^{(k)})}{\tilde\ccalG_\eta^{(k)}}-\tfrac{L_J\eta^2}{2}\norm{\tilde\ccalG_\eta^{(k)}}^2\notag\\
		&\geq J(\bbtheta^{(k)})+\tfrac{\eta}{2}\bigl(1-L_J\eta\bigr)\norm{\tilde\ccalG_\eta^{(k)}}^2-\tfrac{\eta}{2}\norm{\bbe^{(k)}}^2,
	}
	and with $\eta\leq 1/(4L_J)$ the coefficient of $\norm{\tilde\ccalG_\eta^{(k)}}^2$ is at least $3\eta/8$.

	\emph{Step 2 (from the stochastic to the true mapping).} Non-expansiveness of $\Pi_{\Theta_0}$ gives $\norm{\ccalG_\eta(\bbtheta^{(k)})-\tilde\ccalG_\eta^{(k)}}\leq\norm{\bbe^{(k)}}$, hence $\norm{\tilde\ccalG_\eta^{(k)}}^2\geq\tfrac12\norm{\ccalG_\eta(\bbtheta^{(k)})}^2-\norm{\bbe^{(k)}}^2$. Substituting into Step~1 yields the pathwise inequality
	\feq{
		J(\bbtheta^{(k+1)})\geq J(\bbtheta^{(k)})+\tfrac{3\eta}{16}\norm{\ccalG_\eta(\bbtheta^{(k)})}^2-\tfrac{7\eta}{8}\norm{\bbe^{(k)}}^2,\label{eq:pathwise_descent}
	}
	in which the error term is proportional to $\eta$. This inequality holds on every sample path before any expectation is taken; it is exactly the inequality that the predictable-prefix argument of Appendix~\ref{app:sample_complexity_traj} multiplies by the indicator $\bbone_{A_0\cap\cdots\cap A_k}\bbone_{E_{\rm rff}}$.

	\emph{Step 3 (bias--variance and aggregation).} Taking expectations and using that $\bbd^{(k)}$ is $\ccalF_k^Q$-measurable while $\mbE[\bbn^{(k)}\mid\ccalF_k^Q]=\bbzero$ gives $\mbE\norm{\bbe^{(k)}}^2=\mbE\norm{\bbd^{(k)}}^2+\mbE\norm{\bbn^{(k)}}^2$. All norms are aggregate quantities; under the product structure (A4), $\norm{\bbv}^2=\sum_{i=1}^n\norm{\bbv_i}^2$ and the projection acts blockwise, so the analyzed update equals the global projected step. Since $\bbd^{(k)}=(\bbd_i^{(k)})_{i=1}^n$ stacks the per-agent biases, $\mbE[\norm{\bbd^{(k)}}^2\mid\ccalF_k^-]=\sum_i\mbE[\norm{\bbd_i^{(k)}}^2\mid\ccalF_k^-]\leq n\,\bar\epsilon_b^2$ by Lemma~\ref{lem:bias_squared}; shared samples across agents do not affect this step, since the squared norm of the stacked vector is the sum of the squared block norms and no independence between agents is needed. Since $\bbn^{(k)}$ stacks the per-agent mean-zero noises, $\mbE\norm{\bbn^{(k)}}^2\leq n\,\sigma_g^2/M_g$, where $\sigma_g^2=G^2D_{\kappa_c}^2(LW^\ast)^2/(1-\gamma)^2$ bounds the per-agent per-sample gradient-noise variance (the per-sample gradient magnitude is at most $GD_{\kappa_c}LW^\ast/(1-\gamma)$ under the projected linear critic) and the division by $M_g$ reflects the average over $M_g$ i.i.d.\ samples.

	\emph{Step 4 (telescoping).} Summing~\eqref{eq:pathwise_descent} over $k=0,\ldots,K-1$, using $\Delta_J=J^\ast-J(\bbtheta^{(0)})$, dividing by $3\eta K/16$ and then by $n$: the explicit factor $n$ in the error terms cancels against the division by $n$, while the descent term retains the factor $1/n$. This yields~\eqref{eq:convergence} with the constants $16/3$ and $14/3=(7/8)(16/3)$. The per-agent normalization is a consequence of the product structure, not an assumed scaling.
	\end{proof}

	\section{Proof of Theorem~\ref{thm:sample_complexity_fixed}}\label{app:sample_complexity}

	\begin{proof}
	Substitute the parameter choices into~\eqref{eq:explicit_bound}. With $K\geq64L_J\Delta_J/(n\epsilon)$ the descent term $64L_J\Delta_J/(3nK)$ is at most $\epsilon/3$; with $M_g\geq42\sigma_g^2/\epsilon$ the gradient-noise term $14\sigma_g^2/(3M_g)$ is at most $\epsilon/9\leq\epsilon/3$; with $\epsilon_Q^2\leq\epsilon/(84A_\kappa)$ the statistical critic term $28A_\kappa\epsilon_Q^2$ is at most $\epsilon/3$. The choice $\delta_\star=\min\{1/2,\epsilon_Q^2/A_Q^2\}$ and the sample size~\eqref{eq:Ms_choice} at $\ell_\star=\log(4(m+1)/\delta_\star)$ make Theorem~\ref{thm:joint_optimal}(b), and hence Lemma~\ref{lem:bias_squared}, applicable at every $(i,k)$; the per-iteration LSTD failure events are thereby absorbed into the term $3\epsilon_Q^2$ in expectation, and no union bound over iterations is required. The remaining terms of~\eqref{eq:explicit_bound} are $(28/3)\epsilon_{\rm agg}(\kappa)^2\leq\Cfloor\Efloor^{\rm graph}(\kappa)$ and $\tfrac{112A_\kappa}{3(1-\gamma)}(B_{\rm loc}^2+B_{\rm rff}^2)=\Cfloor[\Efloor^{\rm trunc}(\kappa)+\Efloor^{\rm rff}(\kappa,m)]$ with $\Cfloor=112/3$. Hence
	\feq{
		\frac{1}{nK}\sum_{k<K}\mbE\norm{\ccalG_\eta(\bbtheta^{(k)})}^2\leq\epsilon+\Efloor(\kappa,m).
	}
	For the complexity count, $\Delta_J\leq2nQ_{\max}$ gives $K=\ccalO(L_JQ_{\max}/\epsilon)$; $B_Y\leq4L^2W^\ast$ (since $LW^\ast=L^2\bbarr/\mu_M\geq\bbarr/2$ by $\mu_M\leq2L^2$) and $\epsilon_Q^2=\Theta(\epsilon/A_\kappa)$ give $M_s=\widetilde\ccalO(L^6W^{\ast2}A_\kappa/(\mu_M^2\epsilon))$ with $\ell_\star=\ccalO(\log(mA_Q^2A_\kappa/\epsilon))$; and $M_g=\ccalO(\sigma_g^2/\epsilon)$. Thus $N_{\rm global}=K(M_s+M_g)=\widetilde\ccalO(1/\epsilon^2)$ with the displayed prefactor.
	\end{proof}

	\section{Proof of Theorem~\ref{thm:sample_complexity_traj}}\label{app:sample_complexity_traj}

	\begin{proof}
	The only structural difference from Appendix~\ref{app:sample_complexity} is that the TD-stability lower bound $\mu_M$ is replaced by $\mu_{\rm traj}$ and holds only on the trajectory event $E_{\rm traj}=\bigcap_{k<K}A_k$, with each $A_k=\{\inf_i\sigma_{\min}(\bbM_i^{\bbtheta^{(k)}})\geq\mu_{\rm traj}\}$ being $\ccalF_k^-$-measurable (predictable).

	\emph{Predictable-event factorization.} The per-step descent inequality~\eqref{eq:pathwise_descent}, established pathwise before expectations, requires the critic bound only at the current iterate $\bbtheta^{(k)}$, which is controlled by $A_k$. Multiply~\eqref{eq:pathwise_descent} by the predictable indicator $\bbone_{A_0\cap\cdots\cap A_k}\bbone_{E_{\rm rff}}\in\ccalF_k^-$ before taking expectations. Because the indicator is $\ccalF_k^-$-measurable, it commutes with the conditional expectation in the concentration step (Theorem~\ref{thm:lstd_error} applied conditional on $\ccalF_k^-$), so the conditional matrix-Bernstein and Hilbert-space vector-concentration bounds, Lemma~\ref{lem:pop_L2} at threshold $\mu_{\rm traj}$, and Theorem~\ref{thm:joint_optimal} hold on $A_k\cap E_{\rm rff}$ with $\mu_M$ replaced by $\mu_{\rm traj}$ and $E_M$ by $A_k$, because their proofs use the conditioning bound only at the current iterate. No conditioning on the global future event $E_{\rm traj}$ is used.

	\emph{Stopped-process telescoping.} Define the stopping time $\tau:=\inf\{k:A_k^c\}\wedge K$, which is a valid stopping time for the filtration $(\ccalF_k^-)$ since $\{\tau>k\}=A_0\cap\cdots\cap A_k\in\ccalF_k^-$. Telescoping the predictable-prefix-weighted per-step inequality over $k=0,\ldots,K-1$ is equivalent to telescoping the descent inequality over the stopped process $\{\bbtheta^{(k)}\}_{k<\tau}$ on $E_{\rm rff}$; the telescoped objective increments satisfy the pathwise identity $\sum_{k<\tau}[J(\bbtheta^{(k+1)})-J(\bbtheta^{(k)})]=J(\bbtheta^{(\tau)})-J(\bbtheta^{(0)})\leq\Delta_J$, which is the only property of the stopped sum the argument uses; no sign condition on any boundary term is needed. This yields the unconditional predictable-prefix bound
	\feq{
		\mbE\Bigl[\tfrac{1}{nK}\sum_{k=0}^{K-1}\norm{\ccalG_\eta(\bbtheta^{(k)})}^2\,\bbone_{A_0\cap\cdots\cap A_k}\,\bbone_{E_{\rm rff}}\Bigr]\leq\epsilon+\Efloor(\kappa,m),
	}
	which is statement~(a); the equivalent stopped-process form follows from $\bbone_{A_0\cap\cdots\cap A_k}=\bbone_{\{\tau>k\}}$. Statement~(b) follows because $E_{\rm traj}\subseteq A_0\cap\cdots\cap A_k$ for every $k$, whence $\bbone_{E_{\rm traj}}\leq\bbone_{A_0\cap\cdots\cap A_k}$. Statement~(c) is Assumption~\ref{asm:Cp} together with $\mbP(E_{\rm rff})\geq 1-\delta_{\rm rff}$ and a union bound, and statement~(d) follows by dividing the bound of statement~(b) by $\mbP(E_{\rm traj}\cap E_{\rm rff})\geq 1/2$. The complexity count $N_{\rm global}^{\rm traj}=K(M_s+M_g)$ carries the factor $L_J$ in both terms through $K=\ccalO(1+L_J\Delta_J/(n\epsilon))$, with all derived constants ($W^\ast,A_Q,B_Y,\sigma_g^2,\Efloor$) recomputed using $\mu_{\rm traj}$ in place of $\mu_M$.
	\end{proof}

	\section{Proof of Proposition~\ref{prop:adaptive} (offline radius selection)}\label{app:adaptive}

	\begin{proof}
	Under the strengthened hypothesis $\Delta\rho<1$, the graph-tail base satisfies $\tilde\delta=(\Delta-1)\rho<\Delta\rho<1$. By the growth bound $D_{\kappa_c}\leq C_D\Delta^{\kappa+\kappa_\pi}$,
	\feq{
		\Efloor^{\rm trunc}(\kappa)&=\frac{G^2D_{\kappa_c}^2\gamma^2\tilde c^2\rho^{2(\kappa+1)}}{(1-\gamma)^3}\notag\\
		&\leq\frac{G^2C_D^2\Delta^{2\kappa_\pi}\gamma^2\tilde c^2\rho^2}{(1-\gamma)^3}\,(\Delta\rho)^{2\kappa}=C_t(\Delta\rho)^{2\kappa},
	}
	and, since $\tilde\delta\leq\Delta\rho$ and $\kappa_c+1\geq\kappa$,
	\feq{
		\Efloor^{\rm graph}(\kappa)=\epsilon_{\rm agg}(\kappa)^2=C_g\,\tilde\delta^{2(\kappa_c+1)}\leq C_g(\Delta\rho)^{2\kappa}.
	}
	The third term in~\eqref{eq:kappa_star} guarantees $(\Delta\rho)^{2\kappa^\ast}\leq\epsilon/(\Cfloor\max\{C_t,C_g\})$, hence $\Cfloor\Efloor^{\rm trunc}(\kappa^\ast)\leq\epsilon$ and $\Cfloor\Efloor^{\rm graph}(\kappa^\ast)\leq\epsilon$; the first two terms enforce the standing restriction $\kappa^\ast\geq\max\{1,\kappa_\pi+1\}$ required by Lemma~\ref{lem:V_decay} and Theorem~\ref{thm:q_linear}. Under the stated hypotheses (RFF feasibility and Assumption~\ref{asm:C} at $\kappa^\ast$), Theorem~\ref{thm:sample_complexity_fixed} applies at radius $\kappa^\ast$ and gives $\epsilon+\Efloor(\kappa^\ast,m)\leq3\epsilon+\Cfloor\Efloor^{\rm rff}(\kappa^\ast,m)$, which is the stated bound. The envelope for $\Efloor^{\rm rff}(\kappa^\ast,m)$ follows by substituting $\epsilon_P(m)^2=\widetilde\ccalO(C_{\rm RFF}(D_{\kappa^\ast}d_S)\tilde g_\alpha(\kappa^\ast)^2/m)$ from Theorem~\ref{thm:rff_approx} into $\Efloor^{\rm rff}=A_{\kappa^\ast}\gamma^2\bbarr^2\epsilon_P^2/(1-\gamma)^3$. Since $\tilde g_\alpha(\kappa)=\exp(\ccalO(D_\kappa d_S))$ and $\kappa^\ast=\ccalO(\log(1/\epsilon))$, substituting the relevant upper bound on $D_{\kappa^\ast}$ ($\ccalO(\kappa^\ast)$, $\ccalO(\kappa^{\ast p})$, or $\ccalO(\Delta^{\kappa^\ast})$) yields the polynomial, quasi-polynomial, and exponential-in-a-power envelopes of the statement, respectively; in the exponential-growth case $\Delta^{\kappa^\ast}=\ccalO((1/\epsilon)^{\log\Delta/(2\log(1/(\Delta\rho)))})$, which gives the stated exponent $\chi$. When the realized $D_{\kappa_c^\ast}$ reaches $n$, the cap $D_{\kappa_c^\ast}\leq n$ contributes the explicit $n^2$ factor of Remark~\ref{rem:n_scaling}.
	\end{proof}

	\bibliographystyle{IEEEtran}
	\bibliography{reference}

\end{document}